\documentclass[11pt]{article}
\input{preamble.tex}

\excludeversion{mylongform}

\excludeversion{myscribbles}

\includeversion{mysolution}

 \usepackage{authblk}
 \title{Bandwidth selection for kernel density estimators of multivariate level sets and highest density regions} 
 \author[1]{Charles R. Doss\thanks{Email: \href{mailto:cdoss@umn.edu}{cdoss@umn.edu}. Supported in part by NSF Grant DMS-1712664}}
 \author[1]{Guangwei Weng\thanks{Email:
     \href{mailto:wengx076@umn.edu}{wengx076@umn.edu}. Supported in part by a University of Minnesota Grant-in-Aid grant}}
 \affil[1]{School of Statistics\protect\\ University of
   Minnesota\protect \\
   Minneapolis,  MN 55455}

\begin{document}

\maketitle
\begin{abstract}
    We consider bandwidth matrix selection for kernel density estimators (KDEs)
    of density level sets in $\RR^d$, $d \ge 2$.  We also consider estimation
    of highest density regions, which differs from estimating level sets in
    that one specifies the probability content of the set rather than
    specifying the level directly. This complicates the problem.  Bandwidth
    selection for KDEs is well studied, but the goal of most methods is to
    minimize a global loss function for the density or its derivatives.  The
    loss we consider here is instead the measure of the symmetric difference of
    the true set and estimated set.  We derive an asymptotic approximation to
    the corresponding risk.  The approximation depends on unknown quantities
    which can be estimated, and the approximation can then be minimized to
    yield a choice of bandwidth, which we show in simulations performs well.
    We provide an \proglang{R} package \pkg{lsbs} for implementing our procedure.
\end{abstract}

{\hypersetup{linkcolor=black}
  \tableofcontents
  }

\newpage

\section{Introduction}

As computing power has become greater and as data sets have become simultaneously larger and more complicated, demand for statistical methods that are increasingly flexible and data driven has increased.  Two related methods for capturing the complex structure of a data set from a true density $f_0$ are to estimate either the density's {\em level sets} (LS's)  or the density's {\em highest-density regions} (HDR's).  (We will explain the difference between estimating LS's and estimating HDR's shortly.)
For a density function $f_0$ defined on
$\bb{R}^d$ and a given constant $c> 0$, the $c$-level set (sometimes known as a {\em density contour})  of $f_0$ is
$  \beta(c):=\{\bx\in\bb{R}^d:f_0(\bx)= c\},$
and the corresponding super-level set is
\begin{align}
  \label{eq:suplevel-t}
  \mc{L}(c):=\{\bx\in\bb{R}^d:f_0(\bx)\ge c\}.
\end{align}
Under some basic regularity conditions, the density super-level set is a set of minimum volume having $f_0$-probability at least $\int_{ \mc{L}(c)}f_0(\bx)\,d\bx$ \citep{garcia2003level}.  For this reason, perhaps the most common use for HDR estimation occurs in Bayesian statistics.  An HDR of a posterior density is a so-called (minimum volume) credible region, which is one of the most fundamental tools in Bayesian statistics.  There are quite a wide range of other applications for estimation of density LS's or density HDR's and these estimation problems have received increasing attention in the statistics and machine learning literatures in recent years.  (We consider estimation of density level sets and estimation of density super-level sets to be equivalent tasks.)  The applications of LS or HDR estimation include outlier/novelty detection \citep{Lichman:2014jn,Park2010:cx}, discriminant analysis \citep{MR1765618} and clustering analysis \citep{hartigan1975clustering,rinaldo2010generalized, cuevas2001cluster}.  LS estimation is one of the fundamental tools in estimation of cluster trees and persistence diagrams, used in topological data analysis  (\cite{Chen:2017wn}, \cite{Wasserman:2016ua}).

A common way to estimate the density super-level set $\mc{L}(c)$ based on independent and identically distributed (i.i.d.) $\bX_1,\ldots, \bX_n \in \RR^d$ is to replace the density function in
\eqref{eq:suplevel-t} with a kernel density estimator (KDE)
\begin{align}
  \label{eq:kde-def}
  \ffnH (\bx) :=
  \inv{n} \sum_{i=1}^n K( \bH^{-1/2}(\bx-\bX_i)) |\bH|^{-1/2},
\end{align}
where 
$\bH \in \RR^{d\times d}$ is a symmetric positive definite bandwidth matrix and $K$ is a kernel function.  This gives us the so-called plug-in estimator
\begin{align}
  \label{eq:plug-in-t}
  \widehat{\mc{L}}_{n,\bH}(c):=\{\bx\in\bb{R}^d:\ffnH(\bx)\ge c\}.
\end{align}
We now explain the difference between ``LS estimation'' and ``HDR estimation.'' Often the level of interest is only specified indirectly through a given probability $\tau\in(0,1)$ which  yields a level
$f_{\tau,0}:=\inf \{y > 0 :\int_{\bb{R}^d}f_0(\bx)\one_{\{f(\bx)\ge   y\}}\,d\bx\le 1-\tau\}$.
Then the corresponding super-level set is
\begin{align}
  \label{eq:suplevel-ft}
  \mc{L}(f_{\tau,0}):=\{\bx\in\bb{R}^d:f_0(\bx)\ge f_{\tau,0}\},
\end{align}
and the corresponding plug-in estimators are
$$\hat{f}_{\tau,n}:=\inf \lb y\in(0,\infty):\int_{\bb{R}^d}\ffnH(\bx)\one_{\{\ffnH(\bx)\ge
  y\}}\,d\bx\le 1-\tau \rb$$ and
\begin{align}
  \label{eq:plug-in-ft}
  \widehat{\mc{L}}_{n,\bH}(\hat{f}_{\tau,n}):=\{\bx\in\bb{R}^d:\ffnH(\bx)\ge \hat{f}_{\tau,n}\}.
\end{align}
Estimating \eqref{eq:suplevel-ft}
based on specifying $\tau$ is known as the {\em HDR estimation} problem; this has extra complication over the LS estimation problem because $f_{\tau,0}$ has to be estimated rather than being fixed in advance.  Thus we use the phrase {\em LS estimation} to mean estimation  of
\eqref{eq:suplevel-t} with $c$ fixed in advance (equivalently, estimation of
\eqref{eq:suplevel-ft} with $f_{\tau,0}$ fixed).  When we use the phrase {\em HDR estimation} we mean estimation of
\eqref{eq:suplevel-ft} with $\tau$ (but not $f_{\tau,0}$) fixed in advance.
Thus, LS's and HDR's are mathematically equivalent, but estimating LS's and estimating HDR's are statistically different tasks.

Early work on LS
or HDR estimation includes \cite{Hartigan:1987hx},
\cite{Muller:1991vn},
\cite{Polonik:1995kr},
\cite{Tsybakov:1997jv}, and
\cite{Walther:1997eu}.
Some recent work has focused on asymptotic properties of KDE plug-in estimators, including results about consistency, limit distribution theory, and statistical inference.
\citet{Baillo:2001fe} show that the probability content
of the plug-in estimator converges to the probability  of the true
super-level
set as the sample size tends to infinity.
\citet{Baillo:2003ds} proves the
strong consistency of the plug-in  estimator under an integrated
symmetric difference
metric.
\citet{Cadre:2006db}  further obtains the rate of
convergence of  the plug-in estimator when the loss is given by  the  generalized symmetric
difference  of sets.
\citet{Mason:2009dk} give the asymptotic normality of estimated super-level
sets under the same metric as \citet{Cadre:2006db}. \citet{Chen:2015uj} find a more practically usable
limiting distribution of the plug-in estimator for LS's by using Hausdorff
distance as the metric for set difference and provide methods for
constructing confidence regions for LS's based on this limiting distribution.
\cite{Jankowski:2012wv} and \cite{Mammen:2013hs} also investigate the formation of confidence regions for LS's.

It is well known that KDE's are sensitive to the choice of the bandwidth (matrix).  The optimal bandwidth (matrix) depends on the objective of estimation.  There are many tools that have been developed for selecting the bandwidth when $d=1$ or the bandwidth matrix when $d > 1$; these include minimizing an asymptotic approximation to an appropriate risk function, as well as computational methods such as the bootstrap or cross-validation, and are largely focused on globally estimating the density or its derivatives well. A good summary of those methods can be found in \citet{Wand:1995kv}, \citet{sain1994cross}, or \cite{Jones:1996vg}.

However, \citet[page 505]{Duong:2009ek} state that, ``a number of practical issues in highest density region estimation, such as good data-driven rules for choosing smoothing parameters, are yet to be resolved.''  \citet{Samworth:2010cj} is the only published work we know of that investigates the problem of selecting bandwidths for HDR estimation (and we know of no published works that directly investigate bandwidth selection for LS estimation).  \citet{Samworth:2010cj} study the KDE plug-in estimator when $d=1$, and show by simulation that the kernel density estimator aiming for HDR estimation can be very different from the one aiming for global density estimation. They also propose an asymptotic approximation to a risk function that is suitable for HDR estimation and a corresponding bandwidth selection procedure based on the approximation, all when $d=1$.

In this paper, we consider the multivariate setting, where $d\ge
2$. In this case, we are estimating a level set manifold, which
involves some added technical difficulties over the case $d=1$ (in
which case the level set is a finite point set), but we believe that
LS or HDR estimation when $d\ge 2$ is of great practical interest because of the large variety of complicated structures that multivariate level sets can reveal.  We derive asymptotic approximations to a risk function for LS estimation and to a risk function for HDR estimation.  We believe that our approximations and derivations will be very valuable for any future procedures that do (either) LS or HDR bandwidth selection.  Our calculations shed light on the important quantities relating to LS or HDR estimation.  Furthermore, we develop a ``plug-in'' bandwidth selector method based on minimizing an estimate of the LS or the HDR risk approximation.  This approach can be used to optimize over all positive definite bandwidth matrices or over restricted classes of matrices (e.g., diagonal ones).  Our theory applies for all $ d \ge 2$.  We have developed code to implement our bandwidth selector when $d=2$.  It is straightforward to implement a numeric approximation to Hausdorff integrals that appear in our approximations (see Subsection~\ref{subsec:notation-assumptions} for discussion of the Hausdorff measure) when $d=2$.  It is less immediately obvious how to implement such approximations when $d \ge 3$, although
we indeed believe that implementation is feasible for such approximations.
In fact, we believe that computational feasibility is an important benefit of using a closed-form approximation to the risk, particularly in the multivariate setting that we consider in this paper. As will be discussed later in the paper, many simple problems in the univariate setting are more complicated in the multivariate setting and must be solved by Monte Carlo.  Thus performing bootstrap or cross-validation, which involves nested Monte Carlo computations, quickly becomes infeasible.


During the development of the present paper we became aware of the recent
related work, \cite{Qiao:2017wq}.
\cite{Qiao:2017wq} also considers
problems about bandwidth selection for KDE's
in settings related to level set estimation.
However, the main focus of
\cite{Qiao:2017wq} is somewhat different than the one here.
In fact, \cite{Qiao:2017wq} states
that bandwidth selection for multivariate HDR estimation is ``far from
trivial'' and does not consider this problem.
We will discuss the  approach taken by \cite{Qiao:2017wq} again in the
\nameref{sec:conclusion} section.


The structure of the paper is as follows.
We present  our two asymptotic risk approximation theorems, as well as corollaries about the risk approximation minimizers, in Section~\ref{sec:asymptotic-risk-expansions}.
We present methodology to select bandwidth matrices  
in
Section~\ref{sec:methodology}.
In Section~\ref{sec:simulations-data}
we study the performance of our bandwidth selector in  simulation
experiments as well as in analysis of two real data sets, the Wisconsin Breast Cancer
Diagnostic data and the Banknote Authentication data.
We give concluding discussion in Section~\ref{sec:conclusion}.
Proofs of the main results are given in Appendix~\ref{app:A-main-proofs},
and further details, technical results, and intermediate lemmas are given
in Appendix~\ref{app:additional-thms}
and Appendix~\ref{sec:proofs-intermediate}.
Some notation and assumptions are presented in Subsections~\ref{subsec:notation-assumptions} and \ref{subsec:assumptions}.

\section{Asymptotic risk results}
\label{sec:asymptotic-risk-expansions}

\subsection{Notation}
\label{subsec:notation-assumptions}

We use the following notation throughout.  For a density function
$f_0$ on $\RR^d$ and a Borel measurable set $A\subset \RR^d$, define
the measure $\mu_{f_0}(A)=\int_Af_0(\bx)\,d\bx$.
For a function $f$ on
$\RR^d$, a measure $P$, and $1 \le p < \infty$, we let $\| f\|_{p,P}^p
= \int_{\RR^d} |f(\bs{z})|^p dP(\bs{z})$ if this quantity is finite.
If $P$ is Lebesgue measure we abbreviate $\| f \|_{p,P} \equiv \|
f\|_p$, $ 1 \le p < \infty$.  Let $\| f \|_{\infty} = \sup_{\bs{z} \in
  \RR^d} |f( \bs{z}) |$, and for a function $g$ with vector or matrix
values, that is, $g:\bb{R}^d\rightarrow \bb{R}^{p\times q}$, let
$\|g\|_{\infty}=\max_{1\le i\le p,1\le j\le q}\|g_{ij}\|_{\infty}$.
We let $\Vert {\bs x} \Vert = ( \sum_{i=1}^d x_i^2 )^{1/2}$ for $\bx
\in \RR^d$.  Let $\nabla f$ be the gradient (column) vector of $f$ and
let $\nabla^2 f$ be the Hessian matrix $\lp \derivtwo{x_i}{x_j}[f]
\rp_{i,j}$.  Let $\cH$ be $d-1$ dimensional Hausdorff measure
\citep{Evans:2015uy}.  The Hausdorff measure is  useful for measuring the volume of lower dimensional sets, like manifolds, embedded in a higher dimensional
ambient space.
Let $\lambda$  denote Lebesgue measure.
Recall that $\beta(c):=\{\bx\in\RR^d:f_0(\bx)=c\}$ and
$\mc{L}(c):=\{\bx\in\RR^d:f_0(\bx)\ge c\}$, we let $\mc{L}_{\tau}\equiv \mc{L}(f_{\tau,0})$ and
$\widehat{\mc{L}}_{\tau,\bH}\equiv\widehat{\mc{L}}_{\bH}(\widehat{f}_{\tau,n})$. We
generally use bold to denote vectors. We use ``$\equiv$'' to denote
notational equivalences and ``$:=$'' or ``$=:$'' for definitions.
Any integral whose domain is not specified explicitly is taken over
all of $\RR^d$.  We will occasionally omit the integrating
variable when there's no confusion in doing so. We use $\mc{S}$ to denote
the set of all $d\times d$ symmetric positive definite matrices. For a symmetric matrix $\bs{A}$, we use
$\lambda_{\max}(\bs{A})$ and $\lambda_{\min}(\bs{A})$ to denote the
largest and the smallest eigenvalues of $\bs{A}$ respectively. In this paper, we will use the $f_0$-probability volume of the symmetric difference
as the distance between the true  set and its estimator. We
use $\Delta$ to denote the symmetric difference operation between two
sets: for two sets $A$ and $B$, $A\Delta B :=(A\cup B)\setminus
(A\cap B)$ where ``$\setminus$'' is set difference. Figure~\ref{fig:symm-diff} shows  the
symmetric difference between the $0.02$ super-level set of standard
bivariate normal distribution and an ``estimated'' super-level set. We let $A^c$ be the complement of a set $A$.  For $\delta > 0$ and $\bx \in \RR^d$, let $B(\bx, \delta) := \lb \bs y \in \RR^d \colon \| \bs y - \bx \| \le \delta \rb$, and
for a set $A$, let $A^\delta := \cup_{\bx \in A} B(\bx, \delta)$.
\begin{figure}[htbp]
  \centering
  \includegraphics[width=0.5\textwidth]{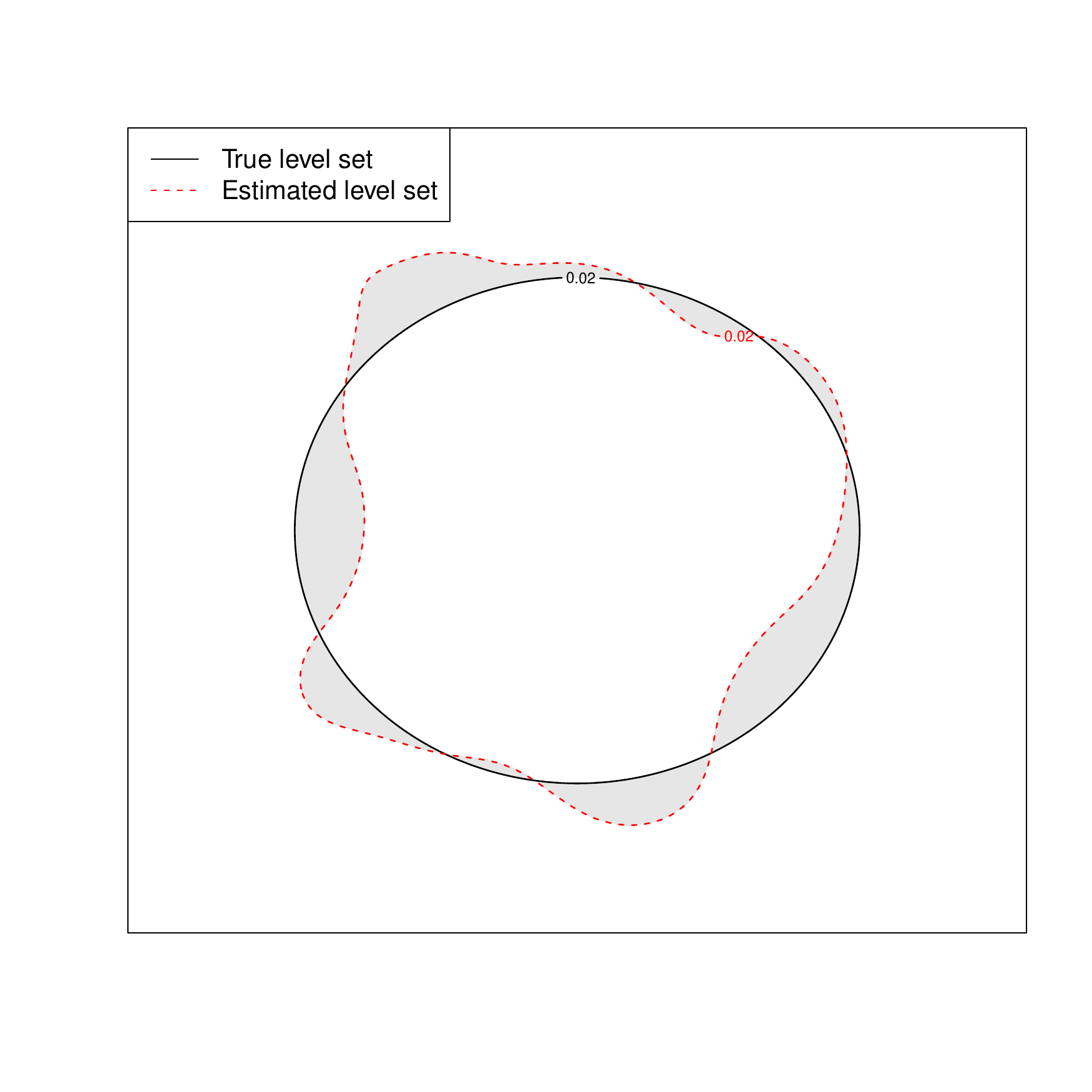}
  \caption{Symmetric difference between the true level set and an estimated level set. The solid black line is the boundary of the true level set and the dashed red line is the boundary of the estimated level set. The shaded area is the symmetric difference of the two sets.  \label{fig:symm-diff}}
\end{figure}

\subsection{Assumptions}
\label{subsec:assumptions}

To derive our asymptotic
expansion, we make the following basic assumptions on the underlying density,
kernel function and bandwidth matrix.
\begin{assumption}{D1a}
  \label{assm:DA-ls}
  $\phantom{blah}$
  \begin{enumerate}
  \item \label{assm:DA:item1}
    Let ${\bX}_1, \ldots, \bX_n$ be  i.i.d.\ from a bounded density $f_0$ on $\RR^d$, $d \ge 2$.
  \item\label{assm:DA:item2} Fix $ \inf_{x \in \RR^d} f_0(x) < c < \| f_0 \|_\infty$.  There exists a
    constant $a>0$ such that
    \begin{enumerate*}
    \item  \label{item:DA-ls-derivs} $f_0$ has two 
      bounded continuous partial derivatives over
      $U_a := \{\bs{x}:c-a\le f_0(\bx)\le c+a\}$,
    \item \label{item:DA-ls-grad}
      $\inf_{U_a}\|\nabla f_0\|>0$, and
    \item \label{item:DA-ls-localization} $U_a$ is contained in $\beta(c)^\delta$ for some $\delta > 0$.
    \end{enumerate*}
  \end{enumerate}
\end{assumption}

\begin{assumption}{D1b}
  \label{assm:DA-hdr}
  $\phantom{blah}$
  \begin{enumerate}
  \item \label{assm:DA:item1}
    Let ${\bX}_1, \ldots, \bX_n$ be  i.i.d. from a bounded density $f_0$ on $\RR^d$, $d \ge 2$.
  \item\label{assm:DA:item2} The density $f_0$ has two bounded continuous partial
    derivatives  for all $\bx\in\RR^d$.
  \item \label{assm:DA:item3} There exists a constant $a > 0$ such that
    $U_a:= \{\bx:f_{\tau,0}-a\le f_0(\bx)\le f_{\tau,0}+a\}$ satisfies
    \begin{enumerate*}
    \item \label{item:DA-hdr-grad} $\inf_{U_a}\|\nabla f_0\|>0$,
      and
    \item \label{item:DA-hdr-localization}
      $U_a$ is contained in  $\beta_\tau^\delta$ for some $\delta > 0$.
    \end{enumerate*}
    \begin{mylongform}
      \begin{longform}
      \item \label{assm:DA:item 4}For any $c>0$, $\lambda(\{\bx\in \RR^d:f_0(\bx)=c\})=0$, where
        $\lambda$ is the Lebesgue measure in $\RR^d$. (indicated by 3, see \citet{NunezGarcia:2003ci})
      \end{longform}
    \end{mylongform}
  \end{enumerate}
\end{assumption}

Assumption~\ref{assm:DA-ls} will be used for LS estimation and
Assumption~\ref{assm:DA-hdr} for HDR estimation.  We need the stronger global twice differentiability assumption in HDR estimation because of the need to estimate $\fftau$ (which involves estimating the $f_0$-probability  content of $\mc{L}_\tau$).  The global twice differentiability assumption in Assumption~\ref{assm:DA-hdr} could be weakened to an assumption of twice differentiability either on  $\mc{L}_\tau^\delta$ or on $(\mc{L}_\tau^c)^\delta$.

Assumptions~\ref{assm:DA-ls} and
\ref{assm:DA-hdr}
entail that the gradient of $f_0$ is nonzero on (a neighborhood of) the level set of interest.  This implies by the preimage theorem that the level set $\beta$, taken to be either $\beta(c)$ or $\beta_\tau$, is a $(d-1)$-dimensional (boundaryless) manifold
\citep{Guillemin:1974ti}. 
The only additional assumption we need is one of compactness, which rules out only very pathological cases, where $f_0$ has ``spikes'' of increasingly small width going out towards infinity.

\begin{assumption}{D2}
  \label{assm:BA}
  Let $ \inf_{x \in \RR^d} f_0(x) < c < \| f_0 \|_\infty$ or $0 < \tau < 1$ be as in
  Assumptions~\ref{assm:DA-ls} and \ref{assm:DA-hdr}.  Assume that $\beta(c)$
  or $\beta_\tau$ is compact.
\end{assumption}


Our assumption on the kernel will come in the form of a so-called
{\em Vapnik-Chervonenkis (VC)} \citep{Dudley:1999dc} type of assumption.
For a
metric space $(T,d)$ and $\tau > 0$, the covering number $N( T, d, \tau)$ is
the smallest number of balls of radius $\tau$ (and centers which may or may
not be in $T$) needed to cover $T$.  If a class of functions ${\mc F}$ is a VC class, we have that
\begin{equation}
  \label{eq:VC-defn}
  \sup_P  N( {\mc F}, \| \cdot \|_{2,P}, \tau  \| F \|_{2, P} ) \le \lp \frac{A}{\tau} \rp^{v}
\end{equation}
for some positive $A, v$, where the sup is over all probability measures $P$, and where $F$ is the
envelope of ${\mc F}$ meaning $\sup_{f \in {\mc F}} |f| \le F$ (Chapter 2.6,
\cite{vanderVaart:1996tf}). We will simply directly assume that the needed
classes satisfy \eqref{eq:VC-defn}.  Thus our assumptions are as follows.
\begin{assumption}{K}
  \label{assm:KA}
  $\phantom{blah}$ 
  \begin{enumerate}[leftmargin=*] 
  \item \label{assm:KA-1}
    The kernel $K$ is an everywhere continuously differentiable  bounded density
    on $\RR^d$ with bounded partial derivatives.
    Both $\int K^2 \, d\lambda$ and
    $\int (\grad K) (\grad K)' \, d\lambda$ are finite or have finite entries,
    respectively.  Assume
    $\int K(\bx)\bx\,d\bx=\boldsymbol{0}$, $\int
    \bx\bx'K(\bx)\,d\bx=\mu_2(K)\boldsymbol{I}$, where
    $\boldsymbol{I}$ is the identity matrix and $\mu_2(K)=\int
    x_i^2K(\bx)\,d\bx$ is independent of $i$.
  \item  \label{assm:KA-VC}
    Assume that  \eqref{eq:VC-defn} is satisfied with $\mc{F}$ taken
    to be
    \begin{align}
      \label{eq:assumption-VC-K}
      & \lb K \lp  \bH^{-1/2}  (t - \cdot) \rp: t \in \RR^d,  \bH \in {\mc S} \rb
        \qquad       \text{ and }
      \\
      \label{eq:assumption-VC-gradK}
      & \lb \| \grad K \big(  \bH^{-1/2}  (t - \cdot) \big) \|
        : t \in \RR^d,  \bH \in {\mc S} \rb.
    \end{align}
  \end{enumerate}
\end{assumption}
\noindent Let $R(K) := \int K^2 d\lambda$ and let $R( \grad K)$ be the largest eigenvalue of $\int (\grad K) (\grad K)' \, d\lambda$.

\begin{assumption}{H}
  \label{assm:HA}
  $\phantom{blah}$
  \begin{enumerate}[leftmargin=*]
  \item\label{assm:HA-1}   Let $\bH \equiv \bH_n \in {\mc S}$, such that for some $c > 0$,
    $| \bH | \searrow 0$,
    $n |\bH|^{1/2} / \log |\bH|^{-1/2} \to \infty$,
    $\log \log n / \log |\bH|^{-1/2} \to 0$,
    as $n \to \infty$,
    and $|\bH_n|^{1/2} \le c | \bH_{2n}|^{1/2}$.
  \item \label{assm:HA-2}  Assume that
    $\lambda_{\max}(\bH)=O\{\lambda_{\min}(\bH)\}$ and
    $n |\bH|^{1/2} \lambda_{\min}(\bH) /
    \log |\bH|^{-1/2} \to \infty$
    and    $\lambda_{\max}=O(n^{-2/(4+d)})$ as $n\to\infty$.
    \begin{mylongform}
      \begin{longform}
        $\lambda_+^{(d+4)/4}=o\{(\log n/n)^{1/2}\}$,     $\lambda_+^{(d+8)/2}=O(n^{-1/2}\log n)$
      \end{longform}
    \end{mylongform}
  \end{enumerate}
\end{assumption}

\noindent Here, $a_n \searrow 0$ means that $a_n$ decreases
monotonically to $0$.
\begin{mylongform}
  \begin{longform}
    From the above assumption, we know
    $\lambda_{\max}(\bH)$ and $\lambda_{\min}(\bH)$ are of the same
    order. So we make  further assumptions on $\bH$  through
    $\lambda_{\max}(\bH)$. Let $\lambda_1\equiv\lambda_{\max}(\bH)$, and
    then we know $|\bH|\sim \lambda_1^d$ and $\tr(\bH)\sim \lambda_1$. We
    list out the assumptions we need for $\bH$ during the proof.
    \begin{enumerate}
    \item $n|\bH|^{1/2}\to \infty$.   This is equivalent to
      $n\lambda_1^{d/2}\to \infty$.
    \item For the proof of Lemma~\ref{lem:hdr-step2}, we require $\log
      n=O(n|\bH|^{1/2})$. So this requires $\liminf
      n\lambda_1^{d/2}/\log n>0$.
    \item In the proof of Lemma~\ref{lem:hdr-step3}, we need
      $\lambda_1=o(\delta_n)$ while $\log
      n=O(\delta_n^2n|\bH|^{1/2})$.
    \item When we apply Lemma~\ref{lem:COV-approx} before
      Lemma~\ref{lem:hdr-step5}, we need
      $\delta_n^2=o(\inv{\sqrt{n|\bH|^{1/2}}})$.
    \end{enumerate}

  \end{longform}
\end{mylongform}
\begin{mylongform}
  \begin{longform}
    We can see the above requirement are satisfied by our assumption
    on $\bH$. Since for any $n$ and  $\bH$, we can always let
    \begin{align*}
      \delta_n^2\sim \frac{\log n}{n\lambda_1^{d/2}}=\frac{\log n}{\sqrt{n\lambda_1^{d/2}}}\frac{1}{\sqrt{n\lambda_1^{d/2}}},
    \end{align*}
    then by our assumption on $\lambda_{-}$, $\log
    n=O(\delta_n^2n|\bH|^{1/2})$ and
    $\delta_n^2=o(\inv{\sqrt{n|\bH|^{1/2}}})$ are satisfied. We need
    to verify that $\lambda_1=o(\delta_n)$, which can be shown by
    \begin{align*}
      \frac{\lambda_1}{\delta_n}=\frac{\lambda_1}{\sqrt{\log
      n/(n\lambda_1^{d/2})}}=\frac{\lambda_1^{(d+4)/4}}{\sqrt{\log
      n/n}}\rightarrow 0,
    \end{align*}
    by our assumption of $\lambda_+$.
  \end{longform}
\end{mylongform}
Assumptios~\ref{assm:DA-ls} and \ref{assm:DA-hdr} are standard in the KDE literature
(see, e.g., page 95 of \cite{Wand:1995kv}). 
Note that Assumption~\ref{assm:DA:item3} of Assumption~\ref{assm:DA-hdr}
implies that there exists a constant $L > 0$ such that  for $\delta>0 $ small enough that
$\lambda( f_0^{-1}( [\fftau - \delta, \fftau+\delta])) \le L \delta $;
this is a standard type of assumption that appears in the level set estimation literature \citep{Polonik:1995kr}.
Assumption~\ref{assm:BA} is not very limiting and only rules out pathological cases.

Our Assumption~\ref{assm:KA} 
on the kernel function is not restrictive and all of the conditions imposed are fairly standard.  For
Assumption~\ref{assm:KA-1} 
see, e.g., page 95 of \cite{Wand:1995kv} where similar conditions are imposed.
Assumption~\ref{assm:KA-VC} is also fairly standard in the KDE literature
(e.g., \cite{Chen:2015uj} uses similar conditions in the context of
inference for level sets).  This assumption is needed to apply the results
of \cite{Gine:2002jc} to get almost sure convergence rates of $\ffnH$ and $\grad \ffnH$.  Assumption $K_1$ of \cite{Gine:2002jc} (or
Assumption~K, page 2572, of \cite{MR2078551}) is an easy-to-verify
condition that implies Assumption~\ref{assm:KA-VC} holds, and shows that
Assumption~\ref{assm:KA-VC} holds for
Gaussian kernels and for many compactly supported kernels.

The expansions given in our Theorem~\ref{thm:levelset} and \ref{thm:hdr}
hold for the range of bandwidths given in Assumption~\ref{assm:HA}.  This
is sufficient to develop a practical bandwidth selector, since larger or
smaller bandwidths can be easily ruled out.  See
Corollaries~\ref{cor:ls-oracle-bandwidth} and \ref{cor:hdr-oracle-bandwidth}.

\subsection{Asymptotic risk expansions}

Our main results are stated in the following
two theorems.  The first gives  the asymptotic risk expansion for level set estimation.
Let $\Phi(\cdot)$ and $\phi(\cdot)$ denote the
standard normal distribution function and density function, respectively.

\begin{theorem}
  \label{thm:levelset}
  For given constant $c$ with $ \inf_{\bx \in \RR^d} f_0(\bx) <c<\|f_0\|_{\infty}$,  let Assumptions
  \ref{assm:KA}, \ref{assm:HA}, \ref{assm:DA-ls} and \ref{assm:BA} hold. Moreover, the kernel function $K$ has bounded
  support. Then
  \begin{equation*}
    \bb{E}\ls \mu_{f_0}\{\mc{L}(c)\Delta\widehat{\mc{L}}_{\bH}(c)\}\rs
    =     \LS(\bH)      + o \lb (n |\bH|^{1/2})^{-1/2}+\tr(\bH)\rb
  \end{equation*}
  as $n\rightarrow \infty$,
  where
  \begin{align*}
    \LS(\bH) :=
    \frac{c}{\sqrt{n|\bH|^{1/2}}} \int_{\beta(c)} \frac{2\phi(B_{\bx}(\bH))+2\Phi(B_{\bx}(\bH))B_{\bx}(\bH) - B_{\bx}(\bH)}{-A_{\bx}}\,d\mc{H}(\bx),
  \end{align*}
  \begin{equation}
    \label{eq:27}
    A_{\bx} := -\frac{\|\nabla f_0(\bx)\|}{\sqrt{R(K)c}},
    \quad \text{ and } \quad
    B_{\bx}(\bH) := -\frac{\sqrt{n|\bH|^{1/2}}D_1(\bx,\bH)}{\sqrt{R(K)c}},
  \end{equation}
  with $D_1(\bx,H):=\frac{1}{2}\mu(K)\tr(\bH\nabla^2f_0(\bx))$.
\end{theorem}

\medskip
\par\noindent
Note that the first summand (including the factor $c/\sqrt{n|\bH|^{1/2}}$) in the integral defining $\LS(\bH)$ is of
the order of magnitude of a variance term in a mean-squared error
decomposition, and the second two summands are of the same order of
magnitude of a squared bias term. The next theorem gives the HDR asymptotic risk expansion.

\begin{theorem}
  \label{thm:hdr}
  Let Assumptions \ref{assm:DA-hdr},\ref{assm:BA},\ref{assm:KA} and
  \ref{assm:HA} hold. Then
  \begin{equation*}
    \bb{E}\ls\mu_{f_0}\{\mc{L}_{\tau}\Delta\widehat{\mc{L}}_{\tau,\bH}\}\rs
    =  \HDR(\bH) + o \lb( n |\bH|^{1/2})^{-1/2}+\tr(\bH)\rb
  \end{equation*}
  as   $n \to \infty$,  where
  \begin{equation*}
    \HDR(\bH)
    :=   \frac{f_{\tau,0}}{\sqrt{n|\bH|^{1/2}}}\int_{\beta_{\tau}}
    \frac{2\phi(C_{\bx}(\bH)) + 2\Phi(C_{\bx}(\bH))C_{\bx}(\bH) - C_{\bx}(\bH)}{-A_{\bx}} \,
    d\mc{H}(\bx),
  \end{equation*}
  \begin{align*}
    C_{\bx}(\bH):=B_{\bx}(\bH)+\sqrt{\frac{n|\bH|^{1/2}}{R(K)f_{\tau,0}}}D_2(\bH).
  \end{align*}
  $A_{\bx}$ and $B_{\bx}(\bH)$ are defined in the same way as in
  Theorem~\ref{thm:levelset} with $c$ replaced by $f_{\tau,0}$. And
  \begin{align*}
    D_2(\bH)&:=w_0\left\{V_1(\bH)+V_2(\bH)\right\},
  \end{align*}
  with $ w_0:=(\int_{\beta_{\tau}}1/\nabla
  f_0\,d\mc{H})^{-1}$ and
  \begin{align*}
    V_1(\bH):=\int_{\beta_{\tau}}\frac{D_1(\bx,\bH)}{\|\nabla
    f_0(\bx)\|}\,d\mc{H}(\bx)\qquad
    V_2(\bH):=\inv{f_{\tau,0}}\int_{\mc{L}_{\tau}}D_1(\bx,\bH)\,d\bx.
  \end{align*}
\end{theorem}

\medskip
We defer  the proofs to the appendix.
Next, we would like to study the theoretical behavior of the
minimizers of $\LS(\cdot)$ and $\HDR(\cdot)$.
Note that the minimizers of $\LS(\cdot)$ or of $\HDR(\cdot)$ are not
practically usable bandwidth matrices, since $\LS(\cdot)$
and $\HDR(\cdot)$  depend on the true, unknown density $f_0$.  We will
discuss estimation of $\HDR(\cdot)$ and of $\LS(\cdot)$ and
practical bandwidth selectors in the next section.  Presently, we
consider the minimizers of  $\LS(\cdot)$ and $\HDR(\cdot)$, which serve as  {\em oracle} bandwidth selectors.

Unfortunately, $\LS(\cdot)$ and $\HDR(\cdot)$ are  quite  complicated functions so studying their minimizers in general is not at all straightforward.  Thus we will make some simplifying assumptions.  We will
consider $f_0$ that is unimodal and spherically symmetric about some
point (taken to be the origin in
Corollary~\ref{cor:ls-oracle-bandwidth} and \ref{cor:hdr-oracle-bandwidth}).  We will consider optimizing
over the subclass $\mc{S}_1 := \lb h^2 \bs{I} : h > 0 \rb $ of
bandwidth matrices, where $\bs{I}$ is the $d \times d$ identity
matrix.
These assumptions are made largely for simplicity and ease of presentation of the following two corollaries, and are far from necessary for the conclusions to hold.  We discuss these assumptions again after the corollaries.
By a slight abuse of notation, we let $\LS(h) \equiv \LS(
h^2 \bs{I})$ and $\HDR(h) \equiv \HDR( h^2 \bs{I}).$

\begin{corollary}
  \label{cor:ls-oracle-bandwidth}
  Let the assumptions of Theorem~\ref{thm:levelset} hold.
  Assume further  that  $f_0(x) = g( \|x\|)$
  and that the function $g(r) $ defined for $r > 0$ is strictly decreasing on
  $[0,\infty)$.
  Then
  there exists
  a constant $s_{\text{opt}}$ depending on $f_0$ and $K$ (but not on $n$) such that
  there is  a unique positive number $h_{\text{opt}} = \argmin_{h \in [0,\infty)} \LS(h)$
  satisfying
  \begin{equation*}
    h_{\text{opt}} = s_{\text{opt}} n^{-1 / (d+4)}
    \quad \text{ and } \quad
    h_0 =  h_{\text{opt}} 
    (1+o(1))
    \qquad     \text{ as } n \to \infty,
  \end{equation*}
  where $h_0$ is any minimizer of $\EE[ \mu_{f_0}\{\mc{L}(c)\Delta\widehat{\mc{L}}_{\bH}(c)\} ]$.
\end{corollary}

\begin{corollary}
  \label{cor:hdr-oracle-bandwidth}
  Let the assumptions of Theorem~\ref{thm:hdr} hold.
  Assume further  that  $f_0(x) = g( \|x\|)$
  and that the function $g(r) $ defined for $r > 0$ is strictly decreasing on
  $[0,\infty)$.
  Then
  there exists
  a constant $s_{\text{opt}}$ depending on $f_0$ and $K$ (but not on $n$) such that
  there is  a unique positive number $h_{\text{opt}} = \argmin_{h \in [0,\infty)} \HDR(h)$
  satisfying
  \begin{equation*}
    h_{\text{opt}} = s_{\text{opt}} n^{-1 / (d+4)}
    \quad \text{ and } \quad
    h_0 =  h_{\text{opt}} 
    (1+o(1))
    \qquad     \text{ as } n \to \infty,
  \end{equation*}
  where $h_0$ is any minimizer of $\EE[ \mu_{f_0}\{\mc{L}_{\tau}\Delta\widehat{\mc{L}}_{\tau,\bH}\} ]$.
\end{corollary}
\medskip
\noindent
The proof of the two corollaries follows exactly the same way, so we
provide the proof for HDR estimation and omit that for LS
estimation.  The corollaries tell us the order of magnitude of the true optimal bandwidths and  of the oracle bandwidths.  We used the assumptions of unimodality and spherical symmetry because
these assumptions imply that $f_0$, $\grad f_0$, and $\hess f_0$ are
constant on $\beta_{\tau}$ and $\beta(c)$.  We believe that (an analogous form of) the conclusions of
Corollary~\ref{cor:ls-oracle-bandwidth} and
\ref{cor:hdr-oracle-bandwidth} hold for $\bH_{\text{opt}} \in
\argmin_{\bH \in \mc{S}} \HDR(\bH)$ and for $\bH_{\text{opt}} \in
\argmin_{\bH \in \mc{S}} \LS(\bH)$, and for much more general densities $f_0$.  Our simulations show that our practical bandwidth selector (studied in the next section) does not require such extreme assumptions.

\section{Bandwidth selection methodology}
\label{sec:methodology}
In the previous section, we provided asymptotic expansions of
symmetric risks for  HDR estimation and LS estimation, which
could be used as guidance for bandwidth selection in those two
scenarios. Minimizers of $\LS(\bH)$ and $\HDR(\bH)$  are natural
bandwidth selectors for HDR estimation and LS estimation, respectively.
The theoretical performance of the bandwidth selector using ``oracle''
knowledge of  the functionals of the true density  is studied in
Corollary~\ref{cor:ls-oracle-bandwidth} and \ref{cor:hdr-oracle-bandwidth}. Of course, in practice, one does
not have this oracle knowledge. In the present section, we develop an
effective practical bandwidth selection procedure for HDR estimation
(a procedure for level set estimation is simpler and can be derived in a
similar way). We will also study the theoretical performance of our bandwidth selector
restricted to a simplified class $\mc{S}_1=\{h^2\bs{I},h>0\}$.

Since there are unknown quantities that $\HDR(\bH)$ depends on, a natural ``plug-in'' approach is to estimate those quantities using
different kernel density estimators and plug the estimates in.
Moreover, the unknown functionals depend on the truth through
$f_0,\grad f_0,\grad^2f_0$, so we will  use three pilot kernel
density estimators. To be specific, we use $\widehat{f}_{n,\bH_0}$ to
estimate $f_{\tau,0}$ and $\mc{L}_{\tau}$; we use
$\nabla\widehat{f}_{n,\bH_1}$ to estimate $\nabla f_0$, and
$\beta_{\tau}$ combined with  the pilot estimator of $f_{\tau,0}$; we use
$\nabla^2\widehat{f}_{n,\bH_2}$ to estimate $\nabla^2 f_0$, where $\bH_0$,
$\bH_1$ and $\bH_2$ are corresponding pilot bandwidth matrices for the three
kernel density estimators.  (One could also use three different kernels for
$\widehat{f}_{n,\bH_i}$, $i=0,1,2$, but we will use the same kernel for all
three.)  For our theoretical results to hold, we require just  the
bandwidth matrix $\bH_r$ to be of the optimal order
for estimating the $r$th derivatives of $f_0$ (see
Corollary~\ref{cor:hdr-bandwidth-selector} and Assumption~\ref{assm:HA2},
below).  We use two-stage direct plug-in estimators for the pilot bandwidths
in our algorithm below, which converge at the correct rate. A detailed
description about plug-in estimators could be found in \citet[Chapter 3]{Wand:1995kv} and \cite{chacon2010multivariate}.

Once we have those estimated functionals, we can plug them into $\HDR(\bH)$ to obtain an estimated loss function $\widehat{\HDR}(\bH)$. Note  $\bH$ appears in the integrand of a Hausdorff integral and cannot be factored out of the integral; thus minimizing $\widehat{\HDR}(\bH)$ directly is infeasible. Instead, we minimize a discretized approximation to
$\widehat{\HDR}(\bH)$. To illustrate this idea, we use the minimization of $\HDR(\bH)$ as an example. Let $\mc{A}=\{A_i\}_{i=1}^m$ be a partition of $\beta_{\tau}$ such that $\mc{H}(A_i)$ is sufficiently small for $i=1,2,\ldots,m$. Then $w_0=(\int_{\beta_{\tau}}\frac{1}{\|\grad f_0\|}\,d\mc{H})^{-1}$ can be approximated by $ \tilde{w}_0=\sum_{i=1}^m\inv{\|\grad f_0(\tilde{\bx}_i)\|}\mc{H}(A_i)$, where $\tilde{\bx}_i$ is an arbitrary point belonging to $A_i$.
Note for $d=2$, $\mc{H}(A_i)$ is well approximated by the length of
the line segment connecting the boundary points of $A_i$. $V_1(\bH)$
and $V_2(\bH)$ can be computed approximately in  similar
ways. Replacing  $w_0$, $V_1(\bH)$, $V_2(\bH)$ with corresponding discretized approximations in $C_{\bx}(\bH)$ gives us an approximation $\tilde{C}_{\bx}(\bH)$ for each $\bx$.  Then
\begin{align}
  \HDR(\bH)
  &\approx
    \frac{f_{\tau,0}}{\sqrt{n|\bH|^{1/2}}}\int_{\beta_{\tau}}
    \frac{2\phi(\tilde{C}_{\bx}(\bH)) + 2\Phi(\tilde{C}_{\bx}(\bH))\tilde{C}_{\bx}(\bH) - \tilde{C}_{\bx}(\bH)}{-A_{\bx}} \,
    d\mc{H}(\bx)  \nonumber 
  \\
  &\approx
    \frac{f_{\tau,0}}{\sqrt{n|\bH|^{1/2}}}\sum_{i=1}^m
    \frac{2\phi(\tilde{C}_{\tilde{\bx}_i}(\bH)) + 2\Phi(\tilde{C}_{\tilde{\bx}_i}(\bH))\tilde{C}_{\tilde{\bx}_i}(\bH) - \tilde{C}_{\tilde{\bx}_i}(\bH)}{-A_{\tilde{\bx}_i}} \mc{H}(A_i).
    \label{eq:numerical-hausdorff-integral-2} 
\end{align}
The last line above provides a computable, optimizable and  close approximation to
$\HDR(\bH)$ as long as $\mc{H}(A_i)$ is small enough for each
$i$. We use $K=\phi$ throughout the algorithm.

\noindent The full algorithm for the HDR bandwidth
selector is as follows:
\begin{enumerate}
\item With given i.i.d random sample $\bX_1,\bX_2,\ldots,\bX_n$, estimate
  $\bH_0$, $\bH_1$, $\bH_2$ using two-stage direct plug-in strategies.
\item Obtain the pilot estimator of $f_0$, $\nabla f_0$, $\nabla^2
  f_0$ based on  the kernel density estimators $\widehat{f}_{n,\bH_0}$,
  $\widehat{f}_{n,\bH_1}$, $\widehat{f}_{n,\bH_2}$.
\item \label{item:step3} Let
  $\widehat{f}_{\tau,n,\bH_0}:=\inf\{y\in(0,\infty):\int_{\RR^d}\widehat{f}_{n,\bH_0}(\bx)\one_{\{\widehat{f}_{n,\bH_0}(\bx)\ge
    y\}}\,d\bx\le 1-\tau\}$ be the pilot estimator of   $f_{\tau,0}$,
  $\widehat{\mc{L}}_{\tau,\bH_0}:=\{\bx\in
  \RR^d:\widehat{f}_{n,\bH_0}(\bx)\ge \widehat{f}_{\tau,n,\bH_0}\}$ be
  the pilot estimator of
  $\mc{L}_{\tau}$ and
  $\widehat{\beta}_{\tau,\bH_1}:=\{\bx\in\RR^d:\widehat{f}_{n,\bH_1}(\bx)=\widehat{f}_{\tau,n,\bH_0}\}$
  be the pilot estimator of
  $\beta_{\tau}$.
\item   Substitute the estimators from Step 2 and 3 into the
  expressions for $C_{\bx}$  and $A_{\bx}$ to obtain
  $\widehat{C}_{\bx}$ and $\widehat{A}_{\bx}$. Then
  \begin{equation*}
    \widehat{\HDR}(\bH) 
    =\frac{\widehat{f}_{\tau,n, \bH_0}}{\sqrt{n|\bH|^{1/2}}}\int_{\widehat{\beta}_{\tau,\bH_1}}
    \frac{2\phi(\widehat{C}_{\bx}(\bH)) + 2\Phi(\widehat{C}_{\bx}(\bH))\widehat{C}_{\bx}(\bH) - \widehat{C}_{\bx}(\bH)}{-\widehat{A}_{\bx}} \,
    d\mc{H}(\bx).
  \end{equation*}
\item Minimize the discretized approximation of $
  \widehat{\HDR}(\bH)$ described in the previous paragraph with Newton's method to obtain the
  estimated optimal HDR bandwidth.
\end{enumerate}

Note for the above procedure, in step 3, unlike the pilot estimator for
$\mc{L}_{\tau}$, the pilot estimator for $\beta_{\tau}$ is obtained using
$\widehat{f}_{n,\bH_1}$ with $\widehat{f}_{\tau,n,\bH_0}$ as the
level. The reason we use $\widehat{f}_{n,\bH_1}$ instead of
$\widehat{f}_{n,\bH_0}$ is because the error bound for estimating
$\beta_{\tau}$ depends on the difference between the gradient of true
density and that of the kernel density estimator and using
$\widehat{f}_{n,\bH_1}$ yields a better error bound (See Lemma
\ref{lem:Hauss-diff}  and proof of Corollary
\ref{cor:ls-bandwidth-selector},  \ref{cor:hdr-bandwidth-selector}
for details).

Newton's method does not guarantee the optimum will be a positive
definite bandwidth matrix. Luckily, in practice the global minimum
appears to always be positive definite. The objective function
$\widehat{\HDR}$ appear to be locally convex although not globally
convex (see Figures~\ref{fig:level-risk}  and \ref{fig:hdr-risk} for
some plots of $\LS(\cdot)$ and $\HDR(\cdot)$), so one has to be slightly careful about starting values for
Newton's algorithm.

Notice also that
in
Step~\ref{item:step3} of the above algorithm
we
need to calculate the level
$\widehat{f}_{\tau,n,\bH_0}$ having $\widehat{f}_{n,\bH_0}$-probability $1-\tau$.  \citet{Hyndman:1996bf} suggests two similar methods for calculating $\widehat{f}_{\tau,n,\bH_0}$.  One is to use an appropriate empirical quantile of the values $\widehat{f}_{n,\bH_0}(\bX_i),$ $i=1,\ldots, n$ (``Approach H1'').  An approach of this type is studied by \cite{Cadre:2013fv} (and by \cite{Chen:2016vv} in calculating his $\hat{\alpha}_n(x)$).  However, this estimator is not equal to $\widehat{f}_{\tau,n,\bH_0}$, and we have not yet quantified the difference, so we choose not to use this approach.  Alternatively, \citet{Hyndman:1996bf} suggests resampling $\tilde{\bX}_1, \ldots, \tilde{\bX}_M \iid \ffnH$, and then using the appropriate empirical quantile of $\widehat{f}_{n,\bH_0}(\tilde{\bX}_i),$ $i=1,\ldots, M$ (``Approach H2'').  Any desired accuracy can be attained by taking $M$ large enough. Another method is to simply use numeric integration: one can do a binary search over $(0, \| \widehat{f}_{n,\bH_0} \|_\infty )$, computing the integral (numerically) at each level until one arrives at $\widehat{f}_{\tau,n,\bH_0}$ within desired accuracy.  When $d=2$, we found the numeric integration and binary search to be the fastest method for calculating $\widehat{f}_{\tau,n,\bH_0}$.  We suspect for higher dimensions, Approach H2 will be faster than numeric integration.
Of course, Approach H1 is faster than the other two, and so it would be helpful to study how the Approach H1 estimator compares to $\widehat{f}_{\tau,n,\bH_0}$.

In our pilot estimation process when $d=2$, we use numerical
interpolation to generate points on $\wh \beta_{\tau, \bH_1}$ and to calculate
$\mc{A}$. In more detail: we generate dense grid points along both the
$x$-axis and the $y$-axis, and we estimate the density values at those
grid points. Then we perform interpolation between grid points to get points such that the
estimated density values at those points are (approximately)
$\wh f_{\tau,n, \bH_0}$, and those points induce a partition of $\wh \beta_{\tau, \bH_1}$.
Then for any $A_i$ in the partition, $A_i$ is defined by two end points, and  $\mc{H}(A_i)$ can be approximated by the length of the line segment connecting those two end points. 
By generating
enough dense and equally spaced grid  points, we   expect those line
segments will approximate the true  partition $\mc{A}$ well and thus
the Hausdorff integral will also be well approximated.
However, this method is hard to implement in dimension larger than
$2$ because  there is no simple  approximation for  the volumes of corresponding
partition sets of $\wh \beta_{\tau,\bH_1}$.  One approach that may be fruitful for solving this problem is to use  Quasi-Monte Carlo integration to calculate the
Hausdorff integral \citep[see][]{de2018quasi}. The idea is to generate
a set of points $\bs{b}_1,\ldots,\bs{b}_m$  on the manifold
$\beta$ such that those points are
approximately uniformly distributed and then we can approximate
$\int_{\beta}\gamma(\bx)\,d\mc{H}$ by
$\frac{1}{m}\sum_{i=1}^m\gamma(\bs{b}_i)$. Analysis and numerical
simulation for the method has been done for special Hausdorff
integrals over special manifolds (cone, cylinder, sphere and
torus). There is further work needed to extend the method to the
more general manifolds that arise in our
problem, which we believe is non-trivial and beyond the scope of this paper.

Note that the method just described for computing the approximation \eqref{eq:numerical-hausdorff-integral-2} 
can be implemented as a so-called midpoint method of numerical integration,
for which
classical analysis  shows an error rate of $O(m^{-2})$ ($m$ is the number of equi-sized partitioning sets of the interval), provided that the  function being integrated has bounded second derivative and the domain being integrated is a compact interval in $\RR$
\citep{Hammerlin:1991kz}. The same error applies for using the midpoint method to numerically compute Hausdorff integrals over one dimensional compact manifolds embedded in $\RR^2$, by the change of variables Theorem 2 (page 99) of \cite{Evans:2015uy}.   Thus the errors for our selected bandwidths 
in the corollaries below will also have an error dependent on $m$, but in our experience  $m$ can be chosen large enough that this is negligible (when $d=2$), so we do not include it in the analysis.

To give the asymptotic  performance of our bandwidth selector,
we need  the following additional assumptions.
\begin{assumption}{D3}
  \label{assm:DA3}
  The true density function $f_0$ has four
  continuous bounded and square integrable derivatives.

\end{assumption}
\begin{assumption}{K2}
  \label{assm:KA2}
  $K$ is symmetric, i.e.,
  $K(x_1,\ldots,x_i,\ldots,x_d)=K(x_1,\ldots,-x_i,\ldots,x_d)$ for
  $i=1,\ldots,d$. And all the first and second partial derivatives of $K$
  are square integrable.
\end{assumption}
\begin{assumption}{H2}
  \label{assm:HA2}
  For $r=0,1,2$,
  the bandwidth matrix $\bH_r$ is symmetric, positive definite,
  such that $\bH_r \to 0$ elementwise, and
  $n^{-1} |\bH_r |^{-1/2}(\bH_r^{-1})^{\otimes r}\to 0$ as $n\to\infty$,
  where $\otimes$ stands for Kronecker product.
\end{assumption}

\noindent
This assumption and notation is as in   \citet{chacon2011asymptotics}.
Here for a matrix $\bs{A}$, $\bs A^{\otimes 0} = 1 \in \RR$ and $\bs A^{\otimes 1} = \bs{A}$.
Now, recall that
\begin{align*}  \LS(h):=\LS(h^2\bs{I})=\frac{c}{(nh^d)^{1/2}}\int_{\beta(c)}\frac{\phi(B_{\bx}(h))+2\Phi(B_{\bx}(h))B_{\bx}(h)-B_{\bx}(h)}{-A_{\bx}}\,d\mc{H}(\bx),
\end{align*}
and $B_{\bx}(h)=(bh^{d+4})^{1/2}F_{\bx}$ with
$F_{\bx}=-\frac{1}{2}\mu(K)\tr(\nabla^2 f_0(\bx))/\sqrt{R(K)c}$. And
\begin{align*}
  \text{HDR}(h):=\text{HDR}(h^2\bs{I})=\frac{f_{\tau,0}}{(nh^d)^{1/2}}\int_{\beta_{\tau}}\frac{\phi(C_{\bx}(h))+2\Phi(C_{\bx}(h))C_{\bx}(h)-C_{\bx}(h)}{-A_{\bx}}\,d\mc{H}(\bx),
\end{align*}
and
$   C_{\bx}(h)=(nh^{d+4})^{1/2}G_{\bx}$,
where
\begin{align*}
  G_{\bx}=-\frac{\mu(K)\tr(\nabla^2f_0(\bx))}{\sqrt{R(K)f_{\tau,0}}}+\frac{w_0\int_{\beta_{\tau}}\frac{\mu(K)\tr(\nabla^2
  f_0)}{2\|\nabla f_0\|}\,d\mc{H}+\frac{w_0}{f_{\tau,0}}\int_{\mc{L}_{\tau}}\frac{\mu(K)\tr(\nabla^2f_0)}{2}\,d\lambda}{\sqrt{R(K)f_{\tau,0}}}.
\end{align*}
By letting $s=(nh^{d+4})^{1/2}$, we see that minimizing $\LS(h)$  is equivalent
to minimizing
\begin{align*} \text{AR}_{\LS}(s):=s^{-d/(d+4)}\int_{\beta(c)}\frac{\phi(sF_{\bx})+2\Phi(sF_{\bx})sF_{\bx}-sF_{\bx}}{-A_{\bx}}\, d\mc{H}(\bx),
\end{align*}
and minimizing  $\text{HDR}(h)$ is equivalent to minimizing
\begin{align*}
  \text{AR}_{\HDR}(s):=s^{-d/(d+4)}\int_{\beta_{\tau}}\frac{\phi(sG_{\bx})+2\Phi(sG_{\bx})sG_{\bx}-sG_{\bx}}{-A_{\bx}}\, d\mc{H}(\bx).
\end{align*}
he following corollaries show the convergence rate of the estimated
optimal bandwidth  for $\bH\in\mc{S}_1$.
\begin{corollary}\label{cor:ls-bandwidth-selector}
  Let Assumptions \ref{assm:DA-ls}, \ref{assm:BA}, \ref{assm:DA3},
  \ref{assm:KA}, \ref{assm:KA2} and
  \ref{assm:HA2} hold. Assume further that $s_{\text{opt}}$ is a unique
  minimizer of $\text{AR}_{\LS}(s)$ for $s>0$ and
  $\text{AR}^{\prime \prime}_{\LS}(s_{\text{opt}})>0$. Then
  \begin{align*}
    \frac{\hat{h}_{\text{opt}}}{h_{\text{opt}}}=1+O_p\lp n^{-2/(d+8)}\rp
    \quad \text{and} \quad
    \frac{\hat{h}_{\text{opt}}}{h_{0}}=1+O_p\lp
    n^{-2/(d+8)}\rp,
  \end{align*}
  as $n\to\infty$, where $\hat{h}_{\text{opt}}$ is the minimizer of
  $\widehat{\text{LS}}(h)$, $h_{\text{opt}}$ is the minimizer of
  $\text{LS}(h)$ and $h_{0}$ is any minimizer of
  $\EE[\mu_{f_0}\{\mc{L}(c)\Delta \widehat{\mc{L}}_{\bH}(c)\}]$ over the
  class $\mc{S}_1=\{h^2\bs{I},h>0\}$.
\end{corollary}

\begin{corollary}\label{cor:hdr-bandwidth-selector}
  Let Assumptions \ref{assm:DA-hdr}, \ref{assm:DA3}, \ref{assm:KA}, \ref{assm:KA2} and
  \ref{assm:HA2} hold. Assume further that $s_{\text{opt}}$ is a unique
  minimizer of $\text{AR}_{\HDR}(s)$ for $s>0$ and
  $\text{AR}^{\prime \prime}_{\HDR}(s_{\text{opt}})>0$. Then
  \begin{align*}
    \frac{\hat{h}_{\text{opt}}}{h_{\text{opt}}}=1+O_p\lp n^{-2/(d+8)}\rp,
  \end{align*}
  as $n\to\infty$, where $\hat{h}_{\text{opt}}$ is the minimizer of
  $\widehat{\text{HDR}}(h)$ and $h_{\text{opt}}$ is the minimizer of $\text{HDR}(h)$.
\end{corollary}
\medskip
\noindent
Corollaries \ref{cor:ls-bandwidth-selector} and
\ref{cor:hdr-bandwidth-selector} both assume existence of a point $s_{\text{opt}}$.
Corollary~\ref{cor:ls-oracle-bandwidth} and \ref{cor:hdr-oracle-bandwidth} show the existence of $s_{\text{opt}}$ under one set of assumptions, although (as discussed after those corollaries) this conclusion holds in many other scenarios.
\begin{remark}
  In Corollary~\ref{cor:ls-bandwidth-selector}, we provide the rates of convergence for both the estimated optimal bandwidth to the oracle bandwidth selector and the estimated optimal bandwidth to the true minimizer of $\EE[\mu_{f_0}\{\mc{L}(c)\Delta\widehat{\mc{L}}_{\bH}(c)\}]$, while in Corollary~\ref{cor:hdr-bandwidth-selector}, we only provide the rate of convergence for the estimated optimal bandwidth to the oracle bandwidth selector. The main difficulty for proving the convergence rate of the estimated optimal bandwidth to the true minimizer of $\EE[\mu_{f_0}\{\mc{L}_{\tau}\Delta\widehat{\mc{L}}_{\tau,\bH}\}]$, as we can see from the proof of Theorem~\ref{thm:hdr}, is understanding the $\Var \fftaun$ term. At present, we can only show that $\Var \fftaun$ is $o(\frac{1}{n|\bH|^{1/2}})$, but do not have a more explicit expression.  Thus (even with higher order derivative assumptions) we cannot say anything stronger  about $\Var \fftaun$, which is different
  than when $\beta_\tau$ is a discrete point set, in the $d=1$ case.
\end{remark}
\begin{remark}
  The rates of convergence given in Corollaries \ref{cor:ls-bandwidth-selector} and \ref{cor:hdr-bandwidth-selector} are known as {\it relative rates of convergence} since they are of the form $(\hat{h}_{\text{opt}} - \tilde{h}) / \tilde{h}$ for some $\tilde{h}$ (which is itself converging to $0$) \citep{Wand:1995kv}.  One can compare the relative rates from Corollaries \ref{cor:ls-bandwidth-selector} and \ref{cor:hdr-bandwidth-selector} to the relative rates of other KDE bandwidth selectors.  If we plug $d=1$ into the rate $n^{-2 / (d+8)}$ we recover the rate that arose in Theorem 3 of \citet{Samworth:2010cj}.  We can also make comparisons to bandwidth selector relative rates based on global loss functions.
  \begin{mylongform}
    \begin{longform}
      We consider the case where $d=1$,
      where the most results are available for comparison, even though we focus on the case $d \ge 2$ for our results in  Corollaries \ref{cor:ls-bandwidth-selector} and \ref{cor:hdr-bandwidth-selector}.
      \cite{Scott:1987vc} show that when $f_0$ has four derivatives that satisfy some further integrability conditions,
      then
      $n^{1/10} (\wh h - h_{\text{MISE}}) / h_{\text{MISE}} = O_p(1)$ where $\wh h$
      is either the bandwidth  based on
      a ``least squares (unbiased) cross validation'' procedure or
      on a ``biased cross validation'' procedure.  (This requires also some assumptions on the kernel.)
      This slow rate of $n^{-1/10}$ can be sped up; under
      various differentiability assumptions on $f_0$,
      $n^{5/14} (\wh h - h_{\text{MISE}}) / h_{\text{MISE}} = O_p(1)$ where $\wh h$
      is the bandwidth based on a ``two-stage direct plug-in'' procedure, a
      two-stage ``solve the equation'' procedure, or a ``smoothed cross
      validation'' procedure \citep{Wand:1995kv,Sheather:1991tp,Hall:1992gx}.  To
      achieve this rate, certain pilot bandwidths and kernels must be chosen
      correctly, in addition to having appropriate smoothness of $f_0$.  (In fact,
      somewhat more complicated versions of these procedures can achieve root-$n$
      rates of relative convergence, again with appropriate smoothness of $f_0$;
      but simulations show that very large sample sizes are needed for these more
      complicated procedures to perform well, see e.g.\ \cite[Section
      3.8]{Wand:1995kv}.)
      The $d=1$ rate exponent of $2/9 \approx .22$ is faster than $1/10$ but slower than $5/14 \approx .36$.
    \end{longform}
  \end{mylongform}
  \cite{Duong:2005ir} study relative rates of convergence for various bandwidth selectors to the bandwidth matrix that minimizes  mean integrated squared error, $E \int_{\RR^d} (\ffnH(\bx) - f_0(\bx))^2 \, d\bx$.
  (An alternative benchmark is the bandwidth that minimizes {\em integrated squared error}, $\int_{\RR^d} (\widehat{f}_{n,h}(\bx) - f_0(\bx))^2 \, d\bx$, for which e.g., LSCV performs well \citep{Hall:1987ik}, but the relative rates for that problem behave quite differently than the ones we study in Corollaries \ref{cor:ls-bandwidth-selector} and \ref{cor:hdr-bandwidth-selector}, so we do not mention them here.)
  Table 1 of \cite{Duong:2005ir}
  presents the  convergence rates for plug-in, unbiased cross validation, biased cross validation, and smoothed cross validation bandwidth matrix estimators.
  (See also \cite{Sain:1994hs,Wand:1994tn,Duong:2003kd,  
    Scott:1987vc,Sheather:1991tp,Hall:1992gx}.) 
  Consider $d \ge 2$.
  The unbiased and biased cross validation methods have relative convergence rates of $n^{-\min(d,4)/ (2d + 8)}$.  The smoothed cross validation method and the plug-in method of
  \cite{Duong:2003kd} both have rates of $n^{-2 / (d + 6)}$.  The
  plug-in method of \cite{Wand:1994tn} has a rate of $n^{-4 / (d + 12)}$ which is the fastest rate for all $d$.  The rate presented in our corollaries is faster than
  $n^{-\min(d,4)/ (2d + 8)}$ but slower than $n^{-2 / ( d + 6)}$.  This suggests that more careful development of our plug-in procedure, perhaps involving more careful pilot bandwidth selection procedures, could potentially improve the asymptotic rate.  However the analysis
  (in particular understanding how $\Var (\fftaun)$ behaves)
  may not be trivial.  Also,
  procedures with better asymptotics may be inferior until the sample size is unrealistically large (this is somewhat common  in bandwidth selection settings \cite[Section
  3.8]{Wand:1995kv}).
\end{remark}

\section{Simulations and data analysis}
\label{sec:simulations-data}
In Section~\ref{sec:methodology},  we used  LS$(\bH)$ and HDR$(\bH)$ to develop a
bandwidth selection procedure for level set and HDR estimation.
We have implemented our procedure in
an \proglang{R} \citep{R-core}
package \pkg{lsbs}.
In
this section, we assess the accuracy of $\text{LS}(\bH)$ and
$\text{HDR}(\bH)$ at approximating the true risks. We also use
simulation to compare our procedure with the least square cross
validation procedure (LSCV),
An established ISE-based
bandwidth selector
\citep[See][]{Rudemo:1982,Bowman:1984iv}. 
We simulate from the 12 bivariate normal mixture densities constructed by \citet{Wand:1993jl}.  These densities have a variety of shapes and have between 1 and 4 modes.  In addition to those 12 density functions, we also simulate from

\begin{align}
  \label{eq:sharpmode-density}
  \frac{2}{3}N\lp
  \begin{pmatrix}
    0\\0
  \end{pmatrix},
  \begin{pmatrix}
    1/4&0\\0&1
  \end{pmatrix}\rp+\frac{1}{3}N\lp
              \begin{pmatrix}
                0\\0
              \end{pmatrix},
  \frac{1}{50}  \begin{pmatrix}
    1/4&0\\0&1
  \end{pmatrix}\rp,
\end{align}
which is constructed to play a bivariate analogy to the sharp mode density 4 in
\citet{marron1992exact} (see also Figure~1 of Samworth and Wand
(2010)).  The specific form in \eqref{eq:sharpmode-density} is chosen to match that used by \cite{Qiao:2017wq}.

We will close this
section with a real data analysis in which we apply HDR estimation to
novelty detection for the Wisconsin Diagnostic Breast Cancer dataset
and Banknote Authentication dataset,
which are available on the UCI Machine Learning Repository (\url{http://archive.ics.uci.edu/ml/}).

\subsection{Assessment of approximation and estimation comparison}
Since it is infeasible to exactly
evaluate the true symmetric risk
$\EE[\mu_{f_0}\{\mc{L}_{\tau}\Delta\widehat{\mc{L}}_{\tau,\bH}\}]$, we approximate
the true risk through Monte Carlo. For given $n,\tau,\bH$,  for
a large Monte Carlo sample size $M$,
$
\EE[\mu_{f_0}\{\mc{L}_{\tau}\Delta\widehat{\mc{L}}_{\tau,\bH}\}]\approx\inv{M}\sum_{i=1}^M\mu_{f_0}\{\mc{L}_{\tau}\Delta\widehat{\mc{L}}^{[i]}_{\tau,\bH}\},
$
where
$\widehat{\mc{L}}^{[1]}_{\tau,\bH},\widehat{\mc{L}}^{[2]}_{\tau,\bH},\ldots,\widehat{\mc{L}}^{[M]}_{\tau,\bH}$
are $M$ independent realizations of $\widehat{\mc{L}}_{\tau,\bH}$. In  a multivariate KDE
the bandwidth matrix contains $d(d+1)/2$ parameters. For the purpose
of visualization, we restrict $\bH\in \mc{S}_1=\{h^2\bs{I}\}$ so that it can be
parametrized by a single parameter $h$.

Figures~\ref{fig:level-risk}
and \ref{fig:hdr-risk} compare the asymptotic risk approximation with
the simulated true risk for HDR estimation and LS estimation,
respectively, for  densities corresponding to Densities C, D,
E and K of \citet{Wand:1993jl}. Contour plots of the densities are
given in the top row of the figures. In Figure~\ref{fig:hdr-risk}, we choose
$\tau$ to be 0.2, 0.5 and 0.8 while in Figure~\ref{fig:level-risk}, we
use the same levels but with true level values computed from the
underlying true density functions. For both  scenarios, the sample
size is chosen to be 2000 and the kernel is set to be the Gaussian
kernel throughout the simulation (Theorem~\ref{thm:levelset} requires
$K$ to be compactly supported, but nonetheless, the simulation results
are not sensitive to the choice of Guassian kernel). We can see from Figures~\ref{fig:level-risk}
and \ref{fig:hdr-risk}, in both  scenarios, our asymptotic
expansions provide  a good approximation to the truth.  The approximation
works fairly well for the small values of bandwidth but  the discrepancy
becomes obvious when $h$ is larger, which is unlike  what was observed
from the simulation in univariate cases
\citep[see][]{Samworth:2010cj}.  This  is consistent with our
Assumption~\ref{assm:HA} which imposes an upper bound on the largest
eigenvalue of the bandwidth matrix, restricting it not to converge too
slowly. One more thing to notice from these two figures is that the
optimal bandwidth chosen from the asymptotic expansion serves as a
good approximation to the true optimal bandwidth, as we can see they
are quite close in most cases in simulation.
\begin{figure}
  \centering
  \includegraphics[width=\textwidth]{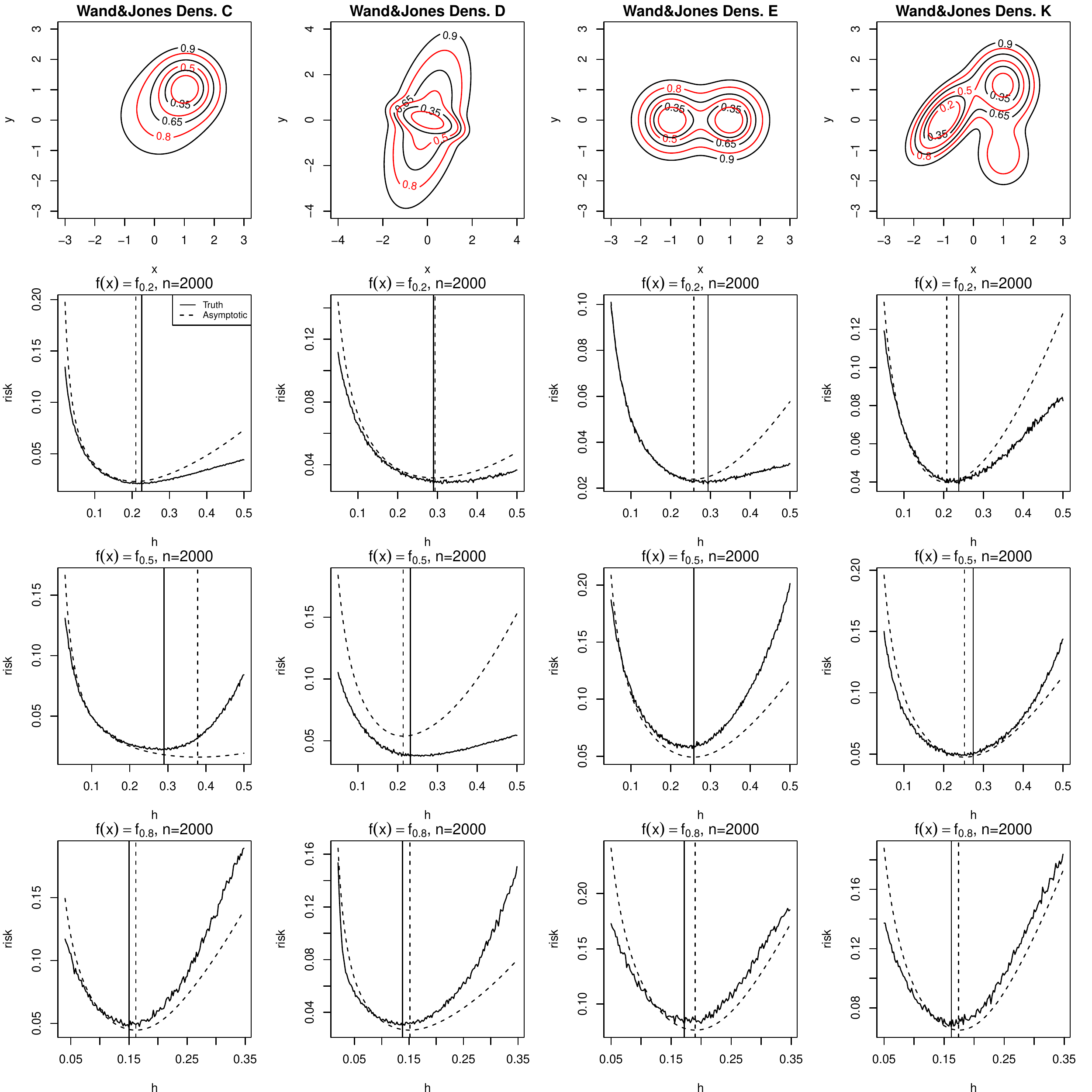}
  \caption{\label{fig:level-risk}Comparison of the simulated true risk function
    $\EE[\mu_{f_0}\{\mc{L}(c)\Delta\widehat{\mc{L}}_{\bH}(c)\}]$
    with $\text{LS}(\bH)$ for  four densities in
    \citet{Wand:1993jl}. The panels in the first row are the contour
    plots for  four densities with the contours of interest plotted
    in red color. The panels in the rest of the rows  are the
    comparison plots for the simulated true risk (solid line) and $\text{LS}(\bH)$ (dashed line) corresponding to the density
    at the top of  the column for $\tau=0.2,0.5,0.8$. The positions of the solid vertical
    line and the dashed line stand for the optimal bandwidths obtained
    from the simulated true risk and the asymptotic approximation
    respectively  over the restricted class $\mc{S}_1$. The sample size
    for all the cases is 2000.}
\end{figure}

\begin{figure}
  \centering
  \includegraphics[width=\textwidth]{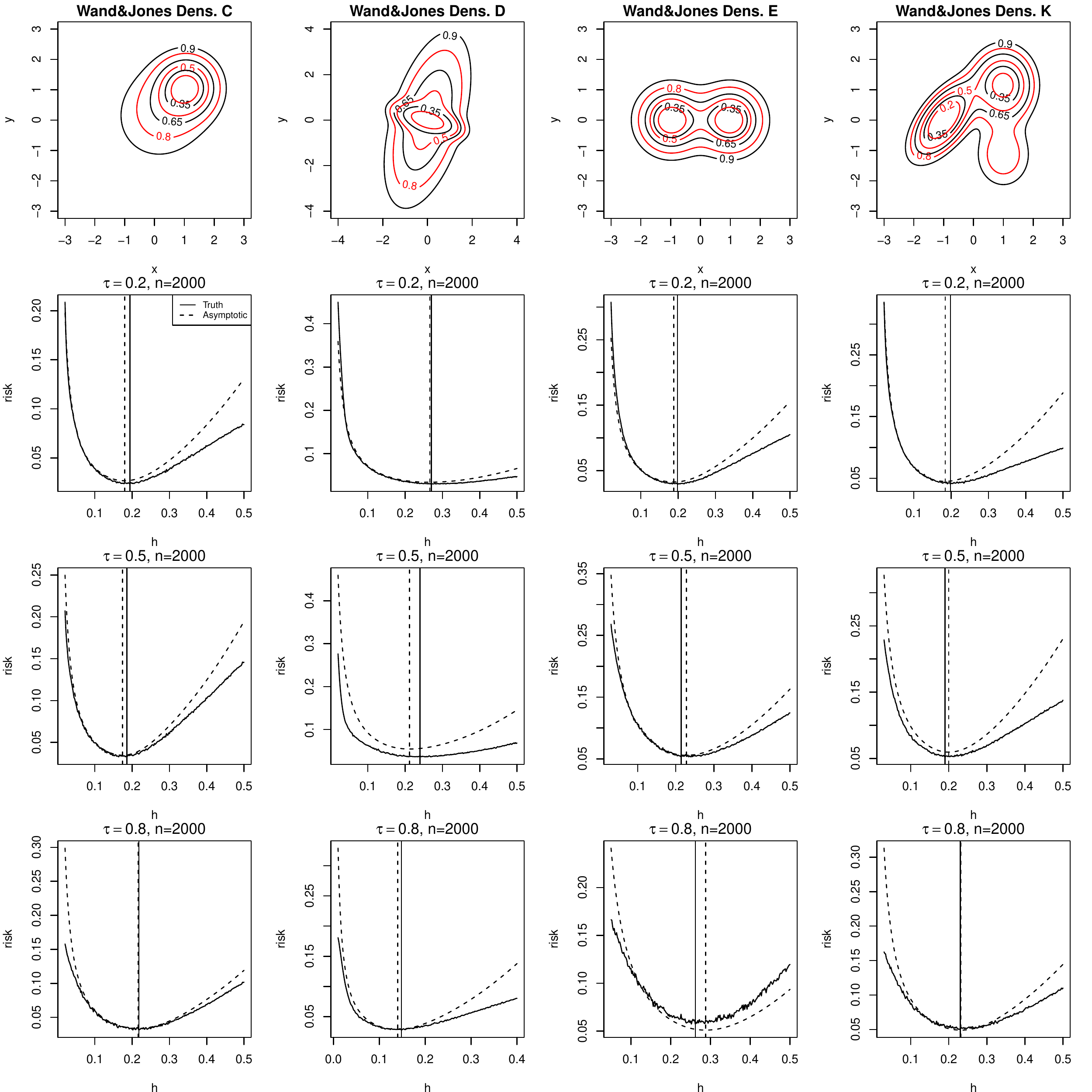}
  \caption{    \label{fig:hdr-risk}Comparison of the simulated true risk function
    $\EE[\mu_{f_0}\{\mc{L}_{\tau}\Delta\widehat{\mc{L}}_{\tau,\bH}\}]$
    with $\text{HDR}(\bH)$ for  four densities in
    \citet{Wand:1993jl}. The panels in the first row are the contour
    plots for  four densities with the contours of interest plotted
    in red color. The panels in the rest of the rows  are the
    comparison plots for the simulated true risk (solid line) and the
    $\text{HDR}(\bH)$ (dashed line) corresponding to the density
    at the top of column for $\tau=0.2,0.5,0.8$. The positions of the solid vertical
    line and the dashed line stand for the optimal bandwidths obtained
    from the simulated true risk and the asymptotic approximation
    respectively  over the restricted class $\mc{S}_1$. The sample size
    for all the cases is 2000.}
\end{figure}
We  ran a  simulation study to compare the performance of our
bandwidth selection method with LSCV for all the 12 densities  in
\citet{Wand:1993jl} and for  density~\eqref{eq:sharpmode-density}. For
each density function, 250  Monte Carlo samples with
2000 observations were generated. For each sample, we estimated the 0.2,
0.5, 0.8 HDR with bandwidth matrices chosen by our method and LSCV
respectively. The HDR error
$\mu_{f_0}\{\mc{L}_{\tau}\Delta\widehat{\mc{L}}_{\tau,\bH}\}$ was
calculated for each method in each replication. 
Figure~\ref{fig:DensityS} shows the plot of the
estimation errors generated by the two methods for
density~\eqref{eq:sharpmode-density}. Figure~\ref{fig:DensityS-contour} shows the boundaries of the
estimated HDR by HDR bandwidth and by the LSCV bandwidth selector from one of the
simulated samples.
We can see for $\tau=0.2,0.5$, the performance of HDR bandwidth selector outperformed  LSCV
bandwidth selector greatly for each simulated instance. For
$\tau=0.8$,  the HDR bandwidth performed slightly less well than the LSCV
bandwidth on average.
One hypothesis for why our method suffers when $\tau=.8$ is that
Assumption~\ref{assm:DA-hdr}
requires that $\| \grad f_0 \| > 0$ in a neighborhood of the HDR.  However, when $\tau=.8$, $f_0$ is close to having gradient zero on the true HDR which is close to the density mode.

\begin{figure}
  \centering
  \includegraphics[width=\textwidth]{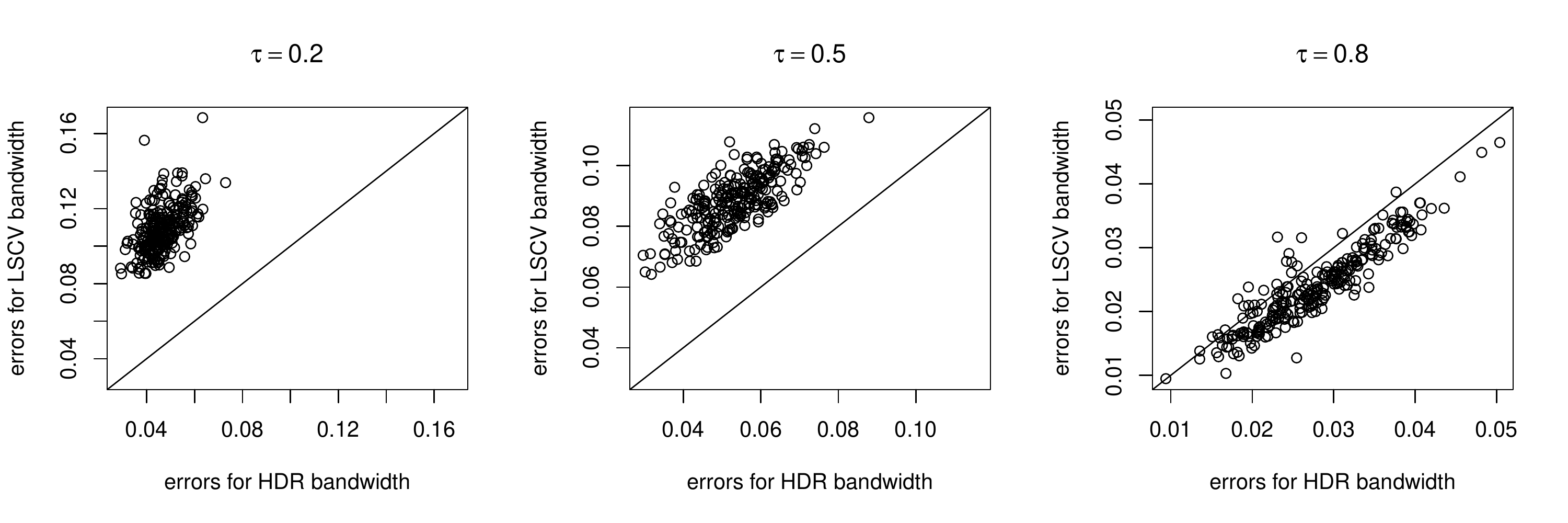}
  \caption{Plot of simulated errors generated by HDR-tailored bandwidth
    and LSCV for the sharp mode density \eqref{eq:sharpmode-density}. The horizontal axis stands for errors of HDR bandwidth
    and vertical axis stands for errors of LSCV bandwidth. }
  \label{fig:DensityS}
\end{figure}
\begin{figure}
  \centering
  \includegraphics[width=\textwidth]{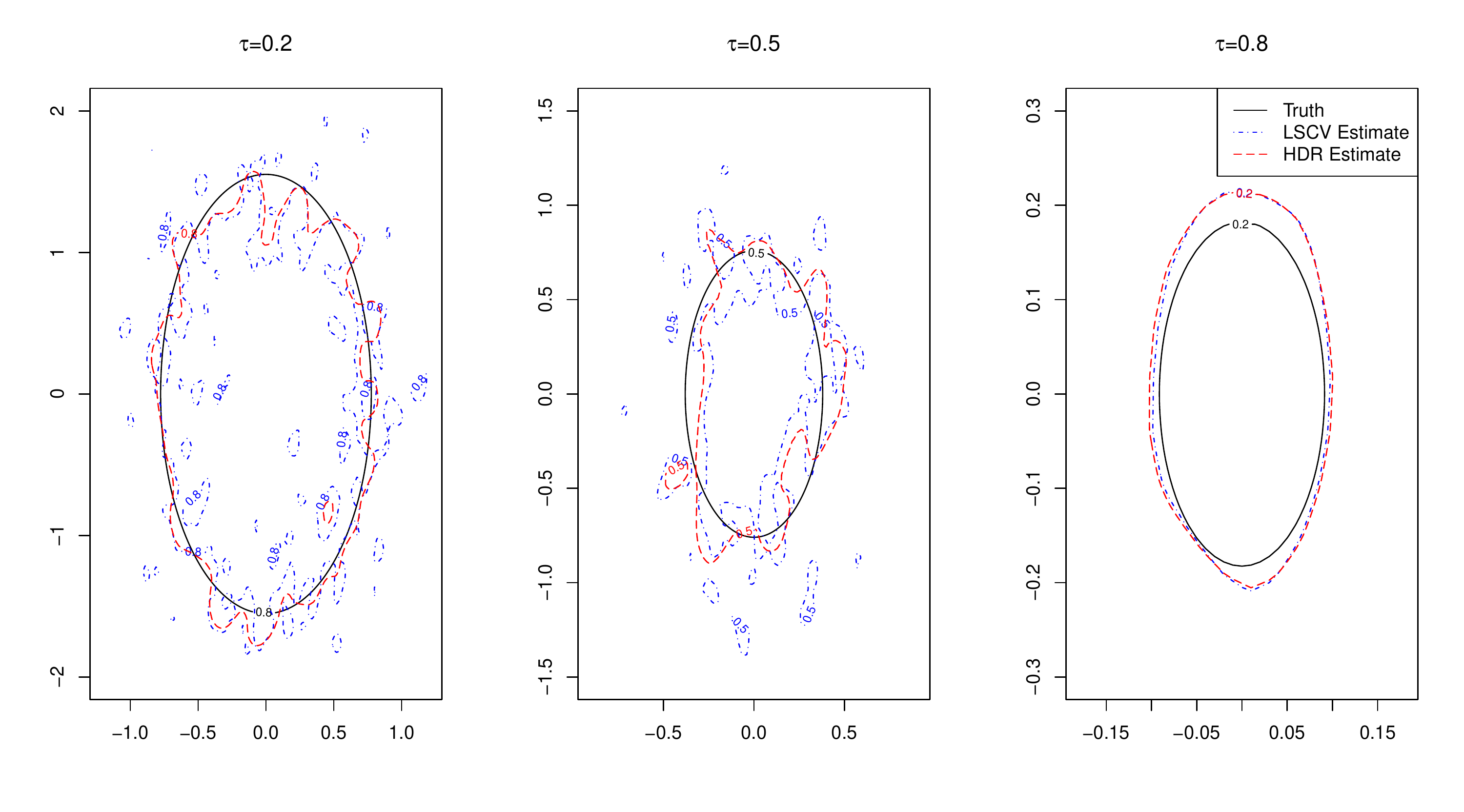}
  \caption{Plot of boundaries of true HDR, HDR estimated by HDR
    bandwidth and HDR estimated by LSCV bandwidth from one simulated
    sample with 2000 observations. The three panels correspond to
    $\tau=0.2,0.5,0.8$ respectively.}
  \label{fig:DensityS-contour}
\end{figure}
It is worth noticing in Figure~\ref{fig:DensityS-contour} that the HDR
estimated by our method discovers the true underlying topological
structure of the density, while the HDR estimated by LSCV does a very poor job of revealing the topological structure when $\tau=.2$ or $.5$ (the LSCV estimates have many spurious separate connected components rather than a single one).

Applying the Wilcoxon signed rank test to the simulated paired errors
genererated by our HDR bandwidth  and LSCV bandwidth
showed that for $\tau=0.2$, our method outperformed LSCV for 12 out of
13
density functions; for $\tau=0.5$, our method did better for 8 out of 13
density functions; for $\tau=0.8$, our method did better in 8 out of 13
density functions.

Note that for any given fixed density, it is likely to be the case for some HDR  that the MISE-optimal bandwidth and the HDR-optimal bandwidth will approximately coincide.  Thus we may not expect our method to be better than LSCV for all densities and levels simultaneously.
Of course, in practice one does not know whether LSCV will work well for the $\tau$ value one is interested in.
Our HDR method appears to work well for lower $\tau$ values, which are the useful values in many applications of HDR estimation.  For example in novelty
detection, the value of $\tau$ equals  the probability of type-I error
which is often set to be  $0.05$ or $0.1$; in clustering analysis, $\tau$
corresponds to  fraction of the data that will be discarded
during analysis and
is also set to be a value close to $0$.
As mentioned in the previous paragraph, this may be related to the assumption that $\| \grad f_0 \| > 0$ on the HDR boundary.  Relaxing this assumption is an important direction for future work, but seems likely to involve somewhat different approximations than the ones used in this paper.
\begin{mylongform}
  \begin{longform}
    \begin{table}[H]
      \centering
      \begin{tabular}{l|cc|cc|cc|}
        \hline
        &\multicolumn{2}{|c|}{$\tau=0.2$}&\multicolumn{2}{|c|}{$\tau=0.5$}&\multicolumn{2}{|c|}{$\tau=0.8$}\\
        \hline
        Density&Test statistic&p-value&Test statistic&p-value&Test
                                                               statistic&p-value\\
        \hline
        Density A&27743&<2.2e-16&25356&1.816e-14&26608&<2.2e-16\\
        Density B&27743&<2.2e-16&25297&<2.2e-16&26361&<2.2e-16\\
        Density C&23200&<2.2e-16&25419&<2.2e-16&20725&5.388e-06\\
        Density D&31349&<2.2e-16&29159&<2.2e-16&799&1\\
        Density E&28630&<2.2.e-16&22430&1.923e-09&24220&4.503e-14\\
        Density F&24756&1.159e-15&3223&1&6433&1\\
        Density G&29997&<2.2e-16&14861&0.765&16265&0.3071\\
        Density H&25968&<2.2e-16&11889&0.9996&10535&1\\
        Density I&8419&1&3713&1&3643&1\\
        Density J&29488&<2.2e-16&18411&0.008676&28905&<2.2e-16\\
        Density K&24555&4.689e-15&21411&2.861e-07&19746&0.0001958\\
        Density L&24658&2.3e-15&15659&0.5101&18997&0.001919\\
        \hline
      \end{tabular}
      \caption{Summary of Wilcoxon tests for the errors generatated by
        LSCV bandwidth and HDR bandwidth. The alternative hypothesis is
        the errors generated by LSCV bandwidth are large. Tests were done
        for 250 errors simulated from each of the 12 densities in
        \citet{Wand:1993jl} at $\tau=0.2,0.5,0.8$. Test statistics and the
        corresponding p-values are summarized in the table. Use
        $\widehat{f}_{n,\bH_0}$ to estimate $\beta_{\tau}$}
    \end{table}

    \begin{table}[H]
      \centering
      \begin{tabular}{l|cc|cc|cc|}
        \hline
        &\multicolumn{2}{|c|}{$\tau=0.5$}&\multicolumn{2}{|c|}{$\tau=0.2$}&\multicolumn{2}{|c|}{$\tau=0.8$}\\
        \hline
        Density&Test statistic&p-value&Test statistic&p-value&Test
                                                               statistic&p-value\\
        \hline
        Density A&25385&<2.2e-16&29238&<2.2e-16&26520&<2.2e-16\\
        Density B&25690&<2.2e-16&29290&<2.2e-16&23990&2.027e-13\\
        Density C&27574&<2.2e-16&24585&3.814e-15&21568&1.392e-07\\
        Density D&29298&<2.2e-16&31363&<2.2e-16&1486&1\\
        Density E&21824&4.133e-08&29934&2.2e-16&22241&5.154e-09\\
        Density F&5344&1&25023&<2.2e-16&9866&1\\
        Density G&15348&0.6168&28737&<2.2e-16&17723&0.0377\\
        Density H&16244&0.3136&25987&<2.2e-16&11732&1\\
        Density I&4573&1&9796&1&4485&1\\
        Density J&20982&<1.868e-06&29863&<2.2e-16&29830&<2.2e-16\\
        Density K&21595&1.227e-07&21842&3.788e-08&22701&4.463e-10\\
        Density L&16854&0.1542&24077&1.153e-13&20861&3.094e-06\\
        \hline
      \end{tabular}
      \caption{Summary of Wilcoxon tests for the errors generatated by
        LSCV bandwidth and HDR bandwidth. The alternative hypothesis is
        the errors generated by LSCV bandwidth are large. Tests were done
        for 250 errors simulated from each of the 12 densities in
        \citet{Wand:1993jl} at $\tau=0.2,0.5,0.8$. Test statistics and the
        corresponding p-values are summarized in the table. Use
        $\widehat{f}_{n,\bH_1}$ to estimate $\beta_{\tau}$}
    \end{table}
  \end{longform}
\end{mylongform}

\subsection{Real data analysis}
We now discuss two real datasets. The Wisconsin Diagnostic Breast Cancer data contains 699 instances of
breast cancer cases with 458 of them being benign instances and 241
being malignant instances. Nine cancer-related features were measured
for each instance. For the Banknote Authentication data, images
were taken of 1372 banknotes, some fake and some genuine. Wavelet
transformation tools were used to extract  four descriptive
features of the images.
For both  datasets, we reduced the original features to the first two
principal components.  We apply our method to perform novelty
detection for the two  data sets. Novelty detection  is like  a
classification problem where only the ``normal'' class is observed in
the training data. 
Then, for a new data point $\bx_{\text{new}}$, we want to test the
null hypothesis
$H_0:\bx_{\text{new}}\text{ is a normal point}$ (or, alternatively, to classify 
$\bx_{\text{new}}$  as ``normal'' or ``anomalous'').
For level set (HDR) based novelty detection, we can consider an oracle decision rule, or acceptance region, 
$A := \{\bx:f_0(\bx)\ge c\}$ (based on knowing $f_0$);
if $f_0(\bx_{\text{new}})\in A$, we accept the null hypothesis, and we reject
otherwise.
For the breast cancer data, ``normal'' means healthy, and for the banknote data, ``normal'' means genuine.
If we take $c=f_{\tau}$, then the oracle decision rule
will have type-I error, or False Positive Rate (FPR), of $\tau$ (under a regularity condition). Additionally, under
regularity conditions, $A$ has the minimum volume of any acceptance rule with FPR of $\tau$, since HDR's are minimum volume
sets \citep{garcia2003level}.  This property is beneficial for controlling the type-II error rate, or False Negative Rate  (although the actual False Negative Rate depends on the unknown ``anomaly'' distribution).

In this section, for each of the two data sets we use a KDE with our bandwidth selection procedure to estimate an HDR based on the ``normal'' class data and use the estimated HDR to perform classification.  We delete the observations with missing values for any covariates and randomly split the data set into two parts, training data and testing data. For the Wisconsin Breast Cancer data, 345 benign instances are contained in the training data and 200 (with half being benign and another half being malignant) are contained in the testing data. For the Banknote Authentication data, 400 genuine instances are contained in the training data and again, 200 (with half being genuine and another half being fake) are contained in the testing data. We estimate the $90\%$ HDR using our method based on the training data.  The first row of Figure~\ref{fig:cancer-bank} shows the plot of the data and the boundaries of the $90\%$ HDR which are the decision boundaries for the two classification problems.  The asymptotic 
FPR
in these two classification problems is $\tau=0.1$.  For the Wisconsin Breast Cancer data, on the test data, the observed FPR is $0.09$ and the 
True Positive Rate (TPR)
is $0.99$. For the Banknote Authentication data, the observed 
FPR is $0.04$, and the observed
TPR
is $0.61$.
We also generated full ROC curves for the two datasets which are shown in the second row of Figure~\ref{fig:cancer-bank}.
The ROC curves are based on $30$ different splits of the data into training and test sets (with the reported FPR and TPR given by the averages over the 30 test sets).
The ROC curve clearly shows that the Wisconsin Breast Cancer data is an example where HDR-based anomaly detection is highly effective.  The Banknote data is not as easy for our method; it may be the case that using an HDR based on all four variables improves the classification performance.  We leave the very interesting question of how best to combine HDR-based classification with dimension reduction for future work.
\begin{figure}
  \centering
  \includegraphics[width=\textwidth]{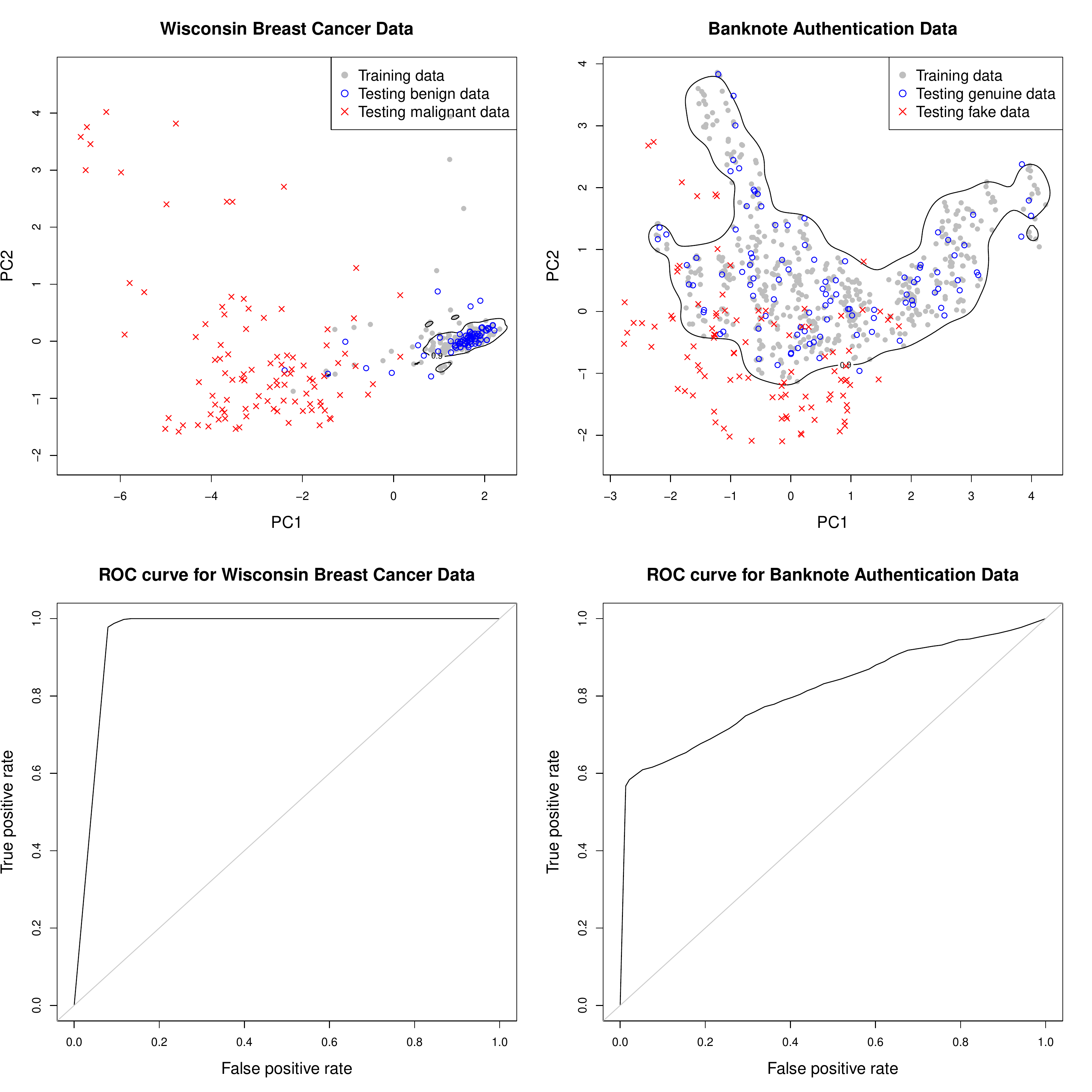}
  \caption{Plot of data and boundary of estimated $90\%$ HDR for the
    Wisconsin Diagnostic Breast Cancer Data and Banknote Authentication Data. Solid dots correspond to
    training data, circles are testing data of normal  instances and
    crosses are testing data of anomaly instances. The two panels in the second row are the corresponding ROC curves for the two classification problems.}
  \label{fig:cancer-bank}
\end{figure}



\section{Discussion}
\label{sec:conclusion}
In this paper, we derive asymptotic expansions of the symmetric risk for LS estimation and HDR estimation based on kernel density estimators. We  provide an efficient bandwidth selection procedure using  a plug-in strategy. We also study by theory and by simulation the  performance of our bandwidth selector. Simulation studies show that both our asymptotic expansion and our bandwidth selector are effective tools.  The two asymptotic risk approximations we provide may also be useful in the analysis of other procedures, developed in future work, for doing LS or HDR bandwidth selection.

As discussed in the Introduction,
the interesting paper  \cite{Qiao:2017wq}  also considers
problems of bandwidth selection for KDE's  via minimizing
asymptotic expansions of risk functions that are based on loss functions
related to level sets.
\cite{Qiao:2017wq} does not consider HDR estimation.   \cite{Qiao:2017wq} does consider the LS  estimation problem.
Our Theorem 2.1 is similar to 
\cite{Qiao:2017wq}'s Corollary~3.1;
both results consider the LS estimation setting, and give risk expansions based on loss functions that are given by integrating the symmetric set differences against $f_0$ (or against something similar).  Our theorem requires only that $f_0$ have two continuous derivatives in a neighborhood of $\beta(c)$ (which we believe to be approximately the weakest possible conditions), whereas 
\cite{Qiao:2017wq} requires four continuous derivatives.  On the other hand, 
\cite{Qiao:2017wq} allows for using higher order kernels if one has higher order smoothness of $f_0$.
While \cite{Qiao:2017wq}'s Corollary~3.1 studies the same risk function approximation, $\LS( \cdot)$, that we study 
in our Theorem~2.1,
\cite{Qiao:2017wq} does not present any algorithm for minimizing $\LS(\cdot)$ and thus presents no simulations related to $\LS(\cdot)$.  
Rather, \cite{Qiao:2017wq} focuses more attention on a different risk function (the ``excess risk'') approximation that allows for an analytic solution, at least when $d=2$.



There are many interesting avenues for extending the work done in the present paper. We describe a few here.
\begin{enumerate}[label=(\Alph*).,leftmargin=*]

\item (Regression and classification)
  In the present paper we have considered only the density estimation context, but estimation of level sets of regression functions estimated by kernel-based methods is also interesting, as is consideration of classification problems.

  Regression level set estimation has received less attention than density level set estimation, although it has been studied in some settings;  \cite{Cavalier:1997ef} studies multivariate nonparametric regression level set 
  minimax rates of convergence.

  One method for classification is to estimate densities for different classes and then classify a point by the class density having highest value at the point.  In that case, rather than estimating a level set of one density, one is estimating the $0$ level set of a difference of two densities.
  \citet[page 1110]{Mason:2009dk} discuss this approach to classification.
  In the context of an application in flow cytometry,
  \cite{Duong:2009ek} 
  also study estimation of HDR's of density differences (without specifically focusing on classification).  We believe the methods of this paper can be extended to those contexts.



\item (Topological data analysis and critical points)
  Another important avenue of research is to consider modifications of the assumptions under which our approximations hold. Level set estimation is one of the main tools in topological data analysis (TDA).
  Estimation of LS's 
  which have zero gradient (at some points) on the boundary (which is ruled out by our assumptions) is of great interest in TDA, because the topology of level sets can change as the level crosses critical points (points having zero gradient).
  In fact, in the context of using tools based on level set estimates, \citet[Section 5]{Wasserman:2016ua} states that ``the problem of choosing tuning parameters is one of the biggest open challenges in TDA''. Thus, developing tools for bandwidth selection when the gradient is zero would be very useful for TDA.  Unfortunately, at points where the gradient is zero we cannot apply the inverse function theorem which is used in Lemma~\ref{lem:hdr-step1} (implicitly) and by
  several results in Appendix~\ref{app:additional-thms}, so a very different analysis than the one we completed here may be necessary in such cases.
  In general, there are very few theoretical works on level set estimation at levels that contain critical values (points where $\grad f_0$ is $0$). In fact, the only one we know of is
  \cite{Chen:2016vv},
  in which a rate of convergence of $\lambda \lb \mc{L}(c) \Delta \widehat{\mc{L}}_{\bH}(c) \rb$ (where $\lambda$ is Lebesgue measure) is derived.

\item (MCMC level sets)
  The work in this paper is restricted to the case where $\bX_1, \ldots, \bX_n$ are independent.  An important extension is to allow  the $\bX_i$ to be samples from a Markov chain.  It is well known that KDE's often work similarly when the data exhibit weak dependence as when they are independent \citep{Wand:1995kv}.  This would allow our tools for HDR estimation to be used to form credible regions based on Markov chain Monte Carlo output in Bayesian statistical analyses.
  At present,  ad-hoc methods are often used for forming credible regions based on Markov chain Monte Carlo output.

\end{enumerate}

\appendix
\section{Proof of main results}
\label{app:A-main-proofs}
\subsection{Proof of Theorem \ref{thm:hdr}}

First, we observe that
\begin{align*}
  \mu_{f_0}(\mc{L}_{\tau}\Delta
  \hat{\mc{L}}_{\tau,\bH})&=\int_{\RR^d}f_0(\bx)\left|\one_{\lb\ffnH(\bx)\ge\fftaun\rb}-\one_{\lb
                            f_0(\bx)\ge
                            f_{\tau,0}\rb}\right|\,d\bx\\
                          &=\int_{\mc{L}_{\tau}^c}f_0(\bx)\one_{\lb\ffnH(\bx)\ge\fftaun\rb}\,d\bx
                            +\int_{\mc{L}_{\tau}}f_0(\bx)\one_{\lb\ffnH(\bx)<\fftaun\rb}\,d\bx.
\end{align*}
Then by Tonelli's Theorem \citep[Theorem 2.37]{MR1681462}, we have
\begin{equation}
  \label{eq:symrisk}
  \begin{split}
    \bb{E}\ls\mu_{f_0}\{\mc{L}_{\tau}\Delta \widehat{\mc{L}}_{\tau,\bH}\}\rs&=  \int_{\mc{L}_{\tau}^c}f_0(\bx)P\lp\ffnH(\bx)\ge\fftaun\rp\,d\bx \\
    &\quad+\int_{\mc{L}_{\tau}}f_0(\bx)P\lp\ffnH(\bx)<\fftaun\rp\,d\bx.
  \end{split}
\end{equation}
For a density function $f$ on
$\RR^d$, let $f_{\tau}(f):=\inf\{\bs{y}\ge
0:\int_{\RR^d}f(\bx)\one_{\{f(\bx)\ge \bs{y}\}}\,d\bx\le 1-\tau\}$. By
this definition, $f_{\tau,0}\equiv f_{\tau}(f_0)$. The following lemma
bounds the modulus of continuity of $f_{\tau}$ when the difference
between two density functions is sufficiently small.
\begin{lemma}
  \label{lem:hdr-step1}
  Let the assumptions of Theorem~\ref{thm:hdr} hold.
  Let $\tilde{f}$ be another uniformly continuous density function
  on $\RR^d$ and $\tilde{f}_{\tau}\equiv f_{\tau}(\tilde{f})$. Then there exists a constant $C_1\ge 1$ such that for all
  $\varepsilon>0$ sufficiently small, $|\tilde{f}_{\tau}-f_{\tau,0}|\le
  C_1\varepsilon$ whenever $\|\tilde{f}-f_0\|_{\infty}\le
  \varepsilon$.
\end{lemma}

\medskip
\noindent

It is intuitively believable that when the sample size $n$ is sufficiently
large, the values of the two integrals on the right of \eqref{eq:symrisk}
are mostly governed by the integrals over a small neighborhood of
$\beta_\tau$.  To shrink the region of integration, for $\delta>0$, and for a
given level $t>0$, we let
$\beta^{\delta}(t) :=\bigcup_{\bx\in \beta(t)}B(\bx,\delta)$,
and
$\beta^{\delta}_{\tau}\equiv\beta^{\delta}(f_{\tau,0})$.  We also let
\begin{align*}
  \mc{L}_{\delta}(f_{\tau,0}):=\bigcup_{\bx\in
  \mc{L}_{\tau}}B(\bx,\delta)\quad\text{and}\quad
  \mc{L}_{-\delta}(f_{\tau,0}):=\mc{L}(f_{\tau,0})\backslash\beta_{\tau}^{\delta}.
\end{align*}
Then we can shrink the integral region using the following lemma.
\begin{lemma}
  \label{lem:hdr-step3}
  Let the assumptions of Theorem~\ref{thm:hdr} hold.
  Then  for a sequence $\delta_n>0$  converging to 0 such
  that $\lambda_{\max}(\bH)=o(\delta_n)$,  we will have
  \begin{equation}
    \label{eq:step3re}
    \int_{\mc{L}_{\delta_n}(f_{\tau,0})^c}f_0(\bx)P\lp\ffnH(\bx)\ge
    \fftaun\rp\,d\bx+\int_{\mc{L}_{-\delta_n}(f_{\tau,0})}f_0(\bx)P\lp\ffnH(\bx)< \fftaun\rp\,d\bx
  \end{equation}
  is $o(n^{-1})$  as $n\rightarrow \infty$.
\end{lemma}

The definition of $\fftaun$ is simple and straightforward, however
there is no explicit form for this quantity. So we want to seek an asymptotic expansion for $\fftaun$. For a uniformly continuous density $f$ on $\RR^d$ and $y\ge 0$,
we define
\begin{align*}
  \psi(f,y):=\int_{\RR^d}f(\bx)\one_{\{f(\bx)\ge y\}}\,d\bx.
\end{align*}
First, we observe for $\varepsilon>0$ sufficiently small,
\begin{align}
  \label{eq:step4-1}
  \begin{split}
    \MoveEqLeft
    \left|\psi(f_0,f_{\tau,0}+\varepsilon)-\psi(f_0,f_{\tau,0})-\varepsilon\int_{\beta_{\tau}}\frac{f_0(\bx)}{\|\nabla
        f_0(\bx)\|}\,d\mc{H}(\bx)\right|\\
    &=\left|\int_{\RR^d}f_0(\bx)\one_{\{f_{\tau,0}\le f_0(\bx)\le f_{\tau,0}+\varepsilon\}}\,d\bx-\varepsilon\int_{\beta_{\tau}}\frac{f_0(\bx)}{\|\nabla
        f_0(\bx)\|}\,d\mc{H}(\bx)\right|=O(\varepsilon^2),
  \end{split}
\end{align}
as $\varepsilon\searrow 0$,
where the last line comes from a similar  argument of
\eqref{eq:step1-leveltohauss} and
\eqref{eq:step1-intdiff}. A similar argument shows the same result
when $\varepsilon \nearrow 0$.
Next, we look at
\begin{align}
  \label{eq:hdr-step4-exp}
  \begin{split}
    \MoveEqLeft   \left| \psi(\tilde{f},\tilde{f}_{\tau}) -  \psi(f_0,\tilde{f}_{\tau}) - f_{\tau,0} \int_{\beta_{\tau}} \frac{g}{\|\nabla
        f_0 \|}\,d\mc{H}-\int_{\mc{L}_{\tau}}g \, d\lambda \right|\\
    & =\left|\int
      \tilde{f}\one_{\{\tilde{f}\ge\tilde{f}_{\tau}\}}\,d\lambda -\int
      f_0\one_{\{f_0\ge\tilde{f}_{\tau}\}}\,d\lambda  -
      f_{\tau,0}\int_{\beta_{\tau}}\frac{g}{\|\nabla
        f_0 \|}\,d\mc{H} - \int_{\mc{L}_{\tau}}g \, d\lambda \right|\\
    &=\left|\int f_0(\one_{\{\tilde{f}\ge\tilde{f}_{\tau}\}}-\one_{\{f_0\ge\tilde{f}_{\tau}\}})\,d\lambda-
      f_{\tau,0} \int_{\beta_{\tau}}\frac{g}{\|\nabla
        f_0 \|}\,d\mc{H} \right.\\
    &\qquad+\left.\int
      g(\one_{\{\tilde{f}\ge\tilde{f}_{\tau}\}}-\one_{\{f_0\ge
        f_{\tau,0}\}})\,d\lambda \right|,
  \end{split}
\end{align}
where $g(\bx)=\tilde{f}(\bx)-f_0(\bx)$. For the first integral  on the last line, since
$\one_{\{\tilde{f}\ge\tilde{f}_{\tau}\}}-\one_{\{f_0\ge\tilde{f}_{\tau}\}}\ne
0$ indicates that $\tilde{f}(\bx)\ge
\tilde{f}_{\tau},f_0(\bx)<\tilde{f}_{\tau}$ or
$\tilde{f}(\bx)<\tilde{f}_{\tau},f_0(\bx)\ge \tilde{f}_{\tau}$, we
have $
f_0(\bx)\in[\tilde{f}_{\tau}-|g(\bx)|,\tilde{f}_{\tau}+|g(\bx)|]$. Combining   \eqref{eq:hdr-step4-f0-bound} with our result in Lemma~\ref{lem:hdr-step1} yields
\begin{align}
  \label{eq:hdr-step4-f0-bound}
  f_0(\bx)=f_{\tau}+O(\|g\|_{\infty}),
\end{align}
for $\bx \in\{\bs{y}:\tilde{f}(\bs{y})\ge \tilde{f}_{\tau}\}\Delta\{\bs{y}:f_0(\bs{y})\ge \tilde{f}_{\tau}\}$.
Next we need the following lemmas.
\begin{lemma}
  \label{lem:hdr-step4-order}
  Let the assumptions of Theorem~\ref{thm:hdr} hold and the notation be as defined above.
  As $\|g\|_\infty^2 + \| g\|_\infty \| \nabla g \|_\infty \to 0$, we have
  \begin{equation}
    \label{eq:22}
    \int \one_{\{\tilde{f}\ge\tilde{f}_{\tau}\}} - \one_{\{f_0\ge\tilde{f}_{\tau}\}}\, d\lambda
    = \int_{\beta_{\tau}}\frac{g}{\|\nabla f_0 \|}\, d\mc{H} +
    O(\|g\|_\infty^2 +  \| g\|_\infty \| \nabla g \|_\infty).
  \end{equation}
\end{lemma}

\begin{lemma}
  \label{lem:hdr-step4-order-g}
  Let the assumptions of Theorem~\ref{thm:hdr} hold and the notation be
  as defined above.
  As $\|g\|^2_{\infty}+\|g\|_{\infty}\|\nabla g\|_{\infty}\rightarrow 0$, we have
  \begin{align*}
    \int_{\RR^d}g(\bx)
    \left(\one_{\{\tilde{f}(\bx)\ge \tilde{f}_{\tau}\}} - \one_{\{f_0(\bx)\ge f_{\tau}\}}\right)\,d\bx=O(\|g\|^2_{\infty}).
  \end{align*}
\end{lemma}

\noindent Now with Lemma~\ref{lem:hdr-step4-order},
\ref{lem:hdr-step4-order-g} and \eqref{eq:hdr-step4-f0-bound}, we see that
\eqref{eq:hdr-step4-exp} equals $O(\|g\|_\infty^2 +  \| g\|_\infty \| \nabla g
\|_\infty)$.
\begin{mylongform}
  \begin{longform} Because
    \begin{align*}
      \MoveEqLeft \left|\psi(\tilde{f},\tilde{f}_{\tau})-\psi(f_0,\tilde{f}_{\tau})-f_{\tau,0}\int_{\beta_{\tau}}\frac{g}{\|\nabla
      f_0\|}\,d\mc{H}-\int_{\mc{L}_{\tau}}g\,d\lambda\right|\\
   &=\left|\int f(\one_{\{\tilde{f}\ge\tilde{f}_{\tau}\}}-\one_{\{f\ge\tilde{f}_{\tau}\}})\,d\bx-
     f_{\tau,0}\int_{\beta_{\tau}}\frac{g}{\|\nabla
     f_0\|}\,d\mc{H}+\int
     g(\one_{\{\tilde{f}\ge\tilde{f}_{\tau}\}}-\one_{\{f\ge
     f_{\tau}\}})\,d\bx\right|\\
   &=\left|(f_{\tau,0}+O(\|g\|_{\infty})\lp \int_{\beta_{\tau}}\frac{g}{\|\nabla f_0\|}\, d\mc{H} +
     O(\|g\|_\infty^2 +  \| g\|_\infty \| \nabla g
     \|_\infty)\rp\right.\\
   &\qquad\left .- f_{\tau,0}\int_{\beta_{\tau}}\frac{g}{\|\nabla
     f_0\|}\,d\mc{H}+ O(\|g\|_\infty^2 +  \| g\|_\infty \| \nabla g
     \|_\infty)\right|\\
   &= O(\|g\|_\infty^2 +  \| g\|_\infty \| \nabla g
     \|_\infty).
    \end{align*}
  \end{longform}
\end{mylongform}
Note that if $\|\nabla g\|_{\infty}\rightarrow 0$, then
$\psi(\tilde{f},\tilde{f}_{\tau})=1-\tau$. Combining
this with \eqref{eq:step4-1} and the order of
\eqref{eq:hdr-step4-exp}, we have
\begin{equation}
  \label{eq:step4-3}
  \begin{split}
    0&=\psi(\tilde{f},\tilde{f}_{\tau})-\psi(f,f_{\tau,0}) \\
    &=\psi(\tilde{f},\tilde{f}_{\tau})-\psi(f,\tilde{f}_{\tau})+\psi(f,\tilde{f}_{\tau})-\psi(f,f_{\tau,0}) \\
    &=-(\tilde{f}_{\tau}-f_{\tau,0})f_{\tau,0}\int_{\beta_{\tau}}\inv{\|\nabla
      f_0\|}\,d\mc{H}+f_{\tau,0}\int_{\beta_{\tau}}\frac{g}{\|\nabla
      f_0\|}\,d\mc{H} \\
    &\qquad+\int_{\mc{L}_{\tau}}g\,d\bx+ O(\|g\|_\infty^2 +  \| g\|_\infty \| \nabla g
    \|_\infty)
  \end{split}
\end{equation}
as $\|g\|_\infty^2 +  \| g\|_\infty \| \nabla g\|_\infty\rightarrow
0$. We want to apply \eqref{eq:step4-3} with $\tilde{f}=\ffnH$, so
that $g=\ffnH-f_0$. To do this, note by Theorem~\ref{thm:KDE-sup-rate}
\begin{mylongform}
  \begin{longform}
    (and so Assumptions~\ref{assm:HA} (as well as
    Assumptions \ref{assm:DA} and \ref{assm:KA}))
  \end{longform}
\end{mylongform}
that
$\|\ffnH-\bb{E}\ffnH\|_\infty=O_{\text{a.s.}}\lp\sqrt{\frac{\log|\bH|^{-1/2}}{n|\bH|^{1/2}}}\rp
$, $\|\nabla
\ffnH-\bb{E}\grad\ffnH\|_{\infty}=O_{a.s.}\lp\sqrt{\frac{\log|\bH|^{-1/2}}{n|\bH|^{1/2}\lambda_{\min}(\bH)}}\rp$,
by \eqref{eq:ffnH-bias-supnorm},
$\|\bb{E}(\ffnH)-f_0\|_{\infty}=O\lb\lambda_{\max}(\bH)\rb$. We also have
$\|\bb{E}\nabla \ffnH-\nabla f_0\|_{\infty}=O\{\lambda^{1/2}_{\max}(\bH)\}$.
\begin{mylongform}
  \begin{longform}
    Because,
    \begin{align*}
      \bb{E}\nabla \ffnH(\bx)=\bb{E}\nabla K_{\bH}(\bx-X_1)=\nabla
      K_H\ast f_0(\bx)=K_{\bH}*\nabla f_0(\bx),
    \end{align*}
    where $f\ast g(\bx)=\int_{\RR^d}
    f(\bx-\bs{y})g(\bs{y})\,d\bs{y}$ and $K_{\bH}(\bx)=|\bH|^{-1/2}K(\bH^{-1/2}\bx)$. So
    \begin{align*}
      \bb{E}\nabla \ffnH(\bx)&=\int_{\RR^d}K_{\bH}(\bx-\bs{y})\nabla
                               f_0(\bs{y})\,d\bs{y}\\
                             &=\int_{\RR^d}K(\bs{y})\nabla
                               f_0(\bx-\bH^{1/2}\bs{y})\,d\bs{y}.
    \end{align*}
    By Taylor Expansion, $\nabla
    f_0(\bx-\bH^{1/2}\bs{y})=\nabla
    f_0(\bx)-\nabla^2f_0(x-s_{\bs{y}}\bH^{1/2}\bs{y})\bH^{1/2}\bs{y}$,
    $s_{\bs{y}}\in(0,1)$ depends on
    $\bs{y}$. Now we have
    \begin{align*}
      \bb{E}\nabla \ffnH(\bx)&=\nabla
                               f_0(\bx)-
                               \int_{\RR^d}K(\bs{y})\nabla^2
                               f_0(\bx-s_{\bs{y}}\bH^{1/2}\bs{y})H^{1/2}\bs{y}\,d\bs{y}.
    \end{align*}
    By assumption \ref{assm:DA} $\nabla^2f_0$ is
    bounded, there exists $A\in\RR^{
      d}$, such that
    \begin{align*}
      \left|    \int_{\RR^d}K(\bs{y})\nabla^2
      f_0(\bx-s_{\bs{y}}\bH^{1/2}\bs{y})H^{1/2}\bs{y}\,d\bs{y}\right|\le H^{1/2}A.
    \end{align*}
  \end{longform}
\end{mylongform}
Then applying the above results, 
we have
\begin{equation}
  \label{eq:fftaun-ftau}
  \begin{split}
    \fftaun-f_{\tau,0}
    & = 
    w_0
    \left\{ \int_{\beta_{\tau}}\frac{\ffnH(\bx)-f_0(\bx)}{\|\nabla
        f_0(\bx)\|}\,d\mc{H}(\bx)
      + \inv{f_{\tau,0}} \int_{\mc{L}_{\tau}}\ffnH(\bx)-f_0(\bx)\,d\bx\right\}\\
    & \qquad  +O_p\lp
    \frac{\log
      |\bH|^{-1/2}}{n|\bH|^{1/2}\sqrt{\lambda_{\max}(\bH)}}+\lambda_{\max}^{3/2}(\bH)\rp.
  \end{split}
\end{equation}
Note from \eqref{eq:kde-def} and  \eqref{eq:fftaun-ftau}, for fixed $\bx$,
$\ffnH(\bx)-\fftaun$ can be expressed as  the average of i.i.d. random
variables  with a negligible stochastic error term.  This
motivates
us
to use the Berry-Essen Theorem \citep[]{Ferguson:1996bn} to
approximate the two probabilities appearing on the right of
\eqref{eq:step3re}.
In order to do so, we will  need to approximate the mean and variance of $\fftaun$,
which we do in the next lemmas.

\medskip

\begin{lemma}
  \label{lem:ftaun-bias-exp}
  Let the assumptions of Theorem~\ref{thm:hdr} hold and the notation be as
  defined above.
  Then we have
  \begin{equation}
    \label{eq:ftaun-exp}
    \begin{split}
      \bb{E} \fftaun-f_{\tau,0}
      &= w_0
      \left\{V_1(\bH)+V_2(\bH)\right\}+o\lb\tr(\bs{H})\rb,
    \end{split}
  \end{equation}
  as $n\rightarrow \infty$.
\end{lemma}

\medskip
\noindent Recall
$V_1$ and $V_2$ are defined in Theorem~\ref{thm:hdr}.  The next lemma shows $\Var \fftaun$ is negligible compared with other
terms in the expansion.

\medskip

\begin{lemma}
  \label{lem:ftaun-var-exp}
  Let the assumptions of Theorem~\ref{thm:hdr} hold and the notation be as
  defined above.
  Then  $\Var \fftaun=o(n^{-1}|\bH|^{-1/2})$.
\end{lemma}

\medskip
Now  according
to Lemma \ref{lem:hdr-step3} and \eqref{eq:symrisk}, we have
\begin{align*}
  \bb{E}\mu_{f_0}(\mc{L}_{\tau}\Delta \hat{\mc{L}}_{\tau,\bH})
  &=  \int_{\mc{L}_{\tau}^c\backslash \mc{L}_{\delta_n}(f_{\tau})^c}f_0(\bx)P\lp\ffnH(\bx)\ge\fftaun\rp\,d\bx\\
  &\quad+\int_{\mc{L}_{\tau}\backslash
    \mc{L}_{-\delta_n}(f_{\tau,0})}f_0(\bx)P\lp\ffnH(\bx)<\fftaun\rp\,d\bx+o\lp
    n^{-1}\rp\\
  &=\int_{\beta_\tau^{\delta_n}}f_0(\bx)\left|P\lp\ffnH(\bx)<\fftaun\rp\,d\bx-\one_{\{f_0(\bx)<f_{\tau,0}\}}\right|\,d\bx+o\lp
    n^{-1}\rp.
\end{align*}
Then by Lemma~\ref{lem:COV-approx}
when $\delta_n$ is small enough, the dominating term on the last line
above  is equal to
\begin{align}
  \label{eq:shrinked-hauss-symrisk}
  \begin{split}
    &  \int_{\beta_\tau} \int_{-\delta_n}^{\delta_n} f_0(\bx+tu_{\bx}) \left| P\lp\ffnH(\bx+tu_{\bx})<\fftaun\rp-\one_{\{f_0(\bx+tu_{\bx})<f_{\tau,0}\}} \right|
    \,dtd\mc{H}(\bx)\\
    &\qquad + O(\delta_n^2),
  \end{split}
\end{align}
where $u_{\bx} := - \grad f_0(\bx) / \| \grad f_0(\bx) \|$ is the unit outer normal vector of $\beta_\tau$ at $\bx$.
Now for a fixed $\bx\in \beta_{\tau}$, let
$\bx^t=\bx+\frac{t}{\sqrt{n|\bH|^{1/2}}}u_{\bx}$ for
$t\in[-\sqrt{n|\bH|^{1/2}}\delta_n,\sqrt{n|\bH|^{1/2}}\delta_n]$, we see \eqref{eq:shrinked-hauss-symrisk} equals
\begin{align}
  \label{eq:trans-hauss-symrisk}
  &\frac{1}{\sqrt{n|\bH|^{1/2}}}\int_{\beta_{\tau}}\int_{-\sqrt{n|\bH|^{1/2}}\delta_n}^{\sqrt{n|\bH|^{1/2}}\delta_n}f_0\lp\bx^t\rp
    \left|P\lp\ffnH\lp\bx^t\rp<\fftaun\rp-\one_{\{t>0\}}\right|\,dtd\mc{H}(\bx)\\
  &\qquad+O(\delta_n^2).
\end{align}
By Taylor Expansion, we have
\begin{align*}
  f_0\lp\bx+\frac{t}{\sqrt{n|\bH|^{1/2}}}u_{\bx}\rp=f_0(\bx)+\nabla f_0\lp\bx+\frac{st}{\sqrt{n|\bH|^{1/2}}}u_{\bx}\rp'\frac{t}{\sqrt{n|\bH|^{1/2}}}u_{\bx},
\end{align*}
for some $s\in[0,1]$.
Since by Assumption~\ref{assm:DA-hdr},  $f_0$ has bounded first
derivatives,  we see the dominating  term in
\eqref{eq:trans-hauss-symrisk} equals
\begin{align}
  \label{eq:LebtoHauss}
  \begin{split}
    \frac{f_{\tau,0}}{\sqrt{n|\bH|^{1/2}}}\int_{\beta_\tau}\int_{-\sqrt{n|\bH|^{1/2}}\delta_n}^{\sqrt{n|\bH|^{1/2}}\delta_n}\left|P\lp\ffnH(\bx^t)<\fftaun\rp-\one_{\{t>0\}}\right|\,dtd\mc{H}(\bx)+O(\delta_n^2),
  \end{split}
\end{align}
as $n\rightarrow \infty$.
\begin{mylongform}
  \begin{longform}
    Because
    \begin{align*}
      \frac{1}{\sqrt{n|\bH|^{1/2}}}\int_{\beta_\tau}\int_{-\sqrt{n|\bH|^{1/2}}\delta_n}^{\sqrt{n|\bH|^{1/2}}\delta_n}\frac{t}{\sqrt{n|\bH|^{1/2}}}\,dt=O(\delta_n^2).
    \end{align*}
  \end{longform}
\end{mylongform}
\noindent  We can  further shrink  the region of interest by the following
lemma.
\begin{lemma}
  \label{lem:hdr-step5}
  Let the assumptions of Theorem~\ref{thm:hdr} hold and the notation be as
  defined above.
  Then for $n$
  sufficiently large, $\bb{E}\{\ffnH(\bx^t)-\fftaun\}$ is a strictly monotone function of
  $t\in[-\sqrt{n|\bH|^{1/2}}\delta_n,\sqrt{n|\bH|^{1/2}}\delta_n]$,
  with a unique zero $t_{\bx}^{\ast}$. For  a sequence $t_n$ diverging to infinity and
  $t_n=O(\sqrt{n|\bH|^{1/2}}\delta_n)$, let
  $$I_{\bx}^n=[-\sqrt{n|\bH|^{1/2}}\delta_n,\sqrt{n|\bH|^{1/2}}\delta_n]\backslash
  [t^{\ast}_{\bx}-t_n,t^{\ast}_{\bx}+t_n].$$ We have
  \begin{align}
    \label{eq:shrink}
    \int_{\beta_\tau}\int_{I_{\bx}^n}|P(\ffnH(\bx^t)<\fftaun)-\one_{\{t>0\}}|\,dtd\mc{H}(\bx)  \rightarrow 0
  \end{align}
  as $n\rightarrow \infty$.
\end{lemma}


\medskip
\noindent
To complete the proof of Theorem~\ref{thm:hdr}, by
\eqref{eq:LebtoHauss} and Lemma~\ref{lem:hdr-step5} it suffices to
show that there exists a sequence $t_n$ diverging to infinity slowly such
that
\begin{align*}
  \MoveEqLeft  \frac{f_{\tau,0}}{\sqrt{n|\bH|^{1/2}}} \int_{\beta_\tau}\int_{t^{\ast}_{\bx}-t_n}^{t^{\ast}_{\bx}+t_n}|P(\ffnH(\bx^t)<\fftaun)-\one_{\{t<0\}}|\,dtd\mc{H}(\bx)\\
& = 
  \HDR(\bH)
  + o\lb  (n|\bH|^{1/2})^{-1/2} +\tr(\bH)\rb.
\end{align*}
For $i=1,2,\ldots, n$, let $Z_{ni}(\bx)=K_{\bH}(\bx- \bX_i)$ and
$\bar{Y}_n=n^{-1}\sum_{i=1}^nY_{ni}$, where
\begin{align*}
  Y_{ni}&=Z_{ni}(\bx^t)-f_{\tau,0}-\left\{\int_{\beta_\tau}\inv{\|\nabla
          f_0\|}\,d\mc{H}\right\}^{-1}\left\{\int_{\beta_\tau}\frac{Z_{ni}(\bx)-f_0(\bx)}{\|\nabla
          f_0(\bx)\|}\,d\mc{H}(\bx)\right.\nonumber\\
        &\qquad\left.+\inv{f_{\tau,0}}\int_{\mc{L}_{\tau}}Z_{ni}(\bx)-f_0(\bx)\,d\bx\right\}.
\end{align*}
Then by \eqref{eq:step4-3} and \eqref{eq:fftaun-ftau}, we can write
$\ffnH(\bx^t)-\fftaun=\bar{Y}_n+R_n$, where
$R_n-\bb{E}(R_n)=o_p\lp\frac{1}{\sqrt{n|\bH|^{1/2}}}\rp$.
\begin{mylongform}
  \begin{longform}
    Because by  \eqref{eq:step4-3} we have
    \begin{align*}
      R_n=\lp \fftaun-f_{\tau,0}\rp
      \int_{\beta_{\tau}}\frac{1}{\|\nabla
      f_0\|}\,d\mc{H}-\int_{\beta_{\tau}}\frac{\ffnH-f_0}{\|\nabla f_0\|}\,d\mc{H}-\inv{f_{\tau,0}}\int_{\mc{L}_{\tau}}\lp\ffnH-f_0\rp\,d\lambda,
    \end{align*}
    and from Lemma~\ref{lem:ftaun-var-exp} we know
    $\Var\fftaun=o(n^{-1}|\bH|^{-1/2})$,
    $\Var(\int_{\beta_{\tau}}\frac{\ffnH-f_0}{\|\nabla
      f_0\|}\,d\mc{H})=o(n^{-1}|\bH|^{-1/2})$ and
    $\Var(\int_{\mc{L}_{\tau}}\lp\ffnH-f_0\rp\,d\lambda=o(n^{-1}|\bH|^{-1/2})$,
    so $\Var(R_n)=o(n^{-1}|\bH|^{-1/2})$. Then for any $\eta>0$,
    \begin{align*}
      P\lp|R_n-\bb{E}(R_n)|\sqrt{n|\bH|^{1/2}}>\eta\rp\le
      \Var(R_n)n|\bH|^{1/2}\eta^2\rightarrow 0,
    \end{align*}
    which suggests $R_n-\bb{E}(R_n)=o_p(\frac{1}{\sqrt{n|\bH|^{1/2}}})$.
  \end{longform}
\end{mylongform}
By Lemma~\ref{lem:ftaun-var-exp}, we know $\Var(\bar{Y}_n)$ is
$O(n^{-1}|\bH|^{-1/2})$ uniformly in $t$ and $\bx$.     Let $t_n$ diverge slowly such
that for fixed $\bx\in \beta_{\tau}$,
\begin{itemize}
\item
  $P\lp\frac{|R_n-\bb{E}(R_n)|}{\Var^{1/2}(\bar{Y_n})}>\inv{t_n^2}\rp\le
  \inv{t_n^2}$ uniformly for
  $t\in[t_{\bx}^{\ast}-t_n,t_{\bx}^{\ast}+t_n]$.
  \begin{mylongform}
    \begin{longform}
      Because we can let $t_n$ diverge slowly such that $R_n-\bb{E}(R_n)=o_p(\frac{t_n^2}{\sqrt{n|\bH|^{1/2}}})$.
    \end{longform}
  \end{mylongform}
\item $\bb{E}(\bar{Y}_n+R_n)=\lb\frac{t}{\sqrt{n|\bH|^{1/2}}}\|\nabla
  f_0(\bx)\|+D_1(\bx,\bH)-D_2(\bx,\bH)\rb\lb 1+o(t_n^{-2})\rb$,
  uniformly for
  $t\in[t_{\bx}^{\ast}-t_n,t_{\bx}^{\ast}+t_n]$ and $\bx\in
  \beta_{\tau}$,
  by Assumption~\ref{assm:DA-hdr} part \ref{assm:DA:item3}.
  \begin{mylongform}
    \begin{longform}
      \begin{align*}
        \MoveEqLeft \bb{E}(\bar{Y}_n+R_n)\\
     &=
       f_0(\bx^t)+\inv{2}\mu_2(K)\tr(\bH\nabla^2f_0(\bx^t))-f_{\tau,0}\\
     &\quad-\left\{\int_{\beta_\tau}\inv{\|\nabla
       f_0(\bx)\|}\,d\mc{H}(\bx)\right\}^{-1}\left\{\int_{\beta_\tau}\frac{\mu_2(K)\tr\{\bs{H}\nabla^2
       f_0(\bx)\}}{2\|\nabla f_0(\bx)\|}\,d\mc{H}(\bx)\right.\\
     &\qquad\left.+\inv{f_{\tau,0}}\int_{\mc{L}_{\tau}}\frac{1}{2}\mu_2(K)\tr\{\bs{H}\nabla^2
       f_0(\bx)\}\,d\bx\right\}+o\lb\tr(\bH)\rb ,
      \end{align*}
      uniformly in $t$ and $\bx$.  By Taylor Expansion we know
      \begin{align*}
        f_0(\bx^t)&=f_0(\bx) +
                    \frac{t}{\sqrt{n|\bH|^{1/2}}}\| \grad  f_0(\bx)\|+O\lp\frac{t_n^2}{n|\bH|^{1/2}}\rp,
      \end{align*}
      uniformly in $t$ and $\bx$ because by Assumption~\ref{assm:DA-hdr} we
      have bounded third derivative around $\beta_{\tau}$.
      And
      $\tr(\bH\nabla^2f_0(\bx^t))=\tr(\bH\nabla^2f_0(\bx))+o(\tr(\bH))$
      uniformly in $\bx$ and $t$. Then
      \begin{align*}
        \MoveEqLeft    \bb{E}(\bar{Y}_n+R_n)\\
     &=\frac{t}{\sqrt{n|\bH|^{1/2}}}\|f_0(\bx)\|+O\lp\frac{t_n^2}{n|\bH|^{1/2}}\rp\\
     &\qquad+D_1(\bx,\bH)-D_2(\bx,\bH)+o(\tr(\bH)),\\
     &=\lb \frac{t}{\sqrt{n|\bH|^{1/2}}}\|f_0(\bx)\|+D_1(\bx,\bH)-D_2(\bx,\bH)\rb(1+o(t_n^{-2}))
      \end{align*}
      uniformly in $t$ and $\bx$.
    \end{longform}
  \end{mylongform}
\item $n|\bH|^{1/2}\Var \bar{Y}_n=R(K)f_{\tau,0}+o(t_n^{-2})$
  uniformly for $t\in[t_{\bx}^{\ast}-t_n,t_{\bx}^{\ast}+t_n]$ and
  $\bx \in \beta_{\tau}$.
  \begin{mylongform}
    \begin{longform}
      Since $\Var
      \ffnH(\bx^t)$ is
      $\frac{f_0(\bx^t)R(K)}{n|\bH|^{1/2}}+o(n^{-1}|\bH|^{-1/2})$ uniformly
      in $t$ and $\bx$, and $\Var \fftaun=o(n^{-1}|\bH|^{-1/2})$, $\Var(\bar{Y}_n)=
      \frac{f_0(\bx^t)R(K)}{n|\bH|^{1/2}}+o(n^{-1}|\bH|^{-1/2})$ uniformly
      for $t$ and $\bx$. Besides, we
      have
      \begin{align*}
        f_0(\bx^t)=f_0(\bx)+\nabla f_0(\bx+\frac{st}{\sqrt{n|\bH|^{1/2}}})\frac{t}{\sqrt{n|\bH|^{1/2}}}=f_{\tau,0}+O(\frac{t_n}{\sqrt{n|\bH|^{1/2}}}),
      \end{align*}
      uniformly. So $n|\bH|^{1/2}\Var \bar{Y}_n=R(K)f_{\tau,0}+o(t_n^{-2})$.
    \end{longform}
  \end{mylongform}
\end{itemize}
Then
\begin{align*}
  \MoveEqLeft P(\ffnH(\bx^t)<\fftaun)-\Phi\lp A_{\bx}t+C_{\bx}(\bH)\rp\\
&=P(\bar{Y}_n+R_n-\bb{E}(\bar{Y}_n+R_n)<-\bb{E}(\bar{Y}_n+R_n))-\Phi\lp A_{\bx}t+C_{\bx}(\bH)\rp\\
&\le
  P\lp\frac{|R_n-\bb{E}(R_n)|}{\Var^{1/2}(\bar{Y}_n)}>\inv{t^2_n}\rp+P\lp\frac{\bar{Y}_n-\bb{E}(\bar{Y_n})}{\Var^{1/2}(\bar{Y}_n)}\le
  \frac{-\bb{E}(\bar{Y}_n+R_n)}{\Var^{1/2}(\bar{Y}_n)}+\inv{t_n^2}\rp\\
&\qquad -\Phi\lp A_{\bx}t+C_{\bx}(\bH)\rp\\
&=O\lp\inv{t_n^2}\rp+P\lp\frac{\bar{Y}_n-\bb{E}(\bar{Y_n})}{\Var^{1/2}(\bar{Y}_n)}\le
  \frac{-\bb{E}(\bar{Y}_n+R_n)}{\Var^{1/2}(\bar{Y}_n)}+\inv{t_n^2}\rp-\Phi\lp
  A_{\bx}t+C_{\bx}(\bH)\rp.
\end{align*}
Applying the Berry-Esseen theorem \citep[]{Ferguson:1996bn} to the
last two terms %
on the last line
yields
\begin{align*}
  \MoveEqLeft  \left|P\lp\frac{\bar{Y}_n-\bb{E}(\bar{Y_n})}{\Var^{1/2}(\bar{Y}_n)}\le
  \frac{-\bb{E}(\bar{Y}_n+R_n)}{\Var^{1/2}(\bar{Y}_n)}+\inv{t_n^2}\rp-\Phi\lp
  \frac{-\bb{E}(\bar{Y}_n+R_n)}{\Var^{1/2}(\bar{Y}_n)}+\inv{t_n^2}\rp\right|
  \\&  \le \frac{C\bb{E}|Y_{ni}|^3}{\Var^{3/2}(Y_{ni})\sqrt{n}}.
\end{align*}
Now since
$\Var(\bar{Y}_n)=R(K)f_{\tau,0}/(n|\bH|^{1/2})+o(n^{-1}|\bH|^{-1/2})$
uniformly,
$\Var(Y_{ni})=R(K)f_{\tau,0}/(|\bH|^{1/2})+o(|\bH|^{-1/2})$. And it
can be shown that $\bb{E}|Y_{ni}|^3=O(|\bH|^{-1})$,
so we further have
\begin{align*}
  \MoveEqLeft  \left|P\lp\frac{\bar{Y}_n-\bb{E}(\bar{Y_n})}{\Var^{1/2}(\bar{Y}_n)}\le
  \frac{-\bb{E}(\bar{Y}_n+R_n)}{\Var^{1/2}(\bar{Y}_n)}+\inv{t_n^2}\rp-\Phi\lp
  \frac{-\bb{E}(\bar{Y}_n+R_n)}{\Var^{1/2}(\bar{Y}_n)}+\inv{t_n^2}\rp\right|
  \\&=  O\lp \inv{\sqrt{n|\bH|^{1/2}}}\rp,
\end{align*}
and then
\begin{align*}
  \MoveEqLeft P(\ffnH(\bx^t)<\fftaun)-\Phi\lp
  A_{\bx}t+C_{\bx}(\bH)\rp\\
&\le O\lp\inv{t_n^2}+\inv{\sqrt{n|\bH|^{1/2}}}\rp+\Phi\lp
  \frac{-\bb{E}(\bar{Y}_n+R_n)}{\Var^{1/2}(\bar{Y}_n)}\rp-\Phi\lp
  A_{\bx}t+C_{\bx}(\bH)\rp,
\end{align*}
uniformly in $t$ and $\bx$. A similar argument shows a lower bound of the same order.
\begin{mylongform}
  \begin{longform}
    \begin{align*}
      \MoveEqLeft P(\ffnH(\bx^t)<\fftaun)-\Phi\lp A_{\bx}t+C_{\bx}(\bH)\rp\\
   &=P(\bar{Y}_n+R_n-\bb{E}(\bar{Y}_n+R_n)<-\bb{E}(\bar{Y}_n+R_n))-\Phi\lp A_{\bx}t+C_{\bx}(\bH)\rp\\
   &\ge
     -P\lp\frac{|R_n-\bb{E}(R_n)|}{\Var^{1/2}(\bar{Y}_n)}>\inv{t^2_n}\rp+P\lp\frac{\bar{Y}_n-\bb{E}(\bar{Y_n})}{\Var^{1/2}(\bar{Y}_n)}\le
     \frac{-\bb{E}(\bar{Y}_n+R_n)}{\Var^{1/2}(\bar{Y}_n)}-\inv{t_n^2}\rp\\
   &\qquad -\Phi\lp A_{\bx}t+C_{\bx}(\bH)\rp\\
   &=O\lp\inv{t_n^2}\rp+P\lp\frac{\bar{Y}_n-\bb{E}(\bar{Y_n})}{\Var^{1/2}(\bar{Y}_n)}\le
     \frac{-\bb{E}(\bar{Y}_n+R_n)}{\Var^{1/2}(\bar{Y}_n)}-\inv{t_n^2}\rp-\Phi\lp
     A_{\bx}t+C_{\bx}(\bH)\rp\\
   &=O\lp\inv{t_n^2}+\frac{1}{\sqrt{n|\bH|^{1/2}}}\rp +\Phi\lp
     \frac{-\bb{E}(\bar{Y}_n+R_n)}{\Var^{1/2}(\bar{Y}_n)}\rp-\Phi\lp
     A_{\bx}t+C_{\bx}(\bH)\rp\\
    \end{align*}
  \end{longform}
\end{mylongform}
Now we look at the integrated error
\begin{align*}
  \inv{\sqrt{n|\bH|^{1/2}}}\int_{\beta_{\tau}}\int_{t_{\bx}^{\ast}-t_n}^{t_{\bx}^{\ast}+t_n}\left|\Phi\lp
  \frac{-\bb{E}(\bar{Y}_n+R_n)}{\Var^{1/2}(\bar{Y}_n)}\rp-\Phi\lp
  A_{\bx}t+C_{\bx}(\bH)\rp \right|\,dt\,d\mc{H}(\bx).
\end{align*}
\begin{mylongform}
  \begin{longform}
    Since
    \begin{align*}
      &  \Phi\lp
        \frac{-\bb{E}(\bar{Y}_n+R_n)}{\Var^{1/2}(\bar{Y}_n)}\rp\\
      &\quad=\Phi\lp\frac{\lb t\|\nabla
        f_0(\bx)\|+\sqrt{n|\bH|^{1/2}}D_1(\bx,\bH)-\sqrt{n|\bH|^{1/2}}D_2(\bx,\bH)\rb\lb1+o(t_n^{-2})\rb}{\sqrt{R(K)f_{\tau,0}+o(t_n^{-2})}}\rp,
    \end{align*}
    uniformly in $\bx$ and $t$, and
    \begin{align*}
      \MoveEqLeft  \frac{\lb t\|\nabla
      f_0(\bx)\|+\sqrt{n|\bH|^{1/2}}D_1(\bx,\bH)-\sqrt{n|\bH|^{1/2}}D_2(\bx,\bH)\rb\lb1+o(t_n^{-2})\rb}{\sqrt{R(K)f_{\tau,0}+o(t_n^{-2})}}\\
   &=\frac{\lb t\|\nabla
     f_0(\bx)\|+\sqrt{n|\bH|^{1/2}}D_1(\bx,\bH)-\sqrt{n|\bH|^{1/2}}D_2(\bx,\bH)\rb\lb1+o(t_n^{-2})\rb}{\sqrt{R(K)f_{\tau,0}}+o(t_n^{-2})}\\
   &=\frac{\lb t\|\nabla
     f_0(\bx)\|+\sqrt{n|\bH|^{1/2}}D_1(\bx,\bH)-\sqrt{n|\bH|^{1/2}}D_2(\bx,\bH)\rb}{\sqrt{R(K)f_{\tau,0}}}\\
   &\qquad+\frac{\lb t\|\nabla
     f_0(\bx)\|+\sqrt{n|\bH|^{1/2}}D_1(\bx,\bH)-\sqrt{n|\bH|^{1/2}}D_2(\bx,\bH)\rb}{\sqrt{R(K)f_{\tau,0}}}\cdot o(t_n^{-2}),
    \end{align*}
  \end{longform}
\end{mylongform}
We can see that
\begin{align*}
  \MoveEqLeft  \left|\Phi\lp
  \frac{-\bb{E}(\bar{Y}_n+R_n)}{\Var^{1/2}(\bar{Y}_n)}\rp-\Phi\lp
  A_{\bx}t+C_{\bx}(\bH)\rp \right|\\
& \le \lb(t_n+|t_{\bx}^{\ast}|)\|\nabla
  f_0\|_{\infty}+\sqrt{n|\bH|^{1/2}}|D_1(\bx,\bH)|+\sqrt{n|\bH|^{1/2}}|D_2(\bx,\bH)|\rb o(t_n^{-2}).
\end{align*}
uniformly in $\bx$. From \eqref{eq:tx-exp} we know $|t_{\bx}^{\ast}|$ is uniformly
$O(\sqrt{n|\bH|^{1/2}}\tr(\bH))$, then
\begin{align*}
  \inv{\sqrt{n|\bH|^{1/2}}}\int_{\beta_{\tau}}\int_{t_{\bx}^{\ast}-t_n}^{t_{\bx}^{\ast}+t_n}(t_n+|t_{\bx}^{\ast}|)\|\nabla
  f_0\|_{\infty}o(t_n^{-2})\,dt\,d\bx=o\lp \inv{\sqrt{n|\bH|^{1/2}}}+\tr(\bH)\rp,
\end{align*}
and similarly
\begin{align*}
  \MoveEqLeft    \inv{\sqrt{n|\bH|^{1/2}}}\int_{\beta_{\tau}}\int_{t_{\bx}^{\ast}-t_n}^{t_{\bx}^{\ast}+t_n}
  \lb\sqrt{n|\bH|^{1/2}}\lp|D_1(\bx,\bH)|+|D_2(\bx,\bH)|\rp\rb
  o(t_n^{-2})\,dt\,d\bx\\
&=o\lp \inv{\sqrt{n|\bH|^{1/2}}}+\tr(\bH)\rp.
\end{align*}
So we have
\begin{align*}
  \MoveEqLeft  \frac{f_{\tau,0}}{\sqrt{n|\bH|^{1/2}}} \int_{\beta_\tau}\int_{t^{\ast}_{\bx}-t_n}^{t^{\ast}_{\bx}+t_n}\left|P\lp\ffnH(\bx^t)<\fftaun\rp-\one_{\{t<0\}}\right|\,dt\,d\mc{H}(\bx)\\
& = \frac{f_{\tau,0}}{\sqrt{n|\bH|^{1/2}}}
  \int_{\beta_\tau}\int_{t^{\ast}_{\bx}-t_n}^{t^{\ast}_{\bx}+t_n}\left|\Phi\lp
  A_{\bx}t+C_{\bx}(\bH)\rp-\one_{\{t<0\}}\right|\,dt\,d\mc{H}(\bx)\\
&\qquad+o\lp \inv{\sqrt{n|\bH|^{1/2}}}+\tr(\bH)\rp.
\end{align*}
It remains to see from Lemma~\ref{lem:normal-integral} that
\begin{align*}
  \MoveEqLeft  \frac{f_{\tau,0}}{\sqrt{n|\bH|^{1/2}}}
  \int_{\beta_\tau}\int_{-\infty}^{\infty}\left|\Phi\lp
  A_{\bx}t+C_{\bx}(\bH)\rp-\one_{\{t<0\}}\right|\,dt\,d\mc{H}(\bx)\\
&=  \frac{f_{\tau,0}}{\sqrt{n|\bH|^{1/2}}}\int_{\beta_{\tau}}
  \frac{2\phi(C_{\bx}(\bH)) + 2\Phi(C_{\bx}(\bH))C_{\bx}(\bH) - C_{\bx}(\bH)}{A_{\bx}} \,
  d\mc{H}(\bx).
\end{align*}

\subsection{Proof of Theorem~\ref{thm:levelset}}
We also provide a brief proof for Theorem~\ref{thm:levelset}, which
is a simpler and shares the same idea as that of
Theorem~\ref{thm:hdr}. First, we have
\begin{align}
  \MoveEqLeft  \bb{E}\ls \mu_{f_0}\{\mc{L}(c)\Delta\hat{\mc{L}}_{\bH}(c)\}\rs\nonumber\\&=\bb{E}\int_{\RR^d}f_0(\bx)\left|\one_{\{\ffnH(\bx)\ge
                                                                                          c\}}-\one_{\{f_0(\bx)\ge c\}}\right|\,d\bx\nonumber\\
&=\int_{\mc{L}(c)^c}f_0(\bx)P\lp\ffnH(\bx)\ge
  c\rp\,d\bx +\int_{\mc{L}(c)}f_0(\bx)P\lp\ffnH(\bx)< c\rp\,d\bx.
\end{align}
Like Lemma~\ref{lem:hdr-step3}, we can
shrink the region of interest.
We  show that for each $\delta>0$ sufficiently small, we
have
\begin{align*}
  \int_{\mc{L}_{\delta}(c)^c}f_0(\bx)P\lp\ffnH(\bx)\ge c\rp\,d\bx
  +\int_{\mc{L}(c)}f_0(\bx)P\lp\ffnH(\bx)< c\rp\,d\bx=o(n^{-1}),
\end{align*}
as $n\rightarrow \infty$.

\noindent Observe that under Assumption \ref{assm:DA-ls} if $\delta>0$ is
sufficiently small, then there exists $\epsilon>0$ s.t $f_0(\bx)\le
c-\epsilon$ for $\bx\in \mc{L}_{\delta}(c)^c$ and
$f_0(\bx)\ge c+\epsilon$ for $\bx\in
\mc{L}_{-\delta}(c)$.  By reducing $\delta>0$ if necessary,
for $\bx\in \mc{L}_{\delta}(c)^c$,
\begin{align*}
  P\lp\ffnH(\bx)\ge c\rp&=  P\lp\ffnH(\bx)-c\ge 0\rp\\
                        &\le
                          P\lp\ffnH(\bx)-c+c-f_0(\bx)\ge
                          \epsilon\rp\\
                        &\le P\lp\|\ffnH-f_0\|_{\infty}\ge \epsilon\rp.
\end{align*}
Similarly we can show the same bound for $\bx\in
\mc{L}_{-\delta}(c)$. Then
\begin{align*}
  \MoveEqLeft  \int_{\mc{L}_{\delta}(c)^c}f_0(\bx)P\lp\ffnH(\bx)\ge c\rp\,d\bx
  +\int_{\mc{L}_{-\delta}(c)}f_0(\bx)P\lp\ffnH(\bx)<
  c\rp\,d\bx\\
&\le P\lp\|\ffnH-f_0\|_{\infty}\ge \epsilon\rp\\
&\le P\lp\|\ffnH-\bb{E}\ffnH\|_{\infty}\ge
  \frac{\epsilon}{2}\rp+P\lp\|\bb{E}\ffnH-f_0\|_{\infty}\ge \frac{\epsilon}{2}\rp,
\end{align*}
where $P(\|\bb{E}\ffnH-f_0\|_{\infty}\ge \frac{\epsilon}{2})=0$ for
$n$ large enough. So with the same argument in proof
Lemma~\ref{lem:hdr-step3}, the above quantity is $o(n^{-1})$.
Further,  we have that for a sequence $\delta_n$ converging to 0 such
that $\lambda_{\max}(\bH)=o(\delta_n)$,
\begin{align}
  \int_{\mc{L}_{\delta_n}(c)^c}f_0(\bx)P\lp\ffnH(\bx)\ge c\rp\,d\bx
  +\int_{\mc{L}_{-\delta_n}(c)}f_0(\bx)P\lp\ffnH(\bx)< c\rp\,d\bx=o(n^{-1}) .
\end{align}
and we also prove this by  showing that $E(\delta,\delta_n)$ which is defined as
\begin{align*}
  &\int_{\mc{L}_{\delta_n}(c)^c}f_0(\bx)P\lp\ffnH(\bx)\ge c\rp\,d\bx
    +\int_{\mc{L}_{-\delta_n}(c)}f_0(\bx)P\lp\ffnH(\bx)<
    c\rp\,d\bx\\
  &-\lb\int_{\mc{L}_{\delta}(c)^c}f_0(\bx)P\lp\ffnH(\bx)\ge c\rp\,d\bx
    +\int_{\mc{L}_{-\delta}(c)}f_0(\bx)P\lp\ffnH(\bx)<
    c\rp\,d\bx\rb\\
  &=\int_{\mc{L}_{\delta_n}(c)^c\backslash
    \mc{L}_{\delta_n}(c)^c}f_0(\bx)P\lp\ffnH(\bx)\ge
    c\rp\,d\bx\\
  &\quad +\int_{\mc{L}_{-\delta_n}(c)\backslash\mc{L}_{-\delta}(c)}f_0(\bx)P\lp\ffnH(\bx)< c\rp\,d\bx
\end{align*}
is $o(n^{-1})$ as $n\rightarrow \infty$. Note that there exits a constant $c_2$ small s.t if we take
$\epsilon_n=c_2\delta_n$, then we have $|f_0(\bx)-c|\ge
\epsilon_n$ when $\bx\in\mc{L}_{\delta_n}(c)^c\backslash
\mc{L}_{\delta_n}(c)^c\cup\mc{L}_{-\delta_n}(c)\backslash\mc{L}_{-\delta}(c)$.
Then for $\bx\in \mc{L}_{\delta_n}(c)^c\backslash
\mc{L}_{\delta_n}(c)^c$,
\begin{align*}
  P\lp\ffnH(\bx)\ge c\rp&\le
                          P\lp\ffnH(\bx)-c+c-f_0(\bx)\ge \epsilon_n\rp\\
                        &\le
                          P\lp\|\ffnH-f_0\|_{\infty}\ge \epsilon_n\rp.
\end{align*}
We can derive the same bound for $\bx \in
\mc{L}_{-\delta_n}(c)\backslash\mc{L}_{-\delta}(c)$. Then
\begin{align*}
  E(\delta,\delta_n)&\le P\lp\|\ffnH(\bx)-f_0(\bx)\|_{\infty}\ge
                      \epsilon_n\rp\\
                    &\le P\lp\|\ffnH-\bb{E}\ffnH\|_{\infty}\ge
                      \frac{\epsilon_n}{2}\rp+P\lp\|\bb{E}\ffnH-f_0\|_{\infty}\ge\frac{\epsilon_n}{2}\rp
\end{align*}
is $o(n^{-1})$ when $n$ is large enough.

\noindent Now the risk function can be expressed as
\begin{align*}
  \bb{E}\mu_{f_0}\lb\mc{L}(c)\Delta\hat{\mc{L}}(c)\rb&=\int_{\mc{L}(c)^c\backslash\mc{L}_{\delta_n}(c)^c}f_0(\bx)P\lp\ffnH(\bx)\ge c\rp\,d\bx\nonumber\\
                                                     &\quad+\int_{\mc{L}(c)\backslash\mc{L}_{-\delta_n}(c)}f_0(\bx)P\lp\ffnH(\bx)<
                                                       c\rp\,d\bx+o(n^{-1})\nonumber\\
                                                     &=\int_{\beta(c)^{\delta_n}}f_0(\bx)\left|P\lp\ffnH(\bx)<c\rp-\one_{\{f_0(\bx)< c\}}\right|\,d\bx+o(n^{-1}).
\end{align*}
Then according to Lemma~\ref{lem:COV-approx}, when $\delta_n$ is small
enough
\begin{align*}
  &
    \int_{\beta(c)^{\delta_n}}f_0(\bx)\left|P\lp\ffnH(\bx)< c\rp-\one_{\{f_0(\bx)<c\}}\right|\,d\bx\\
  &=\int_{\beta(c)}\int_{-\delta_n}^{\delta_n}f_0(\bx+tu_x)\left|P\lp\ffnH(\bx+tu_x)<c\rp-\one_{\{f_0(\bx+tu_x)<c\}}\right|\,dtd\mc{H}(\bx)\\
  &\qquad+O(\delta_n^2),
\end{align*}
where $u_{\bx}$ is the unit normal outer
vector at $\bx\in \beta(c)$. And by simple transformation,
\begin{align*}
  &\int_{\beta(c)}\int_{-\delta_n}^{\delta_n}f_0(\bx+tu_x)\left|P\lp\ffnH(\bx+tu_x)<c\rp-\one_{\{f_0(\bx+tu_x)<c\}}\right|\,dtd\mc{H}(\bx)\\
  &=\inv{\sqrt{n|\bH|^{1/2}}}\int_{\beta(c)}\int_{-\sqrt{n|\bH|^{1/2}}\delta_n}^{\sqrt{n|\bH|^{1/2}}\delta_n}f_0\lp\bx^t\rp\left|P\lp\ffnH(\bx^t)<c\rp-\one_{\{t<0\}}\right|\,dtd\mc{H}(\bx)\\
  &=\frac{c}{\sqrt{n|\bH|^{1/2}}}\int_{\beta(c)}\int_{-\sqrt{n|\bH|^{1/2}}\delta_n}^{\sqrt{n|\bH|^{1/2}}\delta_n}\left|P\lp\ffnH(\bx^t)<c\rp-\one_{\{t<0\}}\right|\,dtd\mc{H}(\bx)+O(\delta_n^2).
\end{align*}
To further shrink the intervals of interest, we also argue that  when $n$ is large
enough, $\bb{E}\{\ffnH(\bx^t\}$ is a
strictly monotone function of $t\in
[-\sqrt{n|\bH|^{1/2}}\delta_n,\sqrt{n|\bH|^{1/2}}\delta_n]$ with a
unique zero $t_x^{\ast}$. Now  we claim for a sequence $t_n$ diverging to infinity,
\begin{align*}
  \int_{\beta(c)}\int_{I_{\bx}^n}\left|P\lp\ffnH(\bx^t)<c\rp-\one_{\{t<0\}}\right|\,dtd\mc{H}(\bx)\to 0
\end{align*}
as $n\rightarrow \infty$, where
$I_x^{n}=[-\sqrt{n|\bH|^{1/2}}\delta_n,\sqrt{n|\bH|^{1/2}}\delta_n]\backslash
[t_x^{\ast}-t_n,t_x^{\ast}+t_n]$. For detail of proof, please refer to
the proof of  Theorem~\ref{thm:hdr}.

\noindent Now by previous steps, we know
\begin{align*}
  \MoveEqLeft \bb{E}\mu_{f_0}\{\mc{L}(c)\Delta\hat{\mc{L}}(c)\}\\
&=\frac{c}{\sqrt{n|\bH|^{1/2}}}\int_{\beta(c)}\int_{t_x^{\ast}-
  t_n}^{t_x^{\ast}+t_n}\left|P\lp\ffnH\lp\bx^t\rp<c\rp-\one_{\{t<0\}}\right|\,dtd\mc{H}(\bx)\\&\qquad+o\lp\inv{\sqrt{n|\bH|^{1/2}}}\rp,
\end{align*}
To complete the proof, it suffices to show the dominating term above is
equal to $\LS(\bH)+o(1/\sqrt{n|\bH|^{1/2}})$.
Let $Z_{ni}(\bx)=K_{\bH}(\bx-\bX_i)$ and
$Y_{ni}=Z_{ni}(\bx^t)-c$. Then $\ffnH(\bx^t)-c=\bar{Y}_n$. Now let $t_n$ diverge slowly such
that
\begin{align}
  \label{eq:thm-ls-eybar}
  \bb{E}(\bar{Y}_n)=\lb\frac{t}{\sqrt{n|\bH|^{1/2}}}\|\nabla f_0(\bx)\|+D_1(\bx,\bH)\rb\{1+o(t_n^{-2})\},
\end{align}
by Assumption~\ref{assm:DA-hdr} part \ref{assm:DA:item3}, and
\begin{align}
  \label{eq:thm-ls-varybar}
  n|\bH|^{1/2}\Var\bar{Y}_n= R(K)c+o(t_n^{-2}),
\end{align}
uniformly for $\in[t_{\bx}^{\ast}-t_n,t_{\bx}^{\ast}+t_n]$ and  $\bx\in\beta(c)$.
Then
\begin{align*}
  \MoveEqLeft    P\lp\ffnH\lp\bx^t\rp<c\rp -\Phi(A_{\bx}t+B_{\bx}(\bH))\\
&=P\lp\frac{\bar{Y}_n-\bb{E}\bar{Y}_n}{\Var^{1/2}\bar{Y}_n}\le
  \frac{-\bb{E}\bar{Y}_n}{\Var^{1/2}\bar{Y}_n}\rp-\Phi(A_{\bx}t+B_{\bx}(\bH)),
\end{align*}
applying the Berry-Esseen theorem \citep[Page 31]{Ferguson:1996bn}
to the first term above yields
\begin{align*}
  \MoveEqLeft \left|P\lp\frac{\bar{Y}_n-\bb{E}\bar{Y}_n}{\Var^{1/2}\bar{Y}_n}\le
  \frac{-\bb{E}\bar{Y}_n}{\Var^{1/2}\bar{Y}_n}\rp-\Phi\lp \frac{-\bb{E}\bar{Y}_n}{\Var^{1/2}\bar{Y}_n}\rp\right|\\&\le\frac{C\EE|Y_{ni}|^3}{\Var^{3/2}(Y_{ni})\sqrt{n}}=O\lp\inv{\sqrt{n|\bH|^{1/2}}}\rp,
\end{align*}
and
\begin{align*}
  \MoveEqLeft     P\lp\frac{\bar{Y}_n-\bb{E}\bar{Y}_n}{\Var^{1/2}\bar{Y}_n}\le
  \frac{-\bb{E}\bar{Y}_n}{\Var^{1/2}\bar{Y}_n}\rp-\Phi(A_{\bx}t+B_{\bx}(\bH))\\
&\le \Phi\lp \frac{-\bb{E}\bar{Y}_n}{\Var^{1/2}\bar{Y}_n}\rp- \Phi(A_{\bx}t+B_{\bx}(\bH))+O\lp\inv{\sqrt{n|\bH|^{1/2}}}\rp,
\end{align*}
uniformly for $\in[t_{\bx}^{\ast}-t_n,t_{\bx}^{\ast}+t_n]$ and
$\bx\in\beta(c)$. A similar argument shows a lower bound of the same
order. Next, with a similar argument as we had in the last step of
proof for Theorem~\ref{thm:hdr}, we can show            the integrated error
\begin{align*}
  \MoveEqLeft  \inv{\sqrt{n|\bH|^{1/2}}}\int_{\beta(c)}\int_{t_{\bx}^{\ast}-t_n}^{t_{\bx}^{\ast}+t_n}
  \left|\Phi\lp\frac{-\EE{\bar{Y}_n}}{\Var^{1/2}\bar{Y}_n}\rp-\Phi(A_{\bx}+B_{\bx}(\bH))\right|\,dtd\mc{H}(\bx)\\&=o\lp\inv{\sqrt{n|\bH|^{1/2}}}+\tr(\bH)\rp.
\end{align*}
So we have
\begin{align*}
  \MoveEqLeft  \frac{c}{\sqrt{n|\bH|^{1/2}}}\int_{\beta(c)}\int_{t_x^{\ast}-
  t_n}^{t_x^{\ast}+t_n}\left|P\lp\ffnH\lp\bx^t\rp<c\rp-\one_{\{t<0\}}\right|\,dtd\mc{H}(\bx)\\
&=\frac{c}{\sqrt{n|\bH|^{1/2}}}\int_{\beta(c)}\int_{t_x^{\ast}-
  t_n}^{t_x^{\ast}+t_n}\left|\Phi(A_{\bx}+B_{\bx}(\bH))-\one_{\{t<0\}}\right|\,dtd\mc{H}(\bx)\\&\qquad+o\lp\inv{\sqrt{n|\bH|^{1/2}}}+\tr(\bH)\rp.
\end{align*}
By Lemma~\ref{lem:normal-integral},
\begin{align*}
  \frac{c}{\sqrt{n|\bH|^{1/2}}}\int_{\beta(c)}\int_{-\infty}^{\infty}\left|\Phi(A_{\bx}+B_{\bx}(\bH))-\one_{\{t<0\}}\right|\,dtd\mc{H}(\bx)=\LS(\bH).
\end{align*}
This completes the proof.

\subsection{Proof of Corollary~\ref{cor:ls-bandwidth-selector}}
Let $\bH=h^2\bs{I}$.  If $h^2$ is of order $n^{-2 / (d+4)}$ then by Assumption~\ref{assm:DA3}, 
\begin{align*}
  \ffnH(\bx)&=f_0(\bx)+\frac{1}{2}\tr\{\bH\grad^2f_0(\bx)\}+O\lp
              n^{-4/(d+4)}\rp,\\
  \Var\ffnH(\bx)&=n^{-1}|\bH|^{-1/2}R(K)f_0(\bx)+O\lp n^{-6/(d+4)}\rp,
\end{align*}
uniformly in $\bx$. From the proof of Theorem~\ref{thm:levelset}, if
we pick $t_n=\sqrt{\log n}$, and $\delta_n=\sqrt{\log
  n/(n|\bH|^{1/2})}$, we can further quantify the error in equations
\eqref{eq:thm-ls-eybar} and \eqref{eq:thm-ls-varybar} as
\begin{align*}
  & \EE(\bar{Y}_n)=\lb \frac{t}{\sqrt{n|\bH|^{1/2}}}\|\grad
    f_0(\bx)\|+D_1(\bx,\bH)\rb\{1+O(n^{-2/(d+4)}\sqrt{\log n})\},\\
  &n|\bH|^{1/2}\Var\bar{Y}_n=R(K)c+O(n^{-2/(d+4)}\sqrt{\log n}),
\end{align*}
and further
\begin{align*}
  \EE[\mu_{f_0}\{\mc{L}(c)\Delta \widehat{\mc{L}}_{\bH}(c)\}]=\LS(h)+O\lp
  n^{-4/(d+4)} (\log n)^{3/2}  \rp.
\end{align*}
With a similar argument as Corollary~\ref{cor:hdr-oracle-bandwidth}, we
see $h_{\text{opt}}/h_0=1+O(n^{-2/(d+4)}(\log n)^{3/2})$.

Now we study  $\hat{h}_{\text{opt}}/h_{\text{opt}}$. Let
$g_{n,\bH_0}=\widehat{f}_{n,\bH_0}-f_0$,
$g_{n,\bH_1}=\widehat{f}_{n,\bH_1}-f_0$,
$g_{n,\bH_2}=\widehat{f}_{n,\bH_0}-f_0$.  Let
\begin{equation*}
  m(\bx):=\frac{\phi(sF_{\bx})+2\Phi(sF_{\bx})sF_{\bx}-sF_{\bx}}{-A_{\bx}},
\end{equation*}
and with slight abuse of notation, we let  $\widehat{m}(\bx)$ be
defined similarly  where
we substitute
$f_0$  with $\widehat{f}_{n,\bH_0}$, $\nabla f_0$  with
$\nabla \widehat{f}_{n,\bH_1}$ and $\nabla^2 f_0$  with $\nabla^2\widehat{f}_{n,\bH_0}$.
We look at the difference
\begin{align}
  \label{eq:cor-ls-bandwidth-diff}
  \begin{split}
    & \left|\int_{\beta(c)}m(\bx)\,d\mc{H}(\bx)-\int_{\widehat{\beta}_{n,\bH_1}(c)}\widehat{m}(\bx)\,d\mc{H}(\bx)\right|\\
    &\quad\le
    \left|\int_{\beta(c)}m(\bx)\,d\mc{H}(\bx)-\int_{\beta(c)}\widehat{m}(\bx)\,d\mc{H}(\bx)\right|
    +\left|\int_{\beta(c)}\widehat{m}(\bx)\,d\mc{H}(\bx)-\int_{\widehat{\beta}_{n,\bH_1}(c)}\widehat{m}(\bx)\,d\mc{H}(\bx)\right|,
  \end{split}
\end{align}
where
$\widehat{\beta}_{n,\bH_1}(c):=\widehat{f}_{n,\bH_1}^{-1}(c)$. By
Lemma~\ref{lem:Hauss-diff}, when  $\bH_1 \to 0$ and $n^{-1} | \bH_1|^{-1/2} (\bH_1^{-1})^{\otimes
  2} = O(1)$ as $n \to \infty$, the second term  on the last line
above is $ O_p (
\sup_{\bx \in \beta(c)} E [ | g_{n,\bH_1} (\bx) | +
\| \hess g_{n,\bH_1} (\bx) \| |g_{n,\bH_1}(\bx) | + \| \nabla
g_{n,\bH1}(\bx) \|  ] )$. For first term
on the last line above, by Jensen's inequality we know
$(\int_{\beta(c)}m(\bx)\,d\mc{H}(\bx)-\int_{\beta(c)}\widehat{m}(\bx)\,d\mc{H}(\bx))^2\le
\int_{\beta(c)}(m(\bx)-\widehat{m}(\bx))^2\,d\mc{H}(\bx)$. So for any (large) $M > 0$ we have
\begin{align*}
  \MoveEqLeft
  P\lp \left| \int_{\beta(c)} m(\bx)\,d\mc{H}(\bx)-\int_{\beta(c)}\widehat{m}(\bx)\,d\mc{H}(\bx)\right|> M \rp \\
&= P\lp\left|\int_{\beta(c)}m(\bx)\,d\mc{H}(\bx)-\int_{\beta(c)}\widehat{m}(\bx)\,d\mc{H}(\bx)\right|^2>M^2\rp 
\end{align*}
which is bounded above by
\begin{align*}
  P\lp\int_{\beta(c)}\lb
  m(\bx)-\widehat{m}(\bx)\rb^2\,d\mc{H}(\bx)>M^2\rp 
  & \le\frac{ \EE \int_{\beta(c)}\lb  m(\bx)-\widehat{m}(\bx)\rb^2\,d\mc{H}(\bx)}{M^2},
\end{align*}
by Markov's inequality. And by Tonelli's Theorem, we have $\EE \int_{\beta(c)}\lb m(\bx)-\widehat{m}(\bx)\rb^2\,d\mc{H}(\bx)= \int_{\beta(c)}\EE\lb m(\bx)-\widehat{m}(\bx)\rb^2\,d\mc{H}(\bx)$.

\begin{mylongform}
  \begin{longform}
    Extension of classical results on the stochastic errors of kernel
    density estimator and derivative estimator using unconstrained
    (full) bandwidth matrix can be found in
    \citet{chacon2011asymptotics}.
  \end{longform}
\end{mylongform}
Since we assume the true density function has 4 continuous bounded
derivatives, by Theorem 4 in \citet{chacon2011asymptotics} with
slight modification, it can be  easily  seen
$|\widehat{f}_{n,\bH_1}(\bx)-f_0(\bx)|=O_p(n^{-2/(d+6)})$,
$\|\nabla\widehat{f}_{n,\bH_1}(\bx)-\nabla
f_0(\bx)\|=O_p(n^{-2/(d+6)})$,
$\|\nabla^2\widehat{f}_{n,\bH_2}(\bx)-\nabla^2
f_0(\bx)\|$ is $O_p(n^{-2/(d+8)})$. Thus
$\widehat{F}_{\bx}=F_{\bx}+O_p(n^{-2/(d+8)})$, and
$\widehat{A}_{\bx}=A_{\bx}+O_p(n^{-2/(d+8)})$. And  we can also see 
\begin{equation*}
  \sup_{\bx \in \beta(c)}E [ | g_{n,\bH_1} (\bx) | +
  \| \hess g_{n,\bH_1} (\bx) \| |g_{n,\bH_1}(\bx) | + \| \nabla
  g_{n,\bH1}(\bx) \|  ]  = O(n^{-2/(d+6)}),
\end{equation*}
Thus, we can check that
$\int_{\beta(c)}\EE\lb
m(\bx)-\widehat{m}(\bx)\rb^2\,d\mc{H}(\bx)$ is $O(n^{-4/(d+8)})$, and
the first term on the last line of \eqref{eq:cor-ls-bandwidth-diff} is $O_p(n^{-2/(d+8)})$. 
We
can conclude that for any $0<s_1<s_2<\infty$, we have
$\widehat{\text{AR}}_{\LS}(s)=\text{AR}_{\LS}(s)\{1+O_p(n^{-2/(d+8)})\}$
uniformly for $s\in [s_1,s_2]$.  Then we have
$\widehat{\text{AR}}'_{\LS}(\hat{s}_{\text{opt}})=\text{AR}'_{\LS}(\hat{s}_{\text{opt}})\{1+O_p(n^{-2/(d+8)})\}=\text{AR}^{''}_{\LS}(\tilde{s})(\hat{s}_{\text{opt}}-s_{\text{opt}})\{1+O_p(n^{-2/(d+8)})\}$,
where $\text{AR}^{''}_{\LS}(\tilde{s})>0$ and is bounded from 0 as $n\to
\infty$. This gives us
$\hat{s}_{\text{opt}}/s_{\text{opt}}=1+O_p(n^{-2/(d+8)})$, and recall that
$\widehat{h}_{\text{opt}}=\hat{s}_{\text{opt}}^{ 2/(d+4)}n^{-1/(d+4)}$,
$h_{\text{opt}}=s_{\text{opt}}^{ 2/(d+4)}n^{-1/(d+4)}$, we conclude
\begin{align*}
  \frac{\hat{h}_{\text{opt}}}{h_{\text{opt}}}=1+O_p\lp n^{-2/(d+8)}\rp.
\end{align*}
Combining this result with $h_{\text{opt}}/h_0=1+O(n^{-2/(d+4)}(\log
n^{3/2}))$ gives us $\hat{h}_{\text{opt}}/h_0=O_p(n^{-2/(d+8)})$.
\subsection{Proof of Corollary~\ref{cor:hdr-bandwidth-selector}}
Similar to the proof of Corollary~\ref{cor:ls-bandwidth-selector}, let
$g_{n,\bH_0}=\widehat{f}_{n,\bH_0}-f_0$,
$g_{n,\bH_1}=\widehat{f}_{n,\bH_1}-f_0$,
$g_{n,\bH_2}=\widehat{f}_{n,\bH_0}-f_0$,
and let $\epsilon=\widehat{f}_{\tau,n,\bH_0}-f_{\tau,0}$.
Since we assume the true density function has 4 continuous bounded
derivatives, again by Theorem 4 in \citet{chacon2011asymptotics}, it can be  easily  seen
$|\widehat{f}_{n,\bH_0}(\bx)-f_0(\bx)|=O_p(n^{-2/(d+4)})$,
$|\widehat{f}_{n,\bH_1}(\bx)-f_0(\bx)|=O_p(n^{-2/(d+6)})$,
$\|\nabla\widehat{f}_{n,\bH_1}(\bx)-\nabla
f_0(\bx)\|=O_p(n^{-2/(d+6)})$,
$\|\nabla^2\widehat{f}_{n,\bH_2}(\bx)-\nabla^2
f_0(\bx)\|=O_p(n^{-2/(d+8)})$. And by
Lemma~\ref{lem:hdr-step1}, $|\widehat{f}_{\tau,n,\bH_0}-f_{\tau,0}
|=O_p(n^{-2/(d+8)})$.
We first look at the difference
\begin{align}
  \begin{split}
    \label{eq:cor2-eq1}
    &\left|\int_{\beta_{\tau}}\frac{1}{\|\grad f_0\|}\,d\mc{H}-
      \int_{\widehat{\beta}_{\tau,\bH_1}}\frac{1}{\|\grad\widehat{f}_{n,\bH_1}\|}\,d\mc{H}\right|\\
    &\le
    \left|\int_{\beta_{\tau}}\frac{1}{\|\grad
        f_0\|}\,d\mc{H}-\int_{\beta_{\tau}}\frac{1}{\|\grad \widehat{f}_{n,\bH_1}\|}\,d\mc{H}\right|
    +\left|\int_{\beta_{\tau}}\frac{1}{\|\grad \widehat{f}_{n,\bH_1}\|}\,d\mc{H}-
      \int_{\widehat{\beta}_{\tau,\bH_1}}\frac{1}{\|\grad
        \widehat{f}_{n,\bH_1}\|}\,d\mc{H}\right|.
  \end{split}
\end{align}
Since $\|\grad f_0-\grad \widehat{f}_{n,\bH_1}\|=O_p(n^{-2/(d+6)})$ by
\citet{chacon2011asymptotics}, it is easy to see the first term on the
last line above is $O_p(n^{-2/(d+6)})$. Recalling $\epsilon=\widehat{f}_{\tau,n,\bH_0}-f_{\tau,0}$,  we can bound the second term as
\begin{align*}
  \MoveEqLeft  \left|\int_{\beta_{\tau}}\frac{1}{\|\grad \widehat{f}_{n,\bH_1}\|}\,d\mc{H}-
  \int_{\widehat{\beta}_{\tau,\bH_1}}\frac{1}{\|\grad
  \widehat{f}_{n,\bH_1}\|}\,d\mc{H}\right|\\
&\le \left|\int_{\{f_0=f_{\tau,0}\}}\frac{1}{\|\grad \widehat{f}_{n,\bH_1}\|}\,d\mc{H}-
  \int_{\{\widehat{f}_{n,\bH_1}=f_{\tau,0}\}}\frac{1}{\|\grad
  \widehat{f}_{n,\bH_1}\|}\,d\mc{H}\right|\\
&\quad+ \left|\int_{\{f_0=f_{\tau,0}\}}\frac{1}{\|\grad \widehat{f}_{n,\bH_1}\|}\,d\mc{H}-
  \int_{\{\widehat{f}_{n,\bH_1}=f_{\tau,0}+\epsilon\}}\frac{1}{\|\grad
  \widehat{f}_{n,\bH_1}\|}\,d\mc{H}\right|,
\end{align*}
By Lemma~\ref{lem:Hauss-diff}, the first term  $|\int_{\{f_0=f_{\tau,0}\}}\inv{\|\grad \widehat{f}_{n,\bH_1}\|}\,d\mc{H}-
\int_{\{\widehat{f}_{n,\bH_1}=f_{\tau,0}\}}\inv{\|\grad
  \widehat{f}_{n,\bH_1}\|}\,d\mc{H}|$ is of order $O_p (
\sup_{\bx \in \beta(c)} E [ | g_{n,\bH_1} (\bx) | +
\| \hess g_{n,\bH_1} (\bx) \| |g_{n,\bH_1}(\bx) | + \| \nabla
g_{n,\bH1}(\bx) \|  ] )$  when    $\bH_1 \to 0$ and $n^{-1} | \bH_1|^{-1/2} (\bH_1^{-1})^{\otimes
  2} = O(1)$ as $n \to \infty$. Then, by Taylor expansion, we have
\begin{align*}
  \MoveEqLeft
  \int_{\{\widehat{f}_{n,\bH_1}=f_{\tau,0}+\epsilon\}}\frac{1}{\|\grad
  \widehat{f}_{n,\bH_1}\|}\,d\mc{H}\\
&=\int_{\{\widehat{f}_{n,\bH_1}=f_{\tau,0}\}}\frac{1}{\|\grad
  \widehat{f}_{n,\bH_1}\|}\,d\mc{H}-\lp\frac{d}{d e}\left.\int_{\{\widehat{f}_{n,\bH_1}=f_{\tau,0}+e\}}\frac{1}{\|\grad
  \widehat{f}_{n,\bH}\|}\,d\mc{H}\right\vert_{e=s\epsilon}\rp\epsilon,
\end{align*}
where $s\in [0,1]$. Then we have
\begin{align}
  \label{eq:cor-hdr-haus-deriv}
  \begin{split}
    \MoveEqLeft  \left|\int_{\{\widehat{f}_{n,\bH_1}=f_{\tau,0}\}}\frac{1}{\|\grad \widehat{f}_{n,\bH_1}\|}\,d\mc{H}-
      \int_{\{\widehat{f}_{n,\bH_1}=f_{\tau,0}+\epsilon\}}\frac{1}{\|\grad \widehat{f}_{n,\bH}\|}\,d\mc{H}\right|\\
    &= \left|\lp\frac{d}{d e}\left.\int_{\{\widehat{f}_{n,\bH_1}=f_{\tau,0}+e\}}\frac{1}{\|\grad
          \widehat{f}_{n\bH_1}\|}\,d\mc{H}\right\vert_{e=s\epsilon}\rp\epsilon\right|.
  \end{split}
\end{align}
From the proof of Lemma~\ref{lem:Hauss-diff}, we can see
when $n$ is sufficiently large, the derivative on the last line of \eqref{eq:cor-hdr-haus-deriv} is uniformly bounded. Moreover,
by Lemma~\ref{lem:hdr-step1}, $\epsilon=O(\|g_{n,\bH_0}\|_{\infty})$, so we
have
\begin{align*}
  \MoveEqLeft  \left|\int_{\beta_{\tau}}\frac{1}{\|\grad
  \widehat{f}_{n,\bH_1}\|}\,d\mc{H}-\int_{\widehat{\beta}_{\tau,\bH_1}}\frac{1}{\|\grad
  \widehat{f}_{n,\bH_1}\|}\,d\mc{H}\right|\\
&=O(\|g_{n,\bH_0}\|_{\infty})+ O_p (
  \sup_{\bx \in \beta(c)} E [ | g_{n,\bH_1} (\bx) | +
  \| \hess g_{n,\bH_1} (\bx) \| |g_{n,\bH_1}(\bx) | + \| \nabla
  g_{n,\bH1}(\bx) \|  ] )\\
&=O_p(n^{-2/(d+6)}).
\end{align*}
Then
\begin{align*}
  \left|\int_{\beta_{\tau}}\frac{1}{\|\grad f_0\|}\,d\mc{H}-
  \int_{\widehat{\beta}_{\tau,\bH_1}}\frac{1}{\|\grad
  \widehat{f}_{n,\bH_1}\|}\,d\mc{H}\right|=O_p(n^{-2/(d+6)}),
\end{align*}
and thus $|w_0-\hat{w}_0|=O_p(n^{-2/(d+6)})$. Using exactly the same trick, we can show
\begin{align*}
  \left|  \int_{\beta_{\tau}}\frac{\mu(K)\tr(\nabla^2
  f_0)}{2\|\nabla f_0\|}\,d\mc{H}-\int_{\widehat{\beta}_{\tau,\bH_1}}\frac{\mu(K)\tr(\nabla^2
  \widehat{f}_{n,\bH_2})}{2\|\nabla
  \widehat{f}_{n,\bH_1}\|}\,d\mc{H}\right|=O_p(n^{-2/(d+8)}).
\end{align*}
Next, we provide the bound for
$|\int_{\mc{L}_{\tau}}\frac{\mu(K)\tr(\nabla^2f_0)}{2}\,d\lambda-\int_{\widehat{\mc{L}}_{\tau,\bH_0}}\frac{\mu(K)\tr(\nabla^2\widehat{f}_{n,\bH_2})}{2}\,d\lambda|$. Similarly,
we have
\begin{align}
  \label{eq:sup-level-integral}
  \begin{split}
    \MoveEqLeft
    \left|\int_{\mc{L}_{\tau}}\frac{\mu(K)\tr(\nabla^2f_0)}{2}\,d\lambda-\int_{\widehat{\mc{L}}_{\tau,\bH_0}}\frac{\mu(K)\tr(\nabla^2\widehat{f}_{n,\bH_2})}{2}\,d\lambda\right|\\
    &\le
    \left|\int_{\mc{L}_{\tau}}\frac{\mu(K)\tr(\nabla^2f_0)}{2}\,d\lambda-\int_{\widehat{\mc{L}}_{\tau,\bH_0}}\frac{\mu(K)\tr(\nabla^2f_0)}{2}\,d\lambda\right|\\
    &\qquad +
    \left|\int_{\widehat{\mc{L}}_{\tau,\bH_0}}\frac{\mu(K)\tr(\nabla^2f_0)}{2}\,d\lambda-\int_{\widehat{\mc{L}}_{\tau,\bH_0}}\frac{\mu(K)\tr(\nabla^2\widehat{f}_{n,\bH_2})}{2}\,d\lambda\right|,
  \end{split}
\end{align}
and since we assume $f_0$ has bounded second derivatives,
the difference on the second line above $|\int_{\mc{L}_{\tau}}\frac{\mu(K)\tr(\nabla^2f_0)}{2}\,d\lambda-\int_{\widehat{\mc{L}}_{\tau,\bH_0}}\frac{\mu(K)\tr(\nabla^2f_0)}{2}\,d\lambda|=O\{|\lambda(\mc{L}_{\tau})-\lambda(\hat{\mc{L}}_{\tau,\bH_0})|\}$. Now
we show
$|\lambda(\mc{L}_{\tau})-\lambda(\hat{\mc{L}}_{\tau,\bH})|=O(\|g_{n,\bH_0}\|_{\infty})$. It
can be seen that
\begin{align*}
  \{\bx:f_0(\bx)\ge f_{\tau,0}+\epsilon+\|g_{n,\bH_0}\|_{\infty}\}\subset
  \widehat{\mc{L}}_{\tau,\bH_0}\subset\{\bx:f_0(\bx)\ge f_{\tau,0}+\epsilon-\|g_{n,\bH_0}\|_{\infty}\},
\end{align*}
and then
\begin{align*}
  \MoveEqLeft |\lambda(\mc{L}_{\tau})-\lambda(\widehat{\mc{L}}_{\tau,\bH_0})|\\
&\le |\lambda(\mc{L}_{\tau})-\lambda \{\bx:f_0(\bx)\ge f_{\tau,0}+\epsilon+\|g_{n,\bH_0}\|_{\infty}\}|\\&\qquad+ |\lambda(\mc{L}_{\tau})-\lambda\{\bx:f_0(\bx)\ge f_{\tau,0}+\epsilon-\|g_{n,\bH_0}\|_{\infty}\}|.
\end{align*}
Further by Proposition A.1 of \cite{Cadre:2006db},
\begin{align*}
  \MoveEqLeft  |\lambda(\mc{L}_{\tau})-\lambda \{\bx:f_0(\bx)\ge
  f_{\tau,0}+\epsilon+\|g_{n,\bH_0}\|_{\infty}\}|\\&=\left|(\epsilon+\|g_{n,\bH_2}\|_{\infty})\int_{\beta_{\tau}}\frac{1}{\|\grad
                                                     f_0\|}\,d\mc{H}\right|+o(\epsilon
                                                     +\|g_{n,\bH_0}\|_{\infty}),
\end{align*}
and
\begin{align*}
  \MoveEqLeft  |\lambda(\mc{L}_{\tau})-\lambda\{\bx:f_0(\bx)\ge f_{\tau,0}+\epsilon-\|g_{n,\bH_0}\|_{\infty}\}|\\&=\left|(\epsilon-\|g_{n,\bH_0}\|_{\infty})\int_{\beta_{\tau}}\frac{1}{\|\grad f_0\|}\,d\mc{H}\right|+o(\epsilon-\|g_{n,\bH_0}\|_{\infty}),
\end{align*}
and thus
\begin{align*}
  |\lambda(\mc{L}_{\tau})-\lambda(\widehat{\mc{L}}_{\tau,\bH_0})|=O(\|g_{n,\bH_0}\|_{\infty}),
\end{align*}
when $\|g_{n,\bH_0}\|_{\infty}\to 0$.

Now for the last term of \eqref{eq:sup-level-integral}, by Jensen's
inequality we have
\begin{align*}
  \lp\int_{\widehat{\mc{L}}_{\tau,\bH_0}}\frac{\mu(K)\tr(\nabla^2 g_{n,\bH_2})}{2}\,d\lambda\rp^2
  &\le\int_{\widehat{\mc{L}}_{\tau,\bH_0}}\lp\frac{\mu(K)\tr(\nabla^2
    g_{n,\bH_2})}{2}\rp^2\,d\lambda\\
  &\le\int\lp\frac{\mu(K)\tr(\nabla^2 g_{n,\bH_2})}{2}\rp^2\,d\lambda,
\end{align*}
and then for any (large) $M > 0$, 
\begin{align*}
  P\lp
  \left|\int_{\widehat{\mc{L}}_{\tau,\bH_0}}\frac{\mu(K)\tr(\nabla^2
  g_{n,\bH_2})}{2}\,d\lambda\right|>M\rp\le\frac{\EE \int\lp\frac{\mu(K)\tr(\nabla^2 g_{n,\bH_2})}{2}\rp^2\,d\lambda}{M^2},
\end{align*}
where we applied Markov's inequality to obtain the upper bound. Applying
Tonelli's Theorem yields
\begin{align*}
  \EE \int\lp\frac{\mu(K)\tr(\nabla^2 g_{n,\bH_2})}{2}\rp^2\,d\lambda= \int\EE\lp\frac{\mu(K)\tr(\nabla^2 g_{n,\bH_2})}{2}\rp^2\,d\lambda=O_p(n^{-4/(d+8)}).
\end{align*}
So
$|\int_{\mc{L}_{\tau}}\frac{\mu(K)\tr(\nabla^2f_0)}{2}\,d\lambda-\int_{\widehat{\mc{L}}_{\tau,\bH_0}}\frac{\mu(K)\tr(\nabla^2\widehat{f}_{n,\bH_2})}{2}\,d\lambda| = O_p(n^{-2/(d+8)})$.

Now from \citet{chacon2011asymptotics}, we know that $\widehat{G}_{\bx}=G_{\bx}+O_p(n^{-2/(d+8)})$,
$\widehat{A}_{\bx}=A_{\bx}+O_p(n^{-2/(d+8)})$. And   using a similar trick as for
$w_0$, we have
\begin{align*}
  |  \widehat{\text{AR}}_{\HDR}(s)-\text{AR}_{\HDR}(s)|=O_p(n^{-2/(d+6)}).
\end{align*}
And we
can conclude that for any $0<s_1<s_2<\infty$, we have
$\widehat{\text{AR}}(s)=\text{AR}(s)\{1+O_p(n^{-2/(d+8)})\}$
uniformly for $s\in [s_1,s_2]$.  Then we have
$\widehat{\text{AR}}'_{\HDR}(\hat{s}_{\text{opt}})=\text{AR}'_{\HDR}(\hat{s}_{\text{opt}})\{1+O_p(n^{-2/(d+8)})\}=\text{AR}^{''}_{\HDR}(\tilde{s})(\hat{s}_{\text{opt}}-s_{\text{opt}})\{1+O_p(n^{-2/(d+8)})\}$,
where $\text{AR}^{''}_{\HDR}(\tilde{s})>0$ and is bounded from 0 as $n\to
\infty$. This gives us
$\hat{s}_{\text{opt}}/s_{\text{opt}}=1+O_p(n^{-2/(d+8)})$, and recalling that
$\widehat{h}_{\text{opt}}=\hat{s}_{\text{opt}}^{ 2/(d+4)}n^{-1/(d+4)}$,
$h_{\text{opt}}=s_{\text{opt}}^{ 2/(d+4)}n^{-1/(d+4)}$, we conclude
\begin{align*}
  \frac{\hat{h}_{\text{opt}}}{h_{\text{opt}}}=1+O_p\lp n^{-2/(d+8)}\rp.
\end{align*}

\section{Additional theorems and proofs}
\label{app:additional-thms}

The following theorem is a slight extension of Theorem~2.3 of
\cite{Gine:2002jc} to allow general bandwidth matrices and to apply to
gradient estimation.  Its proof is essentially the same as that of
their Theorem~2.3, so is omitted.
\begin{theorem}
  \label{thm:KDE-sup-rate}
  Let  
  $\bX_1, \ldots, \bX_n$ be i.i.d.\ from a bounded density on $\RR^d$,
  and let Assumptions~\ref{assm:KA}, and \ref{assm:HA} hold.
  We have
  \begin{align}
    \limsup_{n \to \infty}  \sqrt{\frac{n | \bH_n|^{1/2}}{\log | \bH_n|^{-1/2}}}
    \| \ffnH[n,\bH_n] - \EE\ffnH[n,\bH_n] \|_\infty = C_{0,1} \qquad a.s.,
    \label{eq:KDE-sup-rate}
  \end{align}
  and 
  \begin{align}
    \limsup_{n \to \infty}
    \sqrt{\frac{n | \bH_n|^{1/2} \lambda_{\min}(\bH)}{\log | \bH_n|^{-1/2}} } \| \grad
    \ffnH[n,\bH_n] - \EE \grad \ffnH[n,\bH_n] \|_\infty \le C_{0,2} \qquad a.s.,
    \label{eq:KDE-grad-sup-rate}
  \end{align}
  Here $C_{0,1}$ and $C_{0,2}$
  depend on $d, K,$ and $  \|f_0 \|_\infty.$
\end{theorem}
\begin{mylongform}
  \begin{longform}
    Here $C_{0,1}^2 \le M^2 c \| f_0 \|_\infty \| K\|_2^2$, for a constant $M$
    depending only on the kernel $K$ (and dimension $d$).

    See GineGuillou-notes.tex for details.
  \end{longform}
\end{mylongform}

The proof of Theorem~\ref{thm:KDE-sup-rate} also yields the following probability bound which we need in particular.
\begin{mylongform}
  \begin{longform}
    In the corollary we consider a single $n$ value, rather than the range $n \in (2^{k-1}, 2^k]$, simplifying the argument.
  \end{longform}
\end{mylongform}

\begin{corollary}
  \label{cor:convergerate}
  Let  
  $\bX_1, \ldots, \bX_n$ be i.i.d.\ from a bounded density on $\RR^d$,
  and let
  Assumptions~ \ref{assm:KA}, and \ref{assm:HA} hold. Then
  for
  some constant $C > 0$ and for
  $0 < \epsilon \le C \| K \|_2^2 \| f_0 \|_\infty / \|K\|_\infty $, we have
  \begin{equation}
    \label{eq:11}
    P\lb \lV \ffnH - \EE\ffnH \rV_\infty > \epsilon \rb
    \le L \exp \lb -C_{0,1} \epsilon^2 n | \bH_n |^{1/2} \rb,
  \end{equation}
  where $C_{0,1}$ depends on $K$, $d$, and $\| f_0 \|_\infty$.
  Similarly,
  for $0<\epsilon$ small enough
  (with bound depending on $\grad K$ and $\| f_0 \|_\infty$),
  \begin{equation}
    \label{eq:29}
    P\lb   \lV  \grad \ffnH - \EE \grad\ffnH \rV_\infty > \epsilon \rb
    \le L \exp \lb - C_{0,2} \epsilon^2 n | \bH_n |^{1/2} \lambda_{\bH} \rb,
  \end{equation}
  where $C_{0,2}>0$ depends on $\grad K$, $d$, and $  \| f_0 \|_\infty$, and where
  $\lambda_{\bH}$ is the smallest eigenvalue of $\bH$.
\end{corollary}
\begin{proof}
  We let
  $${\mc F}_{K,\bH_n} := \lb K( \bH_n^{-1/2}( \bs t - \cdot) ) : \bs t \in \RR^d \rb,$$
  (which is a VC class by
  Assumption~\ref{assm:KA}). We have that
  for   $\epsilon > 0$
  \begin{align}
    P \lb \lV \ffnH - \EE\ffnH \rV_\infty > \epsilon \rb
    = P \lb \inv{n |\bH_n|^{1/2}} \lV \sum_{i=1}^n f(\bX_i) - \EE f(\bX_i) \rV_{{\mc F}_{K,\bH}} > \epsilon
    \rb.
    \label{eq:KDE-prob-bound}
  \end{align}
  Thus we set
  \begin{equation*}
    \sigma_n^2 := |\bH_n|^{1/2} \| K\|_2^2 \|f_0\|_\infty
    \quad \mbox{ and } \quad
    U := \| K \|_\infty
  \end{equation*}
  which satisfy the conditions of
  Corollary~2.2 of \cite{Gine:2002jc} 
  so we have $L$ and $C$ (depending on
  $K$ and $d$) from the corollary, so we set $t = \epsilon n |\bH_n|^{1/2}$,
  and $\lambda = C$ so that
  (7) in \cite{Gine:2002jc} 
  is satisfied (using that
  $n |\bH_n|^{1/2} \to \infty$ for the lower bound).  We conclude that
  \eqref{eq:KDE-prob-bound} is bounded above by
  \begin{equation*}
    L \exp \lb - \frac{D \epsilon^2 n |\bH_n|^{1/2}}{\|K\|_2^2 \|f_0\|_\infty }\rb
  \end{equation*}
  where $D := (\log (1 + C/4L))/LC$, completing the proof.

  A similar proof shows that \eqref{eq:29} holds.
  Let $K_{\bH} := |\bH|^{-1/2} K( \bH^{-1/2} \cdot )$.
  Then $\grad K_{\bH}(\bs{y}) = |\bH|^{-1/2} \bH^{-1/2}  \grad K(\bH^{-1/2} \bs{y})$,
  so
  $    P \lb \lV \grad \ffnH - \EE\grad \ffnH \rV_\infty > \epsilon \rb$
  \begin{mylongform}
    \begin{longform}
      equals
    \end{longform}
  \end{mylongform}
  \begin{mylongform}
    \begin{longform}
      \begin{align*}
        P \lb \inv{n |\bH_n|^{1/2}}
        \lV  \bH^{-1/2} \sum_{i=1}^n \grad K(\bH^{-1/2}( \cdot -\bX_i))
        - \EE\grad K(\bH^{-1/2}( \cdot -\bX_i)) \rV_{\infty} > \epsilon  \rb
      \end{align*}
      which is bounded above by
      \begin{align*}
        d
        P \lb \inv{n |\bH_n|^{1/2}}
        \lV  a_{\bH}'  \sum_{i=1}^n \grad K(\bH^{-1/2}( \cdot -\bX_i))
        - \EE\grad K(\bH^{-1/2}( \cdot -\bX_i)) \rV_{\infty} > \epsilon  \rb
      \end{align*}
      where $a_{\bH}'$ is a row of $\bH^{-1/2}$,
      and the previous display
    \end{longform}
  \end{mylongform}
  is bounded above by
  \begin{mylongform}
    \begin{longform}
      (rather than equal to)
    \end{longform}
  \end{mylongform}
  \begin{align*}
    d
    P \lb \inv{n |\bH_n|^{1/2}}
    \| a_{\bH}\|
    \lV    \sum_{i=1}^n \| \grad K(\bH^{-1/2}( \cdot -\bX_i))
    - \EE\grad K(\bH^{-1/2}( \cdot -\bX_i)) \| \rV_{\infty} > \epsilon  \rb
  \end{align*}
  by the Cauchy-Schwarz inequality where $a_{\bH}'$ is a row of $\bH^{-1/2}$
  (and, recall, $\| \cdot \|$ is just Euclidean norm).
  Since $\| a_{\bH}\| \le \lambda_{\bH}^{-1/2}$ where $\lambda_{\bH}^{-1/2}$ is the largest eigenvalue of $\bH^{-1/2}$, the previous display is bounded above by
  \begin{mylongform}
    \begin{longform}
      \begin{align*}
        d
        P \lb \inv{n |\bH_n|^{1/2} \lambda_{\bH}^{1/2} }
        \lV    \sum_{i=1}^n \| \grad K(\bH^{-1/2}( \cdot -\bX_i))
        \| \rV_{\infty} > \epsilon  \rb
      \end{align*}
      which equals
    \end{longform}
  \end{mylongform}
  \begin{align*}
    P \lb \inv{n |\bH_n|^{1/2} \lambda_{\bH}^{1/2}}
    \lV \sum_{i=1}^n f(\bX_i) - \EE f(\bX_i) \rV_{{\mc F}_{K,\bH}} > \epsilon  \rb
    \label{eq:KDE-prob-bound-grad}
  \end{align*}
  where
  $${\mc F}_{K,\bH_n} := \lb  \| \grad K(\bH_n^{-1/2} ( \bs t - \cdot) ) \| : \bs t \in \RR^d \rb,$$
  is a VC class by Assumption~\ref{assm:KA}.  We thus take
  $\sigma_n^2 := |\bH_n|^{1/2} R(\grad K) \| f_0 \|_\infty$ and
  $U := \sqrt{d} \| \grad K \|_\infty$ and apply Corollary 2.2 of
  \cite{Gine:2002jc}.  Here $R(\grad K)$ is the largest eigenvalue of
  $\int (\grad K) (\grad K)' d\lambda$.  We take
  $t = \epsilon n | \bH |^{1/2} \lambda_{\bH}^{1/2}$ and $\lambda = C$.  Then
  (7) of \cite{Gine:2002jc} 
  is satisfied since $n^{1/2} | \bH |^{1/4} \lambda_{\bH}^{1/2} / \sqrt{ \log | \bH|^{-1/2}} \to \infty$.
  \begin{mylongform}
    \begin{longform}
      That is,  the lower bound in
      our \eqref{eq:t-range} ((7)   of \cite{Gine:2002jc}) is
      \begin{equation*}
        C \sqrt{n} | \bH|^{1/4} \sqrt{ \log U/|\bH|^{-1/2} }
        \le \epsilon n | \bH |^{1/2} \lambda_{\bH}^{1/2}
      \end{equation*}
      which holds if and only if
      $C / \epsilon \le \sqrt{n} |\bH|^{1/4} \lambda_{\bH}^{1/2} / \sqrt{ \log U |\bH|^{-1/2}} \to \infty$.
    \end{longform}
  \end{mylongform}
  This yields \eqref{eq:29}.
\end{proof}
The following is referred to as the $\epsilon$-Neighborhood Theorem by
\cite{Guillemin:1974ti}. It states that for certain manifolds, so-called Tubular Neighborhoods exist.
\begin{theorem}[page 69, \cite{Guillemin:1974ti}]
  \label{thm:eps-nbhd}
  For a compact boundaryless manifold $Y$ in $\RR^d$ and $\epsilon > 0$, let $Y^{\epsilon}$ be the open set of points in $\RR^d$ with distance less than $\epsilon$ from $Y$.  If $\epsilon$ is small enough, then each point $w \in Y^\epsilon$ possesses a unique closest point in $Y$, denoted $\pi(w)$. %
  Moreover, the map $\pi: Y^\epsilon \to Y$ is a submersion.
\end{theorem}
A map between manifolds is a submersion if, at all points, the Jacobian map between corresponding tangent spaces is of full rank; see page 20 of \cite{Guillemin:1974ti}.

The following is referred to as the $\epsilon$-Neighborhood Theorem by
\cite{Guillemin:1974ti}. It states that for certain manifolds, so-called Tubular Neighborhoods exist.
\begin{theorem}[page 69, \cite{Guillemin:1974ti}]
  \label{thm:eps-nbhd}
  For a compact boundaryless manifold $Y$ in $\RR^d$ and $\epsilon > 0$, let $Y^{\epsilon}$ be the open set of points in $\RR^d$ with distance less than $\epsilon$ from $Y$.  If $\epsilon$ is small enough, then each point $w \in Y^\epsilon$ possesses a unique closest point in $Y$, denoted $\pi(w)$. %
  Moreover, the map $\pi: Y^\epsilon \to Y$ is a submersion.
\end{theorem}
A map between manifolds is a submersion if, at all points, the Jacobian map between corresponding tangent spaces is of full rank; see page 20 of \cite{Guillemin:1974ti}.

\begin{theorem}[Taylor's Theorem in Several Variables]
  \label{thm:taylor}
  Suppose
  $f:\RR^n\rightarrow\RR$ is of class $C^{k+1}$ on an open convex
  set $S$. If $\bs{a}\in S$ and $\bs{a}+\bs{h}\in S$, then
  \begin{align}
    f(\bs{a}+\bs{h})=\sum_{|\alpha|\le
    k}\frac{\partial^{\alpha}f(\bs{a})}{\alpha!}\bs{h}^{\alpha}+R_{\bs{a},k}(\bs{h}),
  \end{align}
  where the remainder is given in Lagrange's form by
  \begin{align}
    R_{\bs{a},k}(\bs{h})=\sum_{|\alpha|=k+1}\partial^{\alpha}f(\bs{a}+c\bs{h})\frac{\bs{h}^{\alpha}}{\alpha!}
  \end{align}
  for some $c\in (0,1)$.
\end{theorem}
\begin{lemma}
  \label{lem:algebra}
  Let $\bx=(x_1,x_2,\ldots,x_d)'$ be  a $d$-dimensional vector and
  $\bs{A}=\{a_{ij}\}$ be a $d\times d$ matrix. Then
  $|\bx'\bs{A}\bx|\le d\|\bs{A}\|_{\infty}\|\bx\|^2$, where $\|\boldsymbol{A}\|_{\infty}=\max_{i,j}|a_{ij}|$.
\end{lemma}
\begin{proof}
  We have
  \begin{align*}
    |\bx'\boldsymbol{A}\bx|\le
    &\sum_{i,j}|a_{ij}x_ix_j|
      \le \sum_{i,
      j}|a_{ij}|\frac{x_i^2+x_j^2}{2}
      \le \|\boldsymbol{A}\|_{\infty}\sum_{i,j}\frac{x_i^2+x_j^2}{2}
      =d\|\boldsymbol{A}\|_{\infty}\|x\|^2.
  \end{align*}
\end{proof}

\begin{lemma}
  \label{lem:hdr-step3-c2}
  Let Assumption \ref{assm:DA-hdr} and \ref{assm:BA} hold, the for
  $\delta_n>0$ small enough, there exists constant $c_2>0$ and
  another sequence $\varepsilon_n>0$ such that
  $\varepsilon_n=c_2\delta_n$ and $|f_0(\bx)-f_{\tau_,0}|\ge
  \varepsilon_n$ when
  $\bx\in(\mc{L}_{\delta_n}(f_{\tau,0})^c\backslash
  \mc{L}_{\delta}(f_{\tau,0})^c)\cup(\mc{L}_{-\delta_n}(f_{\tau,0})\backslash
  \mc{L}_{-\delta}(f_{\tau,0}))$.
\end{lemma}
\begin{proof}
  The existence of such $c_2$ can be proved by Theorem~\ref{thm:eps-nbhd}, which says for all $\delta>0$ sufficiently
  small, then for each $\bx\in \bigcup_{\bs{y}\in
    \beta}B(\bs{y},\delta)$ there exist a unique $\bs{\theta}\in
  I_d$ and $|s|\le \delta$ such that
  $\bx=\bs{y}(\bs{\theta})+s\bs{u}(\bs{\theta})$, where
  \begin{align*}
    \bs{u}(\bs{\theta})=-\frac{\nabla f_0(\bs{y})}{\|\nabla f_0(\bs{y})\|},
  \end{align*}
  is outer unit normal  vector of $\beta_{\tau}$ at
  $\bs{y}\equiv \bs{y}(\bs{\theta})$. And here we pick $\delta>0$ sufficiently
  small such that not only the Tubular Neighborhood Theorem (Theorem~\ref{thm:eps-nbhd})  but also the
  following hold:
  When $\|\bs{y}_1-\bs{y}_2\|\le \delta$, $|\frac{\partial
    f_0(\bs{y}_1)}{x_i}-\frac{\partial
    f_0(\bs{y}_2)}{x_i}|\le \gamma, i=1,2,\ldots,d$, for some $\gamma>0$.

  Note these two conditions are both feasible because under
  Assumption \ref{assm:DA-hdr}, $f_0$ has two continuous bounded
  derivatives, which indicates both $f_0$ and $\nabla f_0$ are Lipschitz.
  Then for  $\bx\in(\mc{L}_{\delta_n}(f_{\tau,0})^c\backslash
  \mc{L}_{\delta}(f_{\tau,0})^c)\cup(\mc{L}_{-\delta_n}(f_{\tau,0})\backslash
  \mc{L}_{-\delta}(f_{\tau,0}))$,
  \begin{align*}
    |f_0(\bx)-f_{\tau,0}|=|f_0(\bs{y}+s\bs{u})-f_0(\bs{y})|=|\nabla f_0(\bs{\xi})'\bs{u}s|,
  \end{align*}
  where $\bs{\xi}=\bs{y}+ls\bs{u}$ for some $0\le l\le 1$, $\bs{y}\in \beta_{\tau}$. So
  \begin{align*}
    |f_0(\bx)-f_{\tau,0}|=&\left|\nabla f_0(\bs{\xi})'\frac{\nabla
                            f_0(\bs{y})}{\|\nabla f_0(\bs{y})\|}s\right|.
  \end{align*}
  Note that
  \begin{align*}
    |\nabla f_0(\bs{\xi})'\nabla f_0(\bs{y})|=&|\|\nabla
                                                f_0(\bs{y})\|+(\nabla
                                                f_0(\bs{\xi})-\nabla
                                                f_0(\bs{y}))'\nabla f_0 (\bs{y})|
  \end{align*}
  Let $b:=\inf_{\bs{y}\in \beta_{\tau}}\|\nabla f_0(\bs{y})\|$, so by
  Assumption \ref{assm:DA-hdr}, $b>0$. Then by Cauchy-Schwarz inequality
  \begin{align*}
    |(\nabla f_0(\bs{\xi})-\nabla f_0(\bs{y}))'\nabla f_0 (\bs{y})|\le
    \|\nabla f_0(\bs{\xi})-\nabla f_0(\bs{y})\|\|f_0(\bs{y})\|\le
    \sqrt{d}\gamma b
  \end{align*}
  We can choose $\gamma>0$ sufficiently small such that $     |\nabla
  f_0(\bs{\xi})'\nabla f_0(\bs{y})|\ge \inv{2}b$. Then since $\|\nabla
  f_0(\bs{y})\|$ is bounded,  $
  |f_0(\bx)-f_{\tau,0}|\ge \frac{1}{2\sup_{\bs{y}\in\beta}\|\nabla f_0(\bs{y})\|}|s|$. Now for $\bx\in (\mc{L}_{\delta_n}(f_{\tau,0})^c\backslash
  \mc{L}_{\delta}(f_{\tau,0})^c)\cup(\mc{L}_{-\delta_n}(f_{\tau,0})\backslash
  \mc{L}_{-\delta}(f_{\tau,0}))$, $|s|\ge \delta_n$, so
  $|f_0(\bx)-f_{\tau,0}|\ge \varepsilon_n=\frac{1}{2\sup_{\bs{y}\in\beta}\|\nabla f_0(\bs{y})\|}\delta_n$.\\
\end{proof}

\begin{lemma}
  \label{lem:normal-integral}
  Let $a<0$ and $b\in\bb{R}$ be two constants, then
  \begin{align*}
    \int_{\bb{R}}|\Phi(ax+b)-\one_{\{x<0\}}|\,dx=\frac{2\phi(b)+2\Phi(b)b-b}{-a}.
  \end{align*}
\end{lemma}
\begin{proof}
  Note
  \begin{align*}
    \int_{\RR}|\Phi(ax+b)-\one_{\{x<0\}}|\,dx
    &=\int_{-\infty}^0(1-\Phi(ax+b))\,dx+\int_0^{\infty}\Phi(ax+b)\,dx.
  \end{align*}
  And
  \begin{equation*}
    \int_{-\infty}^0(1-\Phi(ax+b)\,dx =
    x(1-\Phi(ax+b))|_{-\infty}^0+\int_{-\infty}^0x\phi(ax+b)a\,dx
  \end{equation*}
  which equals
  \begin{align*}
    \int_{-\infty}^0x\phi(ax+b)a\,dx
    =\inv{a}\int_{\infty}^b(y-b)\phi(y)\,dy
    & = -\frac{1-\Phi(b)}{a}\int_b^{\infty}(y-b)\frac{\phi(y)}{1-\Phi(b)}\,dy \\
    & =-\frac{1-\Phi(b)}{a}\left(\frac{\phi(b)}{1-\Phi(b)}-b\right)\\
    & =-\frac{\phi(b)-(1-\Phi(b))b}{a}.
  \end{align*}
  Also,
  \begin{equation*}
    \int_0^{\infty}\Phi(ax+b)\,dx
    =x\Phi(ax+b)|_0^{\infty}-\int_0^{\infty}ax\phi(ax+b)\,dx
  \end{equation*}
  which equals
  \begin{align*}
    -\int_0^{\infty}ax\phi(ax+b)\,dx
    =\inv{a}\int_{-\infty}^b(y-b)\phi(y)\,dy
    &  = \frac{\Phi(b)}{a}\int_{-\infty}^b(y-b)\frac{\phi(y)}{\Phi(b)}\,dy \\
    & =\frac{\Phi(b)}{a}\left(\frac{-\phi(b)}{\Phi(b)}-b\right)\\
    & =\frac{-\phi(b)-\Phi(b)b}{a}.
  \end{align*}
  Thus
  \begin{align*}
    \int_{\RR}|\Phi(ax+b)-\one_{\{x<0\}}|\,dx=\frac{2\phi(b)+2\Phi(b)b-b}{-a}.
  \end{align*}
\end{proof}

\noindent Recall that $\beta^\delta := \cup_{\bx \in \beta} B(\bx, \delta)$ and that we let $u_{\bx}$ be the unit outer normal vector to the manifold $\beta$ at $\bx$.  The following lemma gives a very useful approximate change of variables type of theorem.
\begin{mylongform}
  \begin{longform}
    It holds for any $h$ such that $h \circ \Phi$ Lebesgue integrable on $N_d^\delta $(with $\Phi$ and $N_d^\delta$ defined below).  We only use it for bounded $h$.
  \end{longform}
\end{mylongform}

\begin{lemma}
  \label{lem:COV-approx}
  Let either Assumption~\ref{assm:DA-ls} or Assumption~\ref{assm:DA-hdr}
  hold, and let \ref{assm:BA} hold for the density $f_0$.
  Let either $\beta := f_0^{-1}(c)$ in the LS setting or let  $\beta :=
  f_0^{-1}( \fftau)$ in the HDR setting.  Let $\delta > 0$ be such that the conclusion of
  Theorem~\ref{thm:eps-nbhd} holds for $\beta^\delta$.  Let $h$ be a bounded Lebesgue measurable
  function on $\beta^\delta$ and let $H(\bx) : = \int_{-\delta}^\delta h(\bx
  + t u_x) dt$.  Then
  \begin{equation}
    \label{eq:19}
    \lv \int_{\beta^\delta} h(\bx) d\bx
    - \int_{\beta} H(\bz) d\cH^{d-1}(\bz) \rv
    \le C \sup_{\bx \in \beta} \int_{-\delta}^\delta t h( \bx + t u_{\bx}) dt
  \end{equation}
  where $C$ is a constant depending  on $f_0$.
\end{lemma}
\begin{proof}
  Since $\beta$ is compact (Assumption~\ref{assm:DA-hdr}), it admits a finite
  ``atlas'', $\lb (U^\alpha, \vp_\alpha) \rb_\alpha$, meaning $\lb U^\alpha
  \rb_\alpha$ is an open cover of $\beta $, that $\vp_{\alpha} : V^\alpha
  \to U^\alpha$ is a diffeomorphism, and that $V^\alpha$ is open in $\RR^{d-1}$.
  Let $V^\alpha_\delta := V^\alpha \times (-\delta,\delta)$.  Let
  $\Phi_\alpha: V^\alpha_\delta \to \beta^\delta$ be defined by
  \begin{equation*}
    \Phi_\alpha(\bs \theta,t) := \vp_\alpha(\bs{\theta}) + t u_{\vp_\alpha(\bs \theta)}
    \quad \mbox{ where } \quad
    u_{\bx} := - \frac{ \grad f_0}{\| \grad f_0 \|}(\bx).
  \end{equation*}
  Thus $u_{\bx}$ is the unit outer normal to $\beta$ at $\bx \in\beta$.
  By the change of variables Theorem 2 (page 99) of \cite{Evans:2015uy} (see also the example on page 101),
  \begin{equation}
    \label{eq:2}
    \int_{V^\alpha} h(\vp_\alpha(\bs{\theta}))) J\vp_\alpha(\bs{\theta}) d\bs{\theta}
    = \int_{U^\alpha} h(\bs{y}) d\mc{H}^{d-1}(\bs{y}).
  \end{equation}
  \begin{mylongform}
    \begin{longform}
      By the change of variables Theorem 2 (page 99) of \cite{Evans:2015uy} (see also the example on page 101), for any $g$ with $\int_{\PC[d-1]} |g(x)| dx  < \infty$,
      \begin{equation}
        \int_{\PC[d-1]} g(x) J\vp(x) dx
        = \int_\beta g( \vp^{-1}(y)) d\mc{H}^{d-1}(y).
      \end{equation}
      Note that in \cite{Evans:2015uy}, a function $g$ is $\mc{L}$-integrable if the integral is well-defined, although possibly $\pm \infty$. Then $g$ is    $\mc{L}$-summable  if it is integrable and the integral is finite.  This is defined on page 18.
      The formula for $J\vp $ is
      \begin{equation}
        \label{eq:12}
        J\vp = [[ D\vp ]] = \sqrt{ \det \ls \grad \vp' \grad \vp \rs }.
      \end{equation}
    \end{longform}
  \end{mylongform}
  Here,
  \begin{equation}
    \label{eq:3}
    J\vp_\alpha(\bs{\theta}) :=  \det \ls (\grad \vp_\alpha(\bs{\theta}))' \grad \vp_\alpha(\bs{\theta}) \rs^{1/2}
  \end{equation}
  by Theorem~3 (page 88) of  \cite{Evans:2015uy}.
  \begin{mylongform}
    \begin{longform}
      $J\vp$ is defined on page 91.  (The above formula is then given as Theorem 3 on page 88 for ``[[ ]]''.)
    \end{longform}
  \end{mylongform}
  Similarly,
  \begin{equation}
    \label{eq:10}
    \int_{V_\delta^\alpha} h(\Phi_\alpha((\bs{\theta},t)) J\Phi_\alpha(\bs{\theta},t) d(\bs{\theta},t)
    = \int_{U_\delta^\alpha}  h(\bs{y}) d\bs{y}
  \end{equation}
  where
  $    J\Phi_\alpha    = | \det \nabla\Phi_\alpha |$
  and $U_\delta^\alpha := \Phi_\alpha(V_\delta^\alpha)$.
  \begin{mylongform}
    \begin{longform}
      This is from
      \begin{equation}
        \label{eq:13}
        J\Phi_\alpha = \det \ls \nabla\Phi_\alpha' \nabla\Phi_\alpha \rs^{1/2}
        = | \det \nabla\Phi_\alpha |.
      \end{equation}
    \end{longform}
  \end{mylongform}
  We can see that
  \begin{equation}
    \label{eq:14}
    \nabla\Phi_\alpha(\bs{\theta},t) =
    \begin{pmatrix}
      \nabla \vp_\alpha(\bs{\theta}) + t \nabla u_{\vp_\alpha({\bs{\theta}})} \vert u_{\vp_\alpha(\bs{\theta})}
    \end{pmatrix}.
  \end{equation}
  Thus,  because $u_{\bx}$ is perpendicular to the tangent space of $\beta$ at ${\bx}$,
  and this tangent space is equal to the span of the columns of $\nabla\vp_\alpha({\bs \theta})$
  for $t \in [-\delta,\delta]$, letting $\bx = \vp_\alpha(\bs{\theta})$,  we have
  \begin{equation}
    \label{eq:16}
    \nabla\Phi_\alpha({\bs{\theta}},t)' \nabla\Phi_\alpha({\bs{\theta}},t)
    =
    \begin{pmatrix}
      A_t & t \nabla u_{\bx}' u_{\bx} \\
      t u_{\bx}' \nabla u_{\bx}   & 1
    \end{pmatrix}
  \end{equation}
  where
  \begin{equation}
    \label{eq:17}
    A_t := \nabla\vp_\alpha({\bs{\theta}})' \nabla\vp_\alpha({\bs{\theta}})
    + t \nabla\vp_\alpha({\bs{\theta}})'\nabla u_{\bx}
    + t \nabla u_{\bx}' \nabla\vp_\alpha({\bs{\theta}}) + t^2 \nabla u'_{\bx} \nabla u_{\bx}.
  \end{equation}
  \begin{mylongform}
    \begin{longform}
      In particular,
      \begin{equation}
        \label{eq:15}
        \nabla \Phi_\alpha({\bs{\theta}},0)' \nabla \Phi_\alpha({\bs{\theta}},0)
        =
        \begin{pmatrix}
          \nabla\vp_\alpha({\bs{\theta}})' \nabla\vp_\alpha({\bs{\theta}}) & 0 \\
          0 & 1
        \end{pmatrix}.
      \end{equation}
    \end{longform}
  \end{mylongform}
  Note that from \eqref{eq:16} we have
  \begin{equation}
    \label{eq:18}
    J\Phi_\alpha({\bs{\theta}},0) = J\vp_\alpha({\bs{\theta}}).
  \end{equation}
  Now
  \begin{equation}
    \label{eq:20}
    \det (A + \epsilon A X) = \det A + \epsilon \det A \tr X + O(\epsilon^2)
  \end{equation}
  as $\epsilon \to 0$   \citep{Magnus:1999vh} for any square matrices $A$ and $X$ of the same dimension.
  Thus
  \begin{align*}
    J\Phi_\alpha({\bs{\theta}},t)
    &= \lp \det \nabla\Phi_\alpha({\bs{\theta}},t)' \nabla\Phi_\alpha({\bs{\theta}},t) \rp^{1/2} \\
    & = \lp \det \nabla\Phi_\alpha({\bs{\theta}},0)' \nabla\Phi_\alpha({\bs{\theta}},0)  + O(t) \rp^{1/2}
      \quad \mbox{ by \eqref{eq:20}, \eqref{eq:16}, and \eqref{eq:17},}
  \end{align*}
  \begin{mylongform}
    \begin{longform}
      which equals
      \begin{align*}
        \lp J\Phi_\alpha({\bs{\theta}},0)^2  + O(t) \rp^{1/2}
      \end{align*}
    \end{longform}
  \end{mylongform}
  which equals
  \begin{equation}
    \label{eq:21}
    J\Phi_\alpha({\bs{\theta}},0) + O(t)
    \quad \mbox{ as } \quad t \to 0,
  \end{equation}
  by differentiability of $z \mapsto z^{1/2}$ away from $0$, since $J\Phi_\alpha({\bs{\theta}},0)$ is uniformly bounded away from $0$.
  \begin{mylongform}
    \begin{longform}
      (Need add/state the conditions from the last few lines
      giving uniform-in-${\bs{\theta}}$?)
    \end{longform}
  \end{mylongform}
  The $O(t)$ term is uniform in $\bs \theta$.
  Thus, by \eqref{eq:18}, \eqref{eq:2}, and \eqref{eq:10},
  \begin{equation}
    \label{eq:39}
    \int_{U_\delta^\alpha} h(\bs y) \, d\bs y
    = \int_{U^\alpha} H( \bs y) \, d \cH^{d-1}(\bs y)
    +     E
  \end{equation}
  where $|E| \le C \int_{I^d} \int_{-\delta}^\delta t h(\Phi_\alpha(\bs
  \theta, t)) dt d\bs \theta$ where $C$ is the constant from the $O(t)$ term
  in \eqref{eq:21}.  This proves the lemma if $\beta$ is parameterizable by a
  single open set; for the general case, we use a partition of unity.  Let
  $\lb \rho_i \rb$ be a finite (smooth) partition of unity subordinate to
  $\lb U^\alpha \rb$ \citep[page 63]{Spivak:1965wf}.  Define $\rho_i^\delta
  (\bx + t u_{\bx}) := \rho_i(\bx)$ for $t \in (-\delta,\delta)$ (which thus
  forms a partition of unity of $\beta^\delta$ subordinate to $\{ U_\delta^\alpha \}_\alpha$).  Then replacing $h$ in
  \eqref{eq:39} by $\rho_i^\delta \cdot h$, since each $\rho_i$ is bounded,
  smooth, and zero outside one of the $U_{\alpha}$,
  \begin{align*}
    \int_{\beta^\delta} h(\bs y) d\bs y
    = \sum_i \int_{\beta^\delta} \rho_i^\delta(\bs y) h(\bs y) d\bs y
    & = \sum_i \int_{\beta} \rho_i H d\mc{H}^{d-1}
      + E_2
    = \int_{\beta} H d\mc{H}^{d-1}  + E_2
  \end{align*}
  since $ \rho_i^\delta(\bx + t u_{\bx}) = \rho_i(\bx),$ and where
  $|E_2| \le C_2 \sup_{\bx \in \beta} \int_{-\delta}^\delta t h(\bx + tu_{\bx}) dt$.
\end{proof}
\begin{mylongform}
  \begin{longform}
    The equality $ \sum_i \int_{\beta} \rho_i H d\mc{H}^{d-1} = \int_{\beta}
    H d\mc{H}^{d-1} $ (and the similar one for the integral over
    $\beta^\delta$) is just by definition of the integral, which is defined
    via a partition of unity.  See spivak.  I don't know if evans and gariepy
    explicitly do this or not.
  \end{longform}
\end{mylongform}

\begin{lemma}
  \label{lem:Hauss-diff}
  Let Assumption~\ref{assm:DA-ls} hold.
  \begin{enumerate}
  \item \label{lem:Hauss-diff-a} Assume that $\gamma$ is a continuously
    differentiable  function on an open neighborhood of $\beta_\tau$ in
    $\RR^d$. For $\epsilon$ near $0$, let
    $\beta_\epsilon : = f_0^{-1}(\fftau + \epsilon)$ and assume
    $\beta_\epsilon$ is compact for all $\epsilon$ in a neighborhood of $0$.
    Then $\epsilon \mapsto \int_{\beta_\epsilon} \gamma d \cH$ is continuously differentiable in a neighborhood of $\epsilon=0$.
  \item \label{lem:Hauss-diff-b} Let $\ffnH$ be the KDE (defined in
    \eqref{eq:kde-def}), where $K$ satisfies Assumptions~\ref{assm:KA} and
    \ref{assm:KA2}, and $\bH$ satisfies $\bH \to 0$ and $n^{-1} | \bH|^{-1/2} (\bH^{-1})^{\otimes
      2} = O(1)$ as $n \to \infty$.  Let $g_n := \ffnH - f_0$.
    Let $\check{\beta}_{\tau,n} := \ffnH^{-1}(\fftau).$ Assume
    $\gamma_n \equiv \gamma$ is potentially random but satisfies
    $\sup_{\bx \in \beta_\tau^\delta} | \gamma(\bx) | = O_p(1)$ and
    $\sup_{\bx \in \beta_\tau^\delta} \| \grad \gamma(\bx) \| = O_p(1)$, for
    some $\delta > 0$.  Then
    \begin{align*}
   \MoveEqLeft   \lv \int_{\beta_\tau} \gamma d\cH^{d-1}
   - \int_{\check{\beta}_{\tau,n}} \gamma d\cH^{d-1} \rv \\
   &=
      O_p (
      \sup_{\bx \in \beta_\tau} \EE \ls | g_n (\bx) | +
      \| \hess g_n (\bx) \| |g_n(\bx) | + \| \nabla g_n(\bx) \|  \rs )
    \end{align*}
    as $n \to \infty$.
  \end{enumerate}
\end{lemma}
\begin{proof}
  {\bf Proof of Part~\ref{lem:Hauss-diff-a}}:
  Fix $\bx_0 \in \beta_0$.  By
  Assumption~\ref{assm:DA-ls}, we may assume without loss of generality that $\pderiv{x_d}f(\bx_0) \ne 0$. Define
  $$F(x_1, \ldots, x_d) :=
  (x_1, \ldots, x_{d-1}, f(x_1, \ldots, x_d))$$ and note that $\det \grad F(\bx_0) = \pderiv{x_d}f(\bx_0) \ne 0$.  Since $f$ is twice continuously differentiable at $x_0$ (Assumption~\ref{assm:DA-ls}), $F$ is twice continuously differentiable at $x_0$.  By the inverse function theorem (pages 67--68, \cite{MR1224675}), $F^{-1}$ exists and is twice continuously differentiable in a neighborhood of $F(\bx_0)$.  Clearly $F^{-1}(y_1, \ldots, y_d)$ equals $(y_1, \ldots , y_{d-1}, k(y_1, \ldots, y_d))$ for some $k$ that is twice continuously differentiable and satisfies
  \begin{equation*}
    f(y_1, \ldots, y_{d-1}, k(y_1, \ldots, y_{d})) = y_d.
  \end{equation*}
  Thus
  \begin{equation*}
    \vp_\epsilon (y_1, \ldots, y_{d-1}) := (y_1, \ldots, y_{d-1}, k(y_1, \ldots, y_{d-1}, \fftau + \epsilon))
  \end{equation*}
  is a twice-continuously differentiable invertible parameterization (is a ``$C^2$ diffeomorphism'') from an open set $U \subset \RR^{d-1}$ to $V_\epsilon \subset \beta_\epsilon$ where $V_\epsilon \ni \bx_0$ is open in $\beta_\epsilon$.
  \begin{mylongform}
    \begin{longform}
      (Open in the subspace topology.)
    \end{longform}
  \end{mylongform}
  Each $\bx_0 \in \beta_\epsilon$ has such a $C^2$ diffeomorphism onto an open neighborhood $V_\epsilon \subset \beta_\epsilon$; since $\beta_\epsilon$ is compact, we can pick a finite number of them that cover $\beta_\epsilon$ and construct a partition of unity \citep[page 63]{Spivak:1965wf} on the cover.
  We will continue considering our fixed $\bx_0 \in \beta_\epsilon$ and the above-constructed parameterization on a neighborhood of $\bx_0$.  At the end of the proof, our local result can be made global by using the partition of unity.

  Now, $\int_{\beta_\epsilon} \gamma d \cH = \int_U (\gamma \circ \vp_\epsilon) J\vp_\epsilon d \lambda^{d-1}$ where $\lambda^{d-1}$ is Lebesgue measure
  \citep{Evans:2015uy}.  Here $J \vp_{\epsilon} = \det ( \grad \vp_\epsilon' \grad \vp_\epsilon)^{1/2}$ is continuously differentiable in $\epsilon$ (in a neighborhood of $0$) since $k$ is twice continuously differentiable and since $\det (\grad \vp_\epsilon' \grad \vp_\epsilon) \ne 0$.
  \begin{mylongform}
    \begin{longform}
      Since $\vp_\epsilon$ is a diffeomorphism, so is of full rank $d-1$.
    \end{longform}
  \end{mylongform}
  We also know that $\gamma \circ \vp_\epsilon$ is continuously differentiable in $\epsilon$ since $\gamma$ is assumed continuously differentiable.
  \begin{mylongform}
    \begin{longform}
      (Note we need $\gamma$ to be differentiable; the setup is quite different than that of the Fundamental Theorem of Calculus.)
    \end{longform}
  \end{mylongform}
  Since $\pderiv{\epsilon} ((\gamma \circ \vp_\epsilon) J \vp_\epsilon)$ is continuous so is bounded on $U \times [-\tilde{\epsilon}, \tilde{\epsilon}]$, some $\tilde{\epsilon}>0$, we can apply the Leibniz rule \citep{Billingsley:2012ti} to see that
  \begin{equation*}
    \pderiv{\epsilon} \int_{V_\epsilon} \gamma  \, d\cH^{d-1}
    =  \pderiv{\epsilon} \int_U (\gamma \circ \vp_\epsilon) J\vp_\epsilon d \lambda^{d-1}
    =  \int_U  \pderiv{\epsilon} ( (\gamma \circ \vp_\epsilon) J\vp_\epsilon ) d \lambda^{d-1}
  \end{equation*}
  \begin{mylongform}
    \begin{longform}
      (Note that to say that $J\vp_\epsilon$ is defined on $U \times [-\tilde{\epsilon}, \tilde{\epsilon}]$ requires an explicit construction such as the one we made above.)
    \end{longform}
  \end{mylongform}
  Thus, the derivative on the left side of the previous display exists, meaning that $\int_{\beta_\epsilon} \gamma \, d\cH$ is indeed differentiable for $\epsilon$ near $0$, as desired.  This is true on the neighborhood $V_\epsilon$; it extends to the case where $V_\epsilon$ is replaced by $\beta_\epsilon$ by using the partition of unity we constructed above.

  \smallskip

  {\bf Proof of Part~\ref{lem:Hauss-diff-b}}:
  We write $g \equiv g_n$, suppressing dependence on $n$.  For $\bx \in
  \RR^d$, let $h(\bx,\delta) := f_0(\bx) + \delta g(\bx)$, and let
  $\beta_{\delta} := h_\delta^{-1}( \fftau)$. We will explicitly construct
  $\phi_\delta: U \to V_\delta$, for some open $U \subset \RR^{d-1}$ and
  $V_\delta \subset \beta_\delta$, by the inverse function theorem, and then
  check that $\pderiv{ \delta}\phi_{\delta}(\bs z)$ is $O_p( | g(
  \phi_\delta(\bs z)) |)$ and that $\pderiv{\delta} J\phi_{\delta}(\bs z)$
  is $O_p( \| \hess g( \phi_\delta(\bs z)) \| | g( \phi_\delta(\bs z)) |
  + | g( \phi_\delta(\bs z)) | +
  \| \nabla g (\phi_\delta( \bs z)) \| )$.  Then the proof can be finished as
  the proof of the previous part was finished.

  Fix $\bx_0 \in \beta_\tau \equiv \beta_0$.  Define $F(x_1, \ldots, x_d,
  \delta) := (x_1, \ldots, x_{d-1}, h(\bx, \delta), \delta)$.  As in the
  proof of the previous part, note that $\det \nabla F(\bx_0) \ne 0$ (when
  $\| \nabla g(\bx_0 \|$ is small), so by the inverse function theorem
  $F^{-1}$ exists, is twice continuously differentiable in a neighborhood of
  $F(\bx_0)$, and clearly satisfies $F^{-1}(y_1, \ldots, y_d, \delta) = (y_1,
  \ldots, y_{d-1}, k(y_1, \ldots, y_d, \delta),\delta)$.  Let $\bs z := (\bx,
  \delta)$ and note by definition
  \begin{equation}
    \label{eq:8}
    k(F(\bs z)) = k(x_1, \ldots, x_{d-1}, h(\bs z), \delta) = x_d.
  \end{equation}
  From this we will derive formulas for the first and second derivatives of $k$.  In this proof, for a function $f \colon \RR^p \to \RR$ we use the notation $f_i(\bx)$ for $\pderiv{x_i}f(x_1, \ldots, x_d)$ and $f_{ij}(\bx)$ for $\derivtwo{x_i}{x_j} f(x_1, \ldots, x_d)$.
  Taking $\pderiv{x_i}$ of  \eqref{eq:8} for $1 \le i \le d-1$, we see that
  \begin{equation}
    \label{eq:9}
    k_i( F(\bs z)) = - k_d (F(\bs z)) h_i(\bx).
  \end{equation}
  Applying $\pderiv{ x_d}$ to \eqref{eq:8}, we get that
  \begin{equation}
    \label{eq:32}
    k_d( F(\bs z)) h_d (\bs z) = 1
    \quad \text{ or } \quad
    k_d(F(\bs z) ) = 1/ h_d( \bs z),
  \end{equation}
  and applying $\pderiv{ \delta}$  to \eqref{eq:8}, we get
  \begin{equation}
    \label{eq:35}
    k_d(F(\bs z)) h_{d+1}(\bs z) + k_{d+1}(F(\bs z)) = 0,
    \quad    \text{ or } \quad
    k_{d+1}(F(\bs z)) = - \frac{h_{d+1}(\bs z)}{h_d(\bs z)}
    = - \frac{g(\bs x)}{h_d(\bs z)}.
  \end{equation}
  Applying $\pderiv{\delta} $ to \eqref{eq:9} yields
  \begin{equation}
    \label{eq:30}
    \begin{split}
      \MoveEqLeft      k_{i,d}(F(\bs z)) h_{d+1}(\bs z) + k_{i,d+1}(F(\bs z))  \\
      & =  - \big( k_{d,d}(F(\bs z)) h_{d+1}(\bs z) + k_{d,d+1} (F(\bs z)) \big)
      h_i( \bs z) -
      k_d(F(\bs z)) h_{i,d+1}(\bs x)
    \end{split}
  \end{equation}
  and,
  letting $\bs y := F(\bs z)$,
  since $h_{d+1}(\bs z) = g(\bs x)$ and $h_{i,d+1}(\bs z) = g_i(\bs x)$,
  this implies that
  \begin{mylongform}
    \begin{longform}
      \begin{equation*}
        \begin{split}
          \MoveEqLeft      k_{i,d}(F(\bs z)) h_{d+1}(\bs z) + k_{i,d+1}(F(\bs z))  \\
          & =
          - ( k_{d,d}(F(\bs z)) g( \bs x) +
          k_{d,d+1}(F( \bs z)) ) h_i (
          \bs z) - k_d( F(\bs z)) g_i ( \bx),
        \end{split}
      \end{equation*}
      so
    \end{longform}
  \end{mylongform}
  \begin{equation}
    \label{eq:31}
    k_{i,d+1}(\bs y)=
    - k_{i,d}(\bs y) g(\bs x)
    -  \Big( k_{d,d}(\bs y)  g( \bs x) + k_{d,d+1}( \bs y) \Big) h_i ( \bs z)
    - k_d( \bs y) g_i ( \bx).
  \end{equation}
  To understand the expression in \eqref{eq:31} we need to control $k_{i,d}$, $k_{d,d}$, and $k_{d,d+1}$.  Applying $\pderiv{\delta}$ to \eqref{eq:32} we see that
  \begin{equation*}
    k_{d,d}(F(\bs z)) h_{d+1}(\bs z) + k_{d,d+1}(F(\bs z))
    = - \frac{ h_{d,d+1}(\bs z)}{h_d^2(\bs z)},
  \end{equation*}
  so
  \begin{equation}
    \label{eq:38}
    k_{d,d+1}(F(\bs z)) = k_{d,d}(F(\bs z)) g( \bs x) - \frac{ g_d(\bs x)}{ h_d^2 (\bs z)}.
  \end{equation}
  We will next verify that $k_{i,d}$ and $k_{d,d}$
  are $O_p(1 + \| \hess g \|)$ (which is $O_p(1)$ under our assumption on $\bH$
  \citep{chacon2011asymptotics}). Then by \eqref{eq:31} and \eqref{eq:38}, we will see, uniformly for $\delta \in [-1,1]$, that
  \begin{equation}
    \label{eq:34}
    k_{i,d+1}(F(\bs z)) = O_p( | g(\bx) | + \| \nabla g(\bx) \|  + \| \hess g(\bx) \| | g(\bx)|  )
    \quad \text{ as } \quad
    n \to \infty.
  \end{equation}
  Note that by \eqref{eq:32}, $k_d(F( \bs z)) = O_p(1)$ and
  $1/h_d(\bs z) = O_p(1)$ (since by assumption $\pderiv{x_d} f( \bs x_0) \ne 0$
  and $\| \grad g (\bx) \| \to_p 0$).

  Now applying $\partial / \partial x_d$ to \eqref{eq:32}, we see
  \begin{equation}
    \label{eq:37}
    k_{d,d}(F(\bs z))
    = -\frac{ k_d^2( F (\bs z))  h_{d,d}(\bs z)}{h_d(\bs z)}
    = - \frac{h_{d,d}(\bs z)}{h_d^3( \bs z)},
  \end{equation}
  so $k_{d,d}(F(\bs z)) = O_p(1 + \| \hess g(\bs x) \|)$.
  Applying $\partial/ \partial x_i$ to (the left expression in) \eqref{eq:32} yields
  \begin{equation}
    \label{eq:36}
    k_{i,d}(F(\bs z)) + k_{d,d}( F(\bs z)) h_i(\bs z)
    = -\frac{ h_{d,i}(\bs z)}{ h_d^2(\bs z)}.
  \end{equation}
  Thus by \eqref{eq:37} we see $k_{i,d}( F(\bs z)) = O_p(1 + \| \hess g(\bs
  x) \|)$, so \eqref{eq:34} holds.
  \begin{mylongform}
    \begin{longform}
      To show the three $k$ terms are $O(1)$, possibly could just do this in one or two steps from $F \circ F^{-1} = $Id and differentiating twice and making a more abstract argument about how the second derivatives of $k$ are algebraic expressions in the zeroth, first, second derivatives of $h$?   Still may need to bound denominators away from zero conceivably so maybe necessary to see what those expressions are.
    \end{longform}
  \end{mylongform}

  Now we let
  \begin{equation*}
    \phi_\delta( y_1, \ldots, y_{d-1})
    := (y_1, \ldots, y_{d-1}, k(y_1, \ldots, y_{d-1}, \fftau, \delta)),
  \end{equation*}
  which we have shown is a $C^2$ parameterization from an open set $U \subset \RR^{d-1}$ to $V_\delta \subset \beta_\delta$  where $V_\delta \ni \bx_0$ is open in $\beta_\delta$.  We can check that $J \phi_\delta = \det (\nabla \phi_\delta' \nabla \phi_\delta)^{1/2}$ is continuously differentiable in $\delta $ for $\delta \in [-1,1]$
  by
  \eqref{eq:34},
  and, by  three Taylor expansions,
  \begin{equation}
    \label{eq:41}
    \int_{V_1} \gamma d \cH^{d-1}
    =\int_U ( \gamma \circ \phi_1) J\phi_1 d \lambda^{d-1}
    = \int_U ( ( \gamma
    \circ \phi_0(\bs y) ) J\phi_0(\bs y)
    + \epsilon(\bs y)
    \,
    d \bs y
  \end{equation}
  where $\epsilon(\bs y) = O_p( |g(\bx) | + \| \nabla g (\bx) \| + \| \hess g(\bx) \| |g(\bx)|)$, since $\pderiv{\delta} J \phi_\delta(\bs y) $ is $O_p( |g(\bx) | + \| \nabla g (\bx) \| + \| \hess g(\bx) \| |g(\bx)|)$ uniformly for $\delta \in [-1,1]$, since $\pderiv{\delta}\phi_\delta(\bs y) = O_p( |g(\bx)|)$ uniformly for $\delta \in [-1,1]$ (by \eqref{eq:35}), and since $\gamma$ is continuously differentiable in a neighborhood of $\beta_\tau$.  In fact, we can see that $\EE | \epsilon(\bs y)| \le C \EE \ls |g(\bx) | + \| \nabla g (\bx) \| + \| \hess g(\bx) \| |g(\bx)| \rs$ for a constant $C >0$.  By the Fubini-Tonelli theorem, $\EE \int_U | \epsilon( \bs y) | \, d\bs y = \int \EE | \epsilon(\bs y) | \, d\bs y$, so we can see
  \begin{equation}
    \label{eq:40}
    \int_U \epsilon(\bs y) \, d\bs y
    = O_p \sup_{\bx \in \beta_\tau}
    \EE \ls |g(\bx) | + \| \nabla g (\bx) \| + \| \hess g(\bx) \| |g(\bx)| \rs
  \end{equation}
  by Markov's inequality.
  Combining
  \eqref{eq:41}, \eqref{eq:40}, and
  $  \int_U ( ( \gamma
  \circ \phi_0(\bs y) ) J\phi_0(\bs y) =   \int_{V_0} \gamma d \cH^{d-1} $ we get
  \begin{equation*}
    \int_{V_1} \gamma d \cH^{d-1}
    = \int_{V_0}
    \gamma d\cH^{d-1} +
    O_p \sup_{\bx \in \beta_\tau}
    \EE \ls |g(\bx) | + \| \nabla g (\bx) \| + \| \hess g(\bx) \| |g(\bx)| \rs.
  \end{equation*}
  Then the proof can be
  finished as in the proof of Part~\ref{lem:Hauss-diff-a}, including using a
  partition of unity to extend $V_1$ to $\check \beta_\tau$ and $V_0$ to $\beta_\tau$ to conclude from 
  the previous display %
  that $ \int_{\check \beta_\tau} \gamma d \cH^{d-1} = \int_{\beta_\tau} \gamma
  d\cH^{d-1} +
  O_p \sup_{\bx \in \beta_\tau}
  \EE \ls |g(\bx) | + \| \nabla g (\bx) \| + \| \hess g(\bx) \| |g(\bx)| \rs.$
\end{proof}

\subsection{Proof of Corollary~\ref{cor:hdr-oracle-bandwidth}}

By our assumptions of unimodality and spherical symmetry of $f_0$, we have that $\grad f_0$ and $\hess f_0$ are constant on $\beta_\tau$, and we denote these two quantities as $\grad_\tau f_0$ and $\hess_\tau f_0$.  Then for $h > 0$ we can write
\begin{equation*}
  B(h) =  - (n h^{d+4})^{1/2} F_1
  \quad    \text{ where } \quad
  F_1 :=     \frac{\mu_2(K) \tr( \hess_\tau f_0)}{2 \sqrt{R(K) \fftau}},
\end{equation*}
and $ C(h) = B(h)+(n h^{d+4})^{1/2} F_2 $ where
\begin{equation*}
  F_2 :=
  \| \grad_\tau f_0  \|  \lp \int_{\beta_\tau} d \cH \rp^{-1}
  \lb \frac{\mu_2(K) \tr( \hess_\tau f_0)}{2 \| \grad_\tau f_0 \|}
  \int_{\beta_\tau} d\cH
  + \frac{\mu_2(K)}{2 \fftau} \tr(\grad^2_\tau f_0)  \int_{\mc{L}_\tau} d\bx \rb.
\end{equation*}
\begin{mylongform}
  \begin{longform}
    Some additional computations and substitutions are as follows.  Note that
    \begin{equation*}
      V_1(h) =
      h^2 \int_{\beta_\tau}  \frac{ \mu_2(K) \tr( \hess_\tau f_0)}{ 2 \| \grad_\tau f_0 \| }
      d\cH
    \end{equation*}
    \begin{equation*}
      V_2(h)
      = h^2 \frac{ \mu_2(K)}{2 \fftau} \int_{\mc{L}_\tau} \tr (\hess_\tau f_0) d \bx.
    \end{equation*}
    Since
    \begin{equation*}
      C(h)
      =B(h)+ \sqrt{n} h^{d/2}  \lp \int_{\beta_\tau}  \| \grad_\tau f_0 \|^{-1}  d\cH \rp^{-1}
      ( V_1(h) + V_2(h))
    \end{equation*}
    we get the formula for $F_2$ and $C(h)$ above.
  \end{longform}
\end{mylongform}
Then
\begin{equation*}
  \HDR(h)
  = \frac{ \fftau}{A}
  \lp \int_{\beta_\tau} d\cH \rp
  (n h^d)^{-1/2}
  \lp 2 \phi(C(h))  + (2 \Phi(C(h)) - 1) C(h) \rp
\end{equation*}
where $A=\|\nabla_{\tau}f_0\|/\sqrt{R(K)f_{\tau,0}}$. Note that $2
\phi(C(h))  + (2 \Phi(C(h)) - 1) C(h)=2 \phi(|C(h)|)  + (2
\Phi(|C(h)|) - 1) |C(h)|$. Let $G:=|C(h)|/(nh^{d+4})^{1/2}=|F_2-F_1|$,
We will thus minimize
\begin{equation}
  \label{eq:1}
  \begin{split}
    & n^{2 / (d+4)} \lp \frac{A}{ \fftau} \int_{\beta_\tau} d\cH \rp^{-1}
    \HDR(h)  \\
    & = (n^{1/2}h^{(d+4)/2})^{-d/(d+4)} \lp 2 \phi(G(nh^{d+4})^{1/2}) + G(nh^{d+4})^{1/2} ( 2 \Phi(G(nh^{d+4})^{1/2}) - 1) \rp
  \end{split}
\end{equation}
over $h \ge 0$.
By
the change of variables
\begin{equation}
  \label{eq:5}
  s =  (n h^{d+4})^{1/2},
\end{equation}
minimizing \eqref{eq:1} is equivalent to minimizing
\begin{equation*}
  \HDR^*(s) := 2 s^{-d/ d+4} \phi(Gs) + Gs^{4 / d+4}  (2 \Phi(Gs) - 1).
\end{equation*}
Note that $\HDR^*(s) \to \infty$ as $s \to \infty$ and as $s \searrow 0$, so $\HDR^*(s)$ attains its minimum on $(0, \infty)$.
Now, $\HDR^*$ has a unique minimum if $(\HDR^*)'(s)$ has a unique $0$,
and by calculation,
\begin{align*}
  (\text{HDR}^{\ast})'(s)=2\frac{-d}{d+4}s^{\frac{-2d-4}{d+4}}\phi(Gs)+G\frac{4}{d+4}s^{\frac{-d}{d+4}}(2\Phi(Gs)-1),
\end{align*}
$(\HDR^*)'(s)$ has a unique $0$ if and only if
\begin{equation}
  \label{eq:4}
  (\HDR^*)'(s) (2 (d/d+4) s^{-(2d+4)/d+4} \phi(Gs))^{-1}
  = - 1 + \frac{2}{d} \frac{Gs( 2 \Phi(Gs) - 1) }{\phi(Gs)}
\end{equation}
has a unique $0$.
We can compute the derivative of \eqref{eq:4} to be
\begin{align*}
  G \frac{2}{d} \lp 2Gs + \frac{ (1 + G^2s^2) ( 2 \Phi(Gs) - 1)}{\phi(Gs)} \rp > 0
\end{align*}
for $s \in (0,\infty)$.
Thus \eqref{eq:4} is strictly increasing on $(0,\infty)$,
is negative at $0$, and approaches $\infty$ as $c \to \infty$, and so \eqref{eq:4} has a unique zero.
\begin{mylongform}
  \begin{longform}
    See mathematica notebook for derivative computations.
    We can compute the first and second derivatives of  \eqref{eq:4} to be
    \begin{align*}
      \frac{2}{d} \lp 2s + \frac{ (1 + s^2) ( 2 \Phi(s) - 1)}{\phi(s)} \rp
      \text{ and } \\
      \frac{2}{d} \lp 2(2 + s^2) +  \frac{ 2 \Phi(s) - 1}{\phi(s)} s(3+s^2) \rp > 0.
    \end{align*}
    (We do not actually need the second derivative it seems.)
  \end{longform}
\end{mylongform}
Let $s_{\text{opt}} > 0$ be the unique minimum of $\HDR^*(s)$, and
let $h_{\text{opt}} :=  s_{\text{opt}} ^{2/ d+4} n^{-1/d+4}$.  By \eqref{eq:5}, $h_{\text{opt}}$ minimizes \eqref{eq:1}, and so minimizes $\HDR(h)$.  By Theorem~\ref{thm:hdr}, we conclude that for any $h_{0}$ that minimizes
$\bb{E}[\mu_{f_0}\{\mc{L}_{\tau}\Delta\hat{\mc{L}}_{\tau,\bH}\}]$,  $h_0 = h_{\text{opt}}(1 + o(1))$.

\section{Proof of intermediate results}
\label{sec:proofs-intermediate}

\begin{proof}[Proof of Lemma~\ref{lem:hdr-step1}]
  \begin{mylongform}
    \begin{longform}
      From the proof, we shall see $C_1$ only depends on $f_0$.
    \end{longform}
  \end{mylongform}
  Let $C_1>1+2\lambda\lp \lb f_0(\bx)\ge
  f_{\tau,0}\rb\rp/\int_{\beta_\tau}\frac{f_0}{\|\nabla f_0\|}\,d\mc{H}$.
  Then when $\varepsilon>0$ is sufficiently small,
  \begin{align*}
    \int \tilde{f}(\bx)\one_{\{\tilde{f}(\bx)\ge
    f_{\tau,0}-C_1\varepsilon\}}\,d\bx&\ge \int (f_0(\bx)-\varepsilon)\one_{\{f(\bx)\ge
                                        f_{\tau,0}-(C_1-1)\varepsilon\}}\,d\bx\\
                                      &=1-\tau+\int f_0(\bx)\one_{\{f_{\tau,0}-(C_1-1)\varepsilon\le f_0(\bx)<
                                        f_{\tau,0}\}}\,d\bx\\
                                      &\qquad-\varepsilon\lambda\lp\lb  f_0(\bx)\ge
                                        f_{\tau,0}-(C_1-1)\varepsilon\rb\rp\\
                                      &\ge 1-\tau+\int f_0(\bx)\one_{\{f_{\tau,0}-(C_1-1)\varepsilon\le f_0(\bx)<
                                        f_{\tau,0}\}}\,d\bx\\
                                      &\qquad -2\varepsilon\lambda(\{f_0(\bx)\ge
                                        f_{\tau,0}\}).
  \end{align*}
  \begin{mylongform}
    \begin{longform}
      On the second line of the above inequalities, $\lambda(\{f_0(\bx)\ge
      f_{\tau,0}-(C_1-1)\varepsilon\})\downarrow \lambda(\{\bx:f_0(\bx)\ge
      f_{\tau,0}\})$ when
      $\varepsilon\downarrow 0$ by continuity of measure. So $\lambda(\{\bx:f_0(\bx)\ge
      f_{\tau,0}-(C_1-1)\varepsilon\})\le 2\lambda(\{\bx:f_0(\bx)\ge
      f_{\tau,0}\})$ when $\varepsilon>0$ is sufficiently small.
    \end{longform}
  \end{mylongform}
  By Proposition A.1 of \cite{Cadre:2006db},
  \begin{align}
    \label{eq:step1-leveltohauss}
    \int f_0(\bx)\one_{\{f_{\tau,0}-(C_1-1)\varepsilon\le f_0(\bx)<
    f_{\tau,0}\}}\,d\bx
    =\int_{f_{\tau,0}-(C_1-1)\varepsilon}^{f_{\tau,0}}\int_{\beta(s)}\frac{f_0(\bx)}{\|\nabla
    f_0(\bx)\|}\,d\mathcal{H}(\bx)\,ds.
  \end{align}
  So we can express $\int f_0(\bx)\one_{\{f_{\tau,0}-(C_1-1)\varepsilon\le f_0(\bx)<
    f_{\tau,0}\}}\,d\bx$ as
  \begin{align}
    \label{eq:step1-intdiff}
    (C_1-1)\varepsilon
    \int_{\beta_{\tau}}\frac{f_0(\bx)}{\|\nabla f_0(\bx)\|}\,d\mathcal{H}(\bx)+O(\varepsilon^2),
  \end{align}
  by Lemma~\ref{lem:Hauss-diff}, and thus see that
  \begin{align*}
    \MoveEqLeft  \int \tilde{f}(\bx)\one_{\{\tilde{f}(\bx)\ge
    f_{\tau,0}-C_1\varepsilon\}}\,d\bx\\&\ge 1-\tau +(C_1-1)\varepsilon
                                          \int_{\beta_{\tau}}\frac{f_0(\bx)}{\|\nabla
                                          f_0(\bx)\|}\,d\mathcal{H}(\bx)+o(\varepsilon) -2\varepsilon\lambda(\{\bx:f_0(\bx)\ge
                                          f_{\tau,0}\})\\&>1-\tau,
  \end{align*}
  when $\varepsilon>0$ is sufficiently small. So
  $\tilde{f}_{\tau}>f_{\tau,0}-C_1\varepsilon$.    For the upper bound, with
  a similar argument,
  \begin{mylongform}
    \begin{longform}
      \begin{align*}
        \int \tilde{f_0}(x)\one_{\{\tilde{f}\ge
        f_{\tau,0}+C_1\varepsilon\}}\,dx\le&\int (f_0(x)+\varepsilon)\one_{\{f_0\ge
                                             f_{\tau,0}+(C_1-1)\varepsilon\}}\,dx\\
        =&1-\tau -\int f_0(x)\one_{\{f_{\tau,0}\le f_0\le f_{\tau,0}+(C_1-1)\varepsilon\}}\,dx +\varepsilon
           \lambda(\{f_0\ge f_{\tau,0}+(C_1-1)\varepsilon\})\\
        \le&1-\tau-\int_{f_{\tau,0}}^{f_{\tau,0}+(C_1-1)\varepsilon}\int_{\beta(s)}\frac{f_0(\bx)}{\|\nabla
             f_0(\bx)\|}\,d\mathcal{H}(\bx)\,ds+\varepsilon\lambda(\{f_0\ge f_\tau,0\})\\
        =&1-\tau-(C_1-1)\varepsilon\int_{\beta_\tau}\frac{f_0(\bx)}{\|\nabla
           f_0(\bx)|}\,d\mathcal{H}(\bx)+o((C_1-1)\varepsilon)+\varepsilon\lambda(\{f_0\ge
           f_\tau,0\})\\
        <&1-\tau,
      \end{align*}
      when $\varepsilon>0$ is sufficiently small,
    \end{longform}
  \end{mylongform}
  we get $\tilde{f}_{\tau}<f_{\tau,0}+C_1\varepsilon$. So we proved $|\tilde{f}_{\tau}-f_{\tau,0}|\le C_1\varepsilon$ for
  $\varepsilon>0$ sufficiently small.
\end{proof}

\begin{proof}[Proof of Lemma~\ref{lem:hdr-step3}]
  We first prove an intermediate result that
  \begin{align}
    \label{eq:step2re}
    \int_{\mc{L}_{\delta}(f_{\tau,0})^c}f_0(\bx)P\lp\ffnH(\bx)\ge
    \fftaun\rp\,d\bx+\int_{\mc{L}_{-\delta}(f_{\tau,0})}f_0(\bx)P\lp\ffnH(\bx)<
    \fftaun\rp\,d\bx
  \end{align}
  is $o(n^{-1})$ as $n\to \infty$ for fixed $\delta>0$ sufficiently
  small. Observe that under Assumption \ref{assm:DA-hdr} if
  $\delta>0$  is sufficiently small, then there exists $\varepsilon>0$ such that $f_0(\bx)\le
  f_{\tau,0}-\varepsilon$ for $\bx\in \mc{L}_{\delta}(f_{\tau,0})^c$ and $f_0(\bx)\ge
  f_{\tau,0}+\varepsilon$ for $\bx\in \mc{L}_{-\delta}(f_{\tau,0})$.
  By reducing $\delta>0$  if necessary, for $\bx\in
  \mc{L}_{\delta}(f_{\tau,0})^c$,
  \begin{align*}
    P\lp\ffnH(\bx)\ge \fftaun\rp&\le
                                  P\lp\ffnH(\bx)-f_0(\bx)-(\fftaun-f_{\tau,0})\ge  \varepsilon\rp\\
                                &  \le P\lp\|\ffnH-f_0 \|_{\infty}\ge
                                  \varepsilon/2\rp+P\lp|\fftaun-f_{\tau,0}|\ge
                                  \varepsilon/2\rp\\
                                &  \le P\left(\|\ffnH-f_0 \|_{\infty}\ge
                                  \frac{\varepsilon}{2C_1}\right)+P\left(|\fftaun-f_{\tau,0}|\ge
                                  \frac{\varepsilon}{2}\right),
  \end{align*}
  where $C_1\ge 1$ is the constant we defined in
  Lemma~\ref{lem:hdr-step1} ; by that lemma, we have
  \begin{align*}
    P\left(|\fftaun-f_{\tau,0}|\ge
    \frac{\varepsilon}{2}\right)\le  P\left(\|\ffnH-f_0 \|_{\infty}\ge
    \frac{\varepsilon}{2C_1}\right),
  \end{align*}
  so
  \begin{align}
    \label{eq:pup}
    P\lp\ffnH(\bx)\ge \fftaun\rp\le  2P\left(\|\ffnH-f_0 \|_{\infty}\ge
    \frac{\varepsilon}{2C_1}\right).
  \end{align}
  A similar argument yields the same upper bound for $P(\ffnH(\bx)<
  \fftaun)$ when $\bx\in\mc{L}_{-\delta}(f_{\tau,0})$.
  \begin{mylongform}
    \begin{longform}
      \begin{align*}
        P(\ffnH(\bx)<\fftaun) &\le
                                P(\fftaun-\ffnH(\bx)+f_0(\bx)-f_{\tau,0}\ge
                                \varepsilon) \\
                              &  \le P(\|f_0(\bx)-\ffnH(\bx)\|_{\infty}\ge \varepsilon/2)+P(|\fftaun-f_{\tau,0}|\ge
                                \varepsilon/2)\\
                              &  \le2P(\|\ffnH-f_0\|_{\infty}\ge \frac{\varepsilon}{2C})
      \end{align*}

      Now, since $f_0$ is uniformly continuous,
      \begin{align*}
        \|\bb{E}(\ffnH)-f_0\|_{\infty}&=\sup_{\bx\in\RR^d}\left|\int
                                        K(\boldsymbol{z})\{f_0(\bx-\bH^{1/2}\boldsymbol{z})-f_0(\bx)\}\,d\boldsymbol{z}\right|\nonumber\\
                                      & \le\sup_{\bx\in\RR^d}\int
                                        K(\boldsymbol{z})|f_0(\bx-\bH^{1/2}\boldsymbol{z})-f_0(\bx)|\,d\boldsymbol{z}\nonumber\\
                                      & \le\int
                                        K(\boldsymbol{z})\sup_{\bx\in\RR^d}\lb|f_0(\bx-\bH^{1/2}\boldsymbol{z})-f_0(\bx)|\rb\,d\boldsymbol{z}\nonumber\\
                                      & \le\int
                                        K(\boldsymbol{z})w\|\bH^{1/2}\bs{z}\|\,d\boldsymbol{z},
      \end{align*}
      here $w$ is the modulus of uniform continuity of $f_0$ and thus the
      above quantity converges to 0
      as $n\rightarrow\infty$ by Lebesgue Dominated Convergence Theorem.
    \end{longform}
  \end{mylongform}
  Now by Assumption~\ref{assm:DA-hdr},
  \begin{align*}
    \|\bb{E}(\ffnH)-f_0\|_{\infty}\rightarrow 0,
  \end{align*}
  as $n\rightarrow \infty$.
  Together with the inequality
  \eqref{eq:pup} together,  this yields that for $n$ sufficiently large,
  \begin{align*}
    \MoveEqLeft    \int_{\mc{L}_{\delta}(f_{\tau,0})^c}f_0(\bx)P\lp\ffnH(\bx)\ge
    \fftaun\rp\,d\bx+\int_{\mc{L}_{-\delta}(f_{\tau,0})}f_0(\bx)P\lp\ffnH(\bx)<
    \fftaun\rp\,d\bx\\
 & \le 2P\left(\|\ffnH-f_0\|_{\infty}\ge
   \frac{\varepsilon}{2C_1}\right)
  \end{align*}
  which is bounded above by
  \begin{align*}
    \MoveEqLeft
    2P\left(\|\ffnH-\bb{E}(\ffnH)\|_{\infty}\ge
    \frac{\varepsilon}{4C_1}\right)
    +2P\left(\|\bb{E}(\ffnH)-f_0\|_{\infty}\ge
    \frac{\varepsilon}{4C_1}\right)\\
 &  =2P\left(\|\ffnH-\bb{E}(\ffnH)\|_{\infty}\ge
   \frac{\varepsilon}{4C_1}\right)
              \le
              L\exp\left\{-\frac{C_{0,1}\varepsilon^2n|\bH|^{1/2}}{16C_1^2}\right\}=o(n^{-1}),
  \end{align*}
  where the last inequality comes from  Corollary~\ref{cor:convergerate}.

  Now  it suffices to show that $ E(\delta,\delta_n)=o(n^{-1})$, where   $
  E(\delta,\delta_n)$ is defined as
  \begin{mylongform}
    \begin{longform}
      \begin{align*}
        &\int_{\mc{L}_{\delta_n}(f_{\tau,0})^c}f_0(\bx)P(\ffnH(\bx)\ge
          \fftaun)\,d\bx+\int_{\mc{L}_{-\delta_n}(f_{\tau,0})}f_0(\bx)P(\ffnH(\bx)<
          \fftaun)\,d\bx\\
        &-\left(\int_{\mc{L}_{\delta}(f_{\tau,0})^c}f_0(\bx)P(\ffnH(\bx)\ge
          \fftaun)\,d\bx+\int_{\mc{L}_{-\delta}(f_{\tau,0})}f_0(\bx)P(\ffnH(\bx)<
          \fftaun)\,d\bx\right)\\
      \end{align*}
    \end{longform}
  \end{mylongform}
  \begin{align*}
    \MoveEqLeft\int_{\mc{L}_{\delta_n}(f_{\tau,0})^c\backslash \mc{L}_{\delta}(f_{\tau,0})^c}f_0(\bx)P(\ffnH(\bx)\ge
    \fftaun)\,d\bx\\
 &  +\int_{\mc{L}_{-\delta_n}(f_{\tau,0})\backslash \mc{L}_{-\delta}(f_{\tau,0})}f_0(\bx)P(\ffnH(\bx)<
   \fftaun)\,d\bx.
  \end{align*}
  Using Taylor expansion, we have
  \begin{mylongform}
    \begin{longform}(See Theorem~\ref{thm:taylor} in Appendix or \cite{folland2002advanced}, Page 91,
      \textbf{Theorem 2.68}) yields
      \begin{align*}
        \|\bb{E}(\ffnH)-f_0\|_{\infty}=&\sup_{\bx\in\RR^d}\left|\int
                                         K(\boldsymbol{z})\{f_0(\bx-\bH^{1/2}\boldsymbol{z})-f_0(\bx)\}\,d\boldsymbol{z}\right|\nonumber\\
        =&\sup_{\bx\in\RR^d}\left|\int
           K(\boldsymbol{z})\lb-(\bH^{1/2}\boldsymbol{z})'\nabla
           f_0(\bx)+\inv{2}(\bH^{1/2}\boldsymbol{z})'\nabla^2f_0(\bx_z)\bH^{1/2}\boldsymbol{z}\rb\,d\boldsymbol{z}\right|\nonumber\\
        =&\sup_{\bx\in\RR^d}\left|\int
           K(\boldsymbol{z})\lb\inv{2}(\bH^{1/2}\boldsymbol{z})'\nabla^2f_0(\bx_z)\bH^{1/2}\boldsymbol{z}\rb\,d\boldsymbol{z}\right|,
      \end{align*}
    \end{longform}
  \end{mylongform}
  \begin{align*}
    \|\bb{E}(\ffnH)-f_0\|_{\infty}=\sup_{\bx\in\RR^d}\left|\int
    K(\boldsymbol{z})\lb\inv{2}(\bH^{1/2}\boldsymbol{z})'\nabla^2f_0(\bx_z)\bH^{1/2}\boldsymbol{z}\rb\,d\boldsymbol{z}\right|,
  \end{align*}
  where $\bx_z=\bx-c\bH^{1/2}\boldsymbol{z}$ for some $c\in (0,1)$.
  Under Assumption \ref{assm:DA-hdr}, $f_0$ has bounded second derivatives
  and let $A>0$ be such that $\|\nabla^2f\|_{\infty}\le A$. Then
  \begin{mylongform}
    \begin{longform}
      \begin{align*}
        \|\bb{E}(\ffnH)-f_0\|_{\infty}=&  \sup_{\bx\in\RR^d}\left|\int
                                         K(\boldsymbol{z})\lb\inv{2}(\bH^{1/2}\boldsymbol{z})'\nabla^2f_0(\bx_z)\bH^{1/2}\boldsymbol{z}\rb\,d\boldsymbol{z}\right|\nonumber\\
        \le&\sup_{\bx\in\RR^d}\inv{2}\int
             K(\bz)\left|\bz'\bH^{1/2}\nabla^2f_0(\bx_{\bz})\bH^{1/2}\bz\right|\,d\bz\nonumber\\
        \le&\inv{2}dA\int K(\bz)\bz'\bH\bz\,d\bz\nonumber\\
        =&\inv{2}dA\tr\lp\bH\int K(\bz)\bz\bz'\,d\bz\rp\nonumber\\
        =&\inv{2}dA\mu_2(K)\tr(\bH)\nonumber\\
        =&O\lb\lambda_{\max}(\bH)\rb,
      \end{align*}
    \end{longform}
  \end{mylongform}
  \begin{align}
    \label{eq:ffnH-bias-supnorm}
    \begin{split}
      \|\bb{E}(\ffnH)-f_0\|_{\infty}
      \le\inv{2}dA\int K(\bz)\bz'\bH\bz\,d\bz
      =\inv{2}dA\mu_2(K)\tr(\bH)
      =O\lb\lambda_{\max}(\bH)\rb,
    \end{split}
  \end{align}
  as $|\bH|\rightarrow 0$.
  Now by Lemma~\ref{lem:hdr-step3-c2}, there exists a constant $c_2$ small enough that if we take
  $\varepsilon_n=c_2\delta_n$, then we have $|f_0(\bx)-f_{\tau_,0}|\ge
  \varepsilon_n$ when
  $\bx\in(\mc{L}_{\delta_n}(f_{\tau,0})^c\backslash
  \mc{L}_{\delta}(f_{\tau,0})^c)\cup(\mc{L}_{-\delta_n}(f_{\tau,0})\backslash
  \mc{L}_{-\delta}(f_{\tau,0}))$. Moreover,
  $\lambda_{\max}(\bH)=o(\epsilon_n)$ by our assumption, so for $n$ sufficiently
  large, by (\ref{eq:ffnH-bias-supnorm}),
  $P(\|\bb{E}(\ffnH)-f_0\|_{\infty}\ge
  \frac{\varepsilon_n}{4C})=0$. Then for $n$ large enough,
  \begin{align}
    \begin{split}
      E(\delta,\delta_n) \le 2P\left(\|\ffnH-\bb{E}(\ffnH)\|_{\infty}\ge
        \frac{\varepsilon_n}{4C}\right)
      \le
      L\exp\left\{-\frac{C_{0,1}\varepsilon_n^2n|\bH|^{1/2}}{16C_1^2}\right\}=o(n^{-1}).
    \end{split}
  \end{align}
  as $n\rightarrow \infty$.
\end{proof}
\begin{mylongform}
  \begin{longform}
    \begin{lemma}
      Let $\tilde{f}=f_0+g$ be another uniformly continuous density
      function satisfying Assumption \ref{assm:DA} and such that $\|g\|_{\infty}\le \varepsilon$ for some $\varepsilon>0$
      sufficiently small. Then if Assumption~\ref{assm:DA}
      holds for $f_0$, we
      will have for every
      $\bs{z}\in\tilde{\beta}_{\tau}:=\{\bx:\tilde{f}(\bx)=\tilde{f}_{\tau}\}$
      where $\tilde{f}_{\tau}:=\inf \{y\in(0,\infty):\int_{\bb{R}^d}\tilde{f}(\bx)\one_{\{\tilde{f}(\bx)\ge
        y\}}\,d\bx\le 1-\tau\}$, it is
      within  $d_{\bs{z}}=g(\bs{z})/\|\nabla
      f_0(\bs{z})\|+O(\|g\|^2_{\infty})$ of
      $\beta_{\tilde{\tau}}=\{\bx:f_0(\bx)=\tilde{f}_{\tau}\}$, i.e.,
      there exists some $w\in \beta_{\tilde{\tau}}$ such that $\|z-w\|=d_z$.
    \end{lemma}
    \begin{proof}
      Let $\bs{z}$ be such that $\tilde{f}(\bs{z})=\tilde{f}_{\tau}$. If $0<
      \tilde{f}(\bs{z})-f_0(\bs{z})\le \varepsilon$, let
      $\bs{u}=-\frac{\nabla f_0(\bs{z})}{\|\nabla f_0(\bs{z})\|}$ be the inner
      outer  unit normal vector to $\tilde{\beta}_{\tau}$ at $\bs{z}$. And
      we let $\bs{y}=\bs{z}-\tau \bs{u}, \tau>0$. Then by Taylor expansion,
      \begin{align*}
        f_0(\bs{y})&=f_0(\bs{z})-\tau\nabla f_0(\bs{z})'\bs{u}+\inv{2}(\tau
                     \bs{u})'\nabla^2 f_0(\bs{z}-s\tau \bs{u})\tau \bs{u}\\
                   &=f_0(\bs{z})+\tau\|\nabla f_0(\bs{z})\|+\frac{\tau^2}{2}\bs{u}'\nabla^2f_0(\bs{z}-s\tau \bs{u}) \bs{u},
      \end{align*}
      where $s\in (0,1)$. Since $f_0$ has bounded second derivatives.
      Let $A>0$ be such that $\|\nabla^2 f_0\|_{\infty}\le
      A$. Then by Lemma~\ref{lem:algebra}, the second
      order term
      \begin{align*}
        \left|\frac{\tau^2}{2}\bs{u}'\nabla^2f_0(\bs{z}-s\tau \bs{u})\tau \bs{u}\right|\le \frac{\tau^2}{2}d\|\bs{u}\|^2A= \frac{\tau^2}{2}dA.
      \end{align*}
      This ensures that we can take $\tau$ small enough such that
      $\tau\|\nabla f_0(\bs{z})\|+\frac{\tau^2}{2}u'\nabla^2f_0(\bs{z}+s\tau u)\tau
      u\ge \frac{3}{4}\tau\|\nabla f_0(\bs{z})\|$.

      This can be achieved by letting $\tau$ be such that
      $-\frac{\tau^2}{2}dA\ge -\frac{\tau}{4}\|\nabla f_0(\bs{z})\|$, that
      is, $\tau\le \|\nabla f_0(\bs{z})\|/(2dA)$. This choice of upper
      bound for $\tau$ depends on $\bs{z}$. However, we can also choose an
      upper bound not depending on $\bs{z}$. Note in our assumption,
      $\|\nabla f(\bx)\|$ is bounded away from 0 for
      $\bx\in\{\bx:f_{\tau,0}-a\le f_0(\bx)\le f_{\tau,0}+a\}$ and by the
      result of step 1, if $\|g\|_{\infty}$ is sufficiently small, we can
      have $\bs{z}\in \{\bx:f_{\tau,0}-a\le f_0(\bx)\le
      f_{\tau,0}+a\}$. Then we can make $\tau$ small enough such that
      $-\frac{\tau^2}{2}dA\ge -\frac{\tau}{4}\inf_{\{\bx:f_{\tau,0}-a\le
        f_0(\bx)\le f_{\tau,0}+a\}}\|\nabla f_0(\bs{x})\|$ and apparently
      in this case the upper bound for $\tau$ does not depend on $\bs{z}$.

      So we have
      \begin{align*}
        f_0(\bs{y})\ge f_0(\bs{z})+\frac{3}{4}\tau\|\nabla f_0(\bs{z})\|,
      \end{align*}
      and we also know
      \begin{align*}
        f_0(\bs{z})-\tilde{f}(\bs{z})\ge -\varepsilon.
      \end{align*}
      Since we can regard $f_0(\bs{y})=f_0(\bs{z}+\tau \bs{u})$ as a continuous
      function in $\tau$, we conclude for some $0\le \tau\le
      \frac{4}{3}\frac{\varepsilon}{\|\nabla f_0(\bs{z})\|}$, we can have
      $f_0(\bs{y})=\tilde{f}_{\tau}$. Then
      \begin{align*}
        \tilde{f}(\bs{z})=\tilde{f}_{\tau}=f_0(\bs{y})=f_0(\bs{z})+\tau\|\nabla f_0(\bs{z})\|+\frac{\tau^2}{2}\bs{u}'f_0(\bs{z}+s\tau \bs{u}) \bs{u},
      \end{align*}
      which gives us
      \begin{align*}
        \tau&=\frac{\tilde{f}(\bs{z})-f_0(\bs{z})}{\|\nabla
              f_0(\bs{z})\|}-\frac{\tau^2}{2\|\nabla
              f_0(\bs{z})\|}\bs{u}'f_0(\bs{z}+s\tau \bs{u}) \bs{u}\\
            &=\frac{g(\bs{z})}{\|\nabla f_0(\bs{z})\|}-\frac{\tau^2}{2\|\nabla
              f_0(\bs{z})\|}\bs{u}'f_0(\bs{z}+s\tau \bs{u})\bs{u}.
      \end{align*}
      To bound the second term of the last line, note
      \begin{align*}
        \left|\frac{\tau^2}{2\|\nabla
        f_0(\bs{z})\|}\bs{u}'f_0(\bs{z}+s\tau \bs{u}) \bs{u}\right|\le
        \frac{8\varepsilon^2}{9\|\nabla f_0(\bs{z})\|^3}dA,
      \end{align*}
      and by Lemma~\ref{lem:hdr-step1}, we know when $\varepsilon>0 $ is
      sufficiently small,  $\|g\|_{\infty}\le \varepsilon$ indicates $
      |\tilde{f}_{\tau}-f_{\tau}|\le C\varepsilon$, for some $C>0$. Thus, we let
      $\varepsilon $ be sufficiently small such that $C\varepsilon \le a$,
      where $a>0$ is the constant defined in Assumption~D4 and this
      indicates $\bs{z}\in \{\bx:f_{\tau}-a\le f_0(\bx)\le f_{\tau}+a\}$. So
      by assumption $\|\nabla f_0(\bs{z})\|^3$ is bounded away from 0. Hence,
      \begin{align*}
        \tau=\frac{g(\bs{z})}{\|\nabla f_0(\bs{z})\|}+O(\|g\|_{\infty}^2).
      \end{align*}
      A similar argument can be used to prove the case
      $0<f_0(\bs{z})-\tilde{f}(\bs{z})\le \varepsilon$.
    \end{proof}
  \end{longform}
\end{mylongform}

\begin{proof}[Proof of Lemma~\ref{lem:hdr-step4-order}]
  Let $\bs z \in \{ \bx\in\RR^d:f_0(\bx) = \tilde{f}_\tau \}$ and let
  $\bs y \in \{\bx\in\RR^d: \tilde{f}(\bx) = \tilde{f}_\tau \} $ be
  such that $\bs z = \bx + \eta_1 u_{\bx}$ and $\bs y = \bx + \eta_2
  u_{\bx}$ for some $\bx \in \beta_{\tau}$ and $\eta_i \equiv
  \eta_i(\bx) \in \RR$, $i=1,2$.  By Taylor expansion, we have
  \begin{align*}
    f_0(\bs{z})&=f_0(\bx+\eta_1u_{\bx})\\
               &=f_0(\bx)+\eta_1u_{\bx}'\nabla
                 f_0(\bx)+\inv{2}u'_{\bx}\nabla^2f_0(\bx+s_1\eta_1u_{\bx})u_{\bx}\eta_1^2,
  \end{align*}
  where $s_1\in [0,1]$, and
  \begin{align*}
    \tilde{f}(\bs{y})&=\tilde{f}(x+\eta_2u_{\bx})\\
                     &=\tilde{f}(\bx)+\eta_2u'_{\bx}\nabla\tilde{f}(\bx+s_2\eta_2u_{\bx}).
  \end{align*}
  We then see
  \begin{equation}
    \label{eq:23}
    \begin{split}
      0 &= f_0(\bs z) - \tilde{f}( \bs y) \\
      & = f_0(\bx)+\eta_1u_{\bx}'\nabla
      f_0(\bx)+\inv{2}u'_{\bx}\nabla^2f_0(\bx+s_1\eta_1u_{\bx})u_{\bx}\eta_1^2\\
      &\qquad-\tilde{f}(\bx)-\eta_2u'_{\bx}\nabla\tilde{f}(\bx+s_2\eta_2u_{\bx}),
    \end{split}
  \end{equation}
  where $s_2\in [0,1]$.   We thus have
  \begin{align}
    \label{eq:24}
    \begin{split}
      \eta_1 - \eta_2& = \frac{ f_0(\bx) - \tilde{f}(\bx)}{ \| \grad f_0(\bx) \|}
      +
      \frac{u'_{\bx}\nabla^2f_0(\bx+s_1\eta_1u_{\bx})u_{\bx}\eta_1^2}{2\|\nabla
        f_0(\bx)\|}\\
      &\qquad+\eta_2 \frac{ \la \grad \tilde{f}(\bx+s_2\eta_2u_{\bx}) - \grad
        f(\bx) , \grad f_0(\bx) \ra }{ \| \grad f_0(\bx) \|^2 }.
    \end{split}
  \end{align}
  A similar analysis as in \eqref{eq:23}, beginning from the identity
  $\tilde{f}_\tau - \fftau = \tilde f(\bs y) - f_0 (\bx) $ shows that
  $\eta_2 = O( \| g \|_\infty)$ since by Lemma~\ref{lem:hdr-step1},
  $\tilde{f}_\tau - \fftau = O ( \| g \|_\infty)$.  (Similarly,
  $\eta_1 = O( \| g \|_\infty)$.)
  Since by Assumption~\ref{assm:DA-hdr}, $f_0$ has bounded second
  derivatives,the second term on the right in \eqref{eq:24} is $O( \|
  g \|_{\infty} ^2)$. For the second term, note
  \begin{align*}
    \MoveEqLeft    \frac{ \la \grad \tilde{f}(\bx+s_2\eta_2u_{\bx}) - \grad
    f(\bx) , \grad f_0(\bx) \ra }{ \| \grad f_0(\bx) \|^2 }\\
 &= \frac{ \la \grad \tilde{f}(\bx+s_2\eta_2u_{\bx}) - \grad
   f_0(\bx+s_2\eta_2u_{\bx}) , \grad f_0(\bx) \ra }{ \| \grad
   f_0(\bx) \|^2 }\\
 &\qquad+ \frac{ \la \grad f_0(\bx+s_2\eta_2u_{\bx}) - \grad f_0(\bx) , \grad f_0(\bx) \ra }{ \| \grad f_0(\bx) \|^2 },
  \end{align*}
  and by Assumption~\ref{assm:DA-hdr}, $\nabla f_0(\bx)$ is Lipschitz, we
  have the third term is $O(\|g\|_{\infty}\|\nabla g\|_{\infty}+\|g\|^2_{\infty})$.

  We will apply Lemma~\ref{lem:COV-approx} to $h(\bs y) = \one_{\{\tilde{f}(\bs y) \ge\tilde{f}_{\tau}\}} - \one_{\{f_0(\bs y) \ge\tilde{f}_{\tau}\}}$.  For $\|g \|_\infty$ small enough, $ \{\tilde{f}\ge\tilde{f}_{\tau}\}\Delta \{f_0 \ge\tilde{f}_{\tau}\} \subset \beta_{\tau}^\delta$ for some $\delta > 0$, by
  Lemma~\ref{lem:hdr-step1},
  and by Assumption~\ref{assm:DA-hdr}~\ref{item:DA-hdr-grad} and \ref{item:DA-hdr-localization}.  Thus the left side of \eqref{eq:22} equals $\int_{\beta_\tau^\delta} h(\bs y) d\bs y$.
  We may shrink $\delta$ so that the conclusion of Theorem~\ref{thm:eps-nbhd} holds, so that for each $\bs y \in \beta_\tau^\delta$ there is a unique closest $\bx_{\bs y} \in \beta_\tau$.
  Now, for $\delta$ small enough, considering $\one_{ \lb \tilde{f}(\bx + t u_{\bx} ) \ge \tilde{f}_\tau \rb}$ as a function of $ t \in [-\delta, \delta]$, we can see that $\one_{ \lb \tilde{f}(\bx + t u_{\bx} ) \ge \tilde{f}_\tau \rb} = \one_{\lb -\delta \le t \le \eta_2(\bx) \rb }$, because $\grad \tilde{f}(\bx) ' \grad {f}_0(\bx) > 0$, so $\tilde{f}$ is locally strictly decreasing in the direction of $u_{\bx} = - \grad f_0(\bx) / \| \grad f_0(\bx) \|$.  Similarly $\one_{ \lb f_0(\bx + tu_{\bx}) \ge \tilde{f}_\tau \rb} = \one_{\lb -\delta \le t \le \eta_1(\bx) \rb } $. Thus for $\bs y \in \beta_\tau^\delta$,
  \begin{equation}
    \label{eq:25}
    h(\bs{y})  = \one_{\lb \eta_1(\bs{x}_{\bs y}) \le t \le \eta_2(\bs{x}_{\bs y}) \rb}
    - \one_{ \lb \eta_2(\bs{x}_{\bs y}) \le t \le \eta_1(\bs{x}_{\bs y}) \rb},
  \end{equation}
  (where $\one_{\lb a \le t \le b \rb}$ is just identically $0$ if $b < a$) and so for $\bx \in \beta_\tau$,
  \begin{equation}
    \label{eq:26}
    H(\bx)
    := \int_{-\delta}^\delta h(\bx + t u_{\bx}) \, dt
    = \eta_2(\bx) - \eta_1(\bx).
  \end{equation}
  We can now apply Lemma~\ref{lem:COV-approx} to see
  \begin{align*}
    \int_{\beta_\tau^\delta} h(\bx) d\bx
    & = \int_{\beta_\tau} H( \bx) \, d\cH(\bx)
      +  O( \sup_{\bx \in \beta_\tau} \eta_2(\bx)-\eta_1(\bx) )^2 \quad \text{ as }
      \sup_{\bx \in \beta_\tau} \eta_2(\bx) - \eta_1(\bx) \to 0, \\
    & =\int_{\beta_\tau} \frac{ \tilde{f}(\bx) - f_0(\bx) }{ \| \grad f_0(\bx) \| } \, d \cH(\bx)
      + O( \| g \|_\infty^2 ) + O( \| g \|_\infty \| \grad g \|_\infty )
  \end{align*}
  as
  $ \| g \|_\infty^2 +  \| g \|_\infty \| \grad g \|_\infty  \to 0$,
  by  \eqref{eq:26} and \eqref{eq:24} (and because $\sup_{\bx \in \beta_\tau} \eta_2(\bx) - \eta_1(\bx) = O( \| g \|_\infty )$ and the term on the right of \eqref{eq:24} is $O( \| g \|_\infty \| \grad g \|_\infty )$).
  \begin{mylongform}
    \begin{longform}
      Note that for $i=1,2$, $\eta_i = O( \| g \|_\infty)$.  This is actually needed to bound the order of the error from the lemma.
    \end{longform}
  \end{mylongform}
\end{proof}

\begin{proof}[Proof of Lemma~\ref{lem:hdr-step4-order-g}]
  Let $\bs{y}\in\{
  \bx\in\RR^d:\tilde{f}(\bx)=\tilde{f}_{\tau}\}$ be such
  that $\bs{y}=\bx+\eta u_{\bx}$ for some $\bx\in \beta_{\tau}$. Then
  \begin{align*}
    \tilde{f}(\bs{y})=\tilde{f}(\bx)+\eta\nabla \tilde{f}(\bx+s\eta u_{\bx})'u_{\bx},
  \end{align*}
  where $s\in[0,1]$   depends on $\bx$. Then subtracting
  $f_0(\bx)$ on both sides yields
  \begin{align*}
    \tilde{f}_{\tau}-f_{\tau,0}=\tilde{f}(\bx)-f_0(\bx)+\eta\nabla \tilde{f}(\bx+s\eta u_{\bx})'u_{\bx},
  \end{align*}
  so
  \begin{align*}
    \eta = \frac{\tilde{f}_{\tau}-f_{\tau,0}-g(\bx)}{\nabla f(\bx+s\eta u_{\bx})'u_{\bx}},
  \end{align*}
  and by Lemma~\ref{lem:hdr-step1},
  $\tilde{f}_{\tau}-f_{\tau,0}-g(\bx)=O(\|g\|_{\infty})$. We also know $\nabla
  f(\bx+s\eta u_{\bx})'u_{\bx}$ is bounded away from zero as
  $\|g\|^2_{\infty}+\|g\|_{\infty}\|\nabla g\|_{\infty}\rightarrow 0$. Then
  \begin{align*}
    \int_{\RR^d}g(\bx)\left(\one_{\{\tilde{f}(\bx)\ge
    \tilde{f}_{\tau}\}}-\one_{\{f(\bx)\ge f_{\tau}\}}\right)\,d\bx&\le \|g\|_{\infty}\int_{\RR^d}\left|\one_{\{\tilde{f}(\bx)\ge
                                                                    \tilde{f}_{\tau}\}}-\one_{\{f(\bx)\ge f_{\tau}\}}\right|\,d\bx\\
                                                                  &=O(\|g\|_{\infty}^2).
  \end{align*}
\end{proof}

\begin{proof}[Proof of Lemma~\ref{lem:ftaun-bias-exp}]
  It is well known (e.g., \cite{Wand:1995kv}) that
  \begin{align*}
    \bb{E}\ffnH(\bx)=f_0(\bx)+\frac{1}{2}\mu_2(K)\tr\{\bs{H}\nabla^2
    f_0(\bx)\}+o\{\tr(\bs{H})\}.
  \end{align*}
  This statement and all asymptotic statements in this proof are  as $n \to \infty$ (implying $\bH \to 0$).
  Now we show
  \begin{align*}
    \MoveEqLeft\int_{\beta_{\tau}}\frac{\bb{E}\ffnH(\bx)-f_0(\bx)}{\|\nabla
    f_0(\bx)\|}\,d\mc{H}(\bx)+\inv{f_{\tau,0}}\int_{\mc{L}_{\tau}}\bb{E}\ffnH(\bx)-f_0(\bx)\,d\bx\\
 &=V_1(\bH)+V_2(\bH)+o\{\tr(\bH)\}.
  \end{align*}
  \begin{mylongform}
    \begin{longform}
      For fixed $\bx\in\beta_{\tau}$,
      by change of variable and a Taylor expansion, we  have
      \begin{align}
        \MoveEqLeft \bb{E}\ffnH(\bx)-f_0(\bx)-\frac{1}{2}\mu_2(K)\tr\{\bs{H}\nabla^2
        f_0(\bx)\} \label{eq:lemmaA6-Efhat} \\
     &=\int_{\RR^d}K(\bs{z})f_0(\bx-\bH^{1/2}\bs{z})\,d\bs{z}-f_0(\bx)-\frac{1}{2}\mu_2(K)\tr\{\bs{H}\nabla^2
       f_0(\bx)\}, \nonumber\\
     &=\int_{\RR^d}K(\bs{z})\lb f_0(\bx)
       +\frac{1}{2}(\bH^{1/2}\bs{z})^T\nabla^2
       f_0(\bx_{\bs{z}})(\bH^{1/2}\bs{z})\rb\,d\bs{z} \nonumber \\
     &\qquad -f_0(\bx)-\frac{1}{2}\mu_2(K)\tr\{\bs{H}\nabla^2
       f_0(\bx)\}, \nonumber
      \end{align}
      where $\bx_{\bs{z}}=\bx-s_{\bs{z}}\bH^{1/2}\bs{z}$ for some $s_{\bs{z}}\in
      (0,1)$ depending on $\bs{z}$.
      We used that,  by Assumption~\ref{assm:KA}, $\int_{\RR^d}K(\bs{z})(\bH^{1/2}\bs{z})^T\nabla
      f_0(\bx)\,d\bs{z}=0$.
      (Recall
      \begin{align*}
        \MoveEqLeft \bb{E}\ffnH(\bx)-f_0(\bx)-\frac{1}{2}\mu_2(K)\tr\{\bs{H}\nabla^2
        f_0(\bx)\}\\
     &=\int_{\RR^d}K_{\bH}(\bx-\bs{y})f_0(\bs{y})\,d\bs{y}-f_0(\bx)-\frac{1}{2}\mu_2(K)\tr\{\bs{H}\nabla^2
       f_0(\bx)\}.)
      \end{align*}
    \end{longform}
  \end{mylongform}
  \begin{mylongform}
    \begin{longform}
      \begin{align*}
        \frac{1}{2}\mu_2(K)\tr\{\bs{H}\nabla^2
        f_0(\bx)\}     =\int_{\RR^d}\frac{1}{2}K(\bs{z})(\bH^{1/2}\bs{z})^T\nabla^2
        f_0(\bx)(\bH^{1/2}\bs{z})\,d\bs{z}.
      \end{align*}
      We further have
      \begin{align*}
        \MoveEqLeft \bb{E}\ffnH(\bx)-f_0(\bx)-\frac{1}{2}\mu_2(K)\tr\{\bs{H}\nabla^2
        f_0(\bx)\}\\
     &=\inv{2}\int_{\RR^d}K(\bs{Z})(\bH^{1/2}\bs{z})^T\nabla^2
       f_0(\bx_{\bs{z}})(\bH^{1/2}\bs{z})\,d\bs{z}-\inv{2}\int_{\RR^d}K(\bs{Z})(\bH^{1/2}\bs{z})^T\nabla^2
       f_0(\bx)(\bH^{1/2}\bs{z})\,d\bs{z}.
      \end{align*}
    \end{longform}
  \end{mylongform}
  For fixed $\bx\in\beta_{\tau}$,
  by change of variable and a Taylor expansion, we  have
  \begin{align}
    \MoveEqLeft \bb{E}\ffnH(\bx)-f_0(\bx)-\frac{1}{2}\mu_2(K)\tr\{\bs{H}\nabla^2
    f_0(\bx)\} \label{eq:lemmaA6-Efhat-0} \\
 &\le \inv{2}\int_{\RR^d}K(\bs{z})(\bH^{1/2}\bs{z})^T\left|\nabla^2
   f_0(\bx_{\bs{z}})-\nabla^2
   f_0(\bx)\right|(\bH^{1/2}\bs{z})\,d\bs{z}, \nonumber
  \end{align}
  where $\bx_{\bs{z}}=\bx-s_{\bs{z}}\bH^{1/2}\bs{z}$ for some $s_{\bs{z}}\in
  (0,1)$ depending on $\bs{z}$.
  Now let $M(\bx,\bs{z}) =\max\lb\left|\nabla^2
    f_0(\bx_{\bs{z}})-\nabla^2 f_0(\bx)\right|\rb_{i,j}$ which also implicitly
  depends on $\bH$ and is uniformly bounded since $\nabla^2 f_0$ is uniformly
  bounded. Then \eqref{eq:lemmaA6-Efhat-0} is bounded above by
  $\inv{2}\tr\lp\bH\int_{\RR^d}M(\bx,\bs{z})K(\bs{Z})\bs{z}\bs{z}^T\,d\bs{z}\rp.$
  Then
  \begin{align}
    \MoveEqLeft    \int_{\beta_{\tau}}\frac{\bb{E}\ffnH(\bx)-f_0(\bx)-\inv{2}\mu_2(K)f_0(\bx)}{\|\nabla
    f_0(\bx)\|}\,d\mc{H}(\bx) \label{eq:lemmaA6-Efhat-1} \\
 &\le\int_{\beta_{\tau}}\frac{1}{\|\nabla
   f_0(\bx)\|}\inv{2}\tr\lp\bH\int_{\RR^d}M(\bx,\bs{z})K(\bs{z})\bs{z}\bs{z}^T\,d\bs{z}\rp\,d\mc{H}(\bx)
  \end{align}
  which equals
  \begin{equation*}
    \inv{2}\tr\lp\bH\int_{\beta_{\tau}}\frac{1}{\|\nabla
      f_0(\bx)\|}\int_{\RR^d}M(\bx,\bs{z})K(\bs{z})\bs{z}\bs{z}^T\,d\bs{z}\,d\mc{H}(\bx)\rp.
  \end{equation*}
  Applying the Dominated Convergence theorem to both the outer integral and
  the inner integral yields
  \begin{align*}
    \int_{\beta_{\tau}}\frac{1}{\|\nabla
    f_0(\bx)\|}\int_{\RR^d}M(\bx,\bs{z})K(\bs{z})\bs{z}\bs{z}^T\,d\bs{z}\,d\mc{H}(\bx)\rightarrow
    0,
  \end{align*}
  and thus
  \eqref{eq:lemmaA6-Efhat-1} equals
  \begin{align*}
    \int_{\beta_{\tau}}\frac{\bb{E}\ffnH(\bx)-f_0(\bx)-\inv{2}\mu_2(K)f_0(\bx)}{\|\nabla
    f_0(\bx)\|}\,d\mc{H}(\bx)
    &=  \int_{\beta_{\tau}}\frac{\bb{E}\ffnH - f_0 }{\|\nabla
      f_0 \|}\,d\mc{H} - V1(\bH) \\
    & = o\lb\tr(\bH)\rb.
  \end{align*}
  With the same argument, we can show
  \begin{align*}
    \inv{f_{\tau,0}}\int_{\mc{L}_{\tau}}\bb{E}\ffnH(\bx)-f_0(\bx)\,d\bx-V_2(\bH)=o\lb\tr(\bH)\rb.
  \end{align*}
  In order to finish the proof
  \begin{mylongform}
    \begin{longform}
      that
      \begin{align*}
        \bb{E} \fftaun&=f_{\tau,0} +
                        w_0 \left\{V_1(\bH)+V_2(\bH)\right\}+o\{\tr(\bs{H})\},
      \end{align*}
    \end{longform}
  \end{mylongform}
  it is sufficient to show that for any $\eta>0$,
  \begin{equation}
    \label{eq:28}
    \EE \left|\fftaun - f_{\tau,0} -
      w_0
      \left\{V_1(\bH)+V_2(\bH)\right\}\right|\one_{\{\|\ffnH-f_0\|_{\infty}+\|\nabla
      \ffnH-\nabla f_0\|_{\infty}>\eta\}}
  \end{equation}
  is $o\{\tr(\bs{H})\}$. It can be show that $\fftaun=O(1)$.
  \begin{mylongform}
    \begin{longform}
      It follows from Lemma~\ref{lem:hdr-step1} that
      there exists $\varepsilon>0$ small enough, we have
      \begin{align*}
        \EE\lp\fftaun^2\rp&=\EE\lp\fftaun^2\one_{\lb \|f_0-\ffnH\|_{\infty}\le
                            \epsilon\rb}\rp+\EE\lp\fftaun^2\one_{\lb
                            \|f_0-\ffnH\|_{\infty}> \epsilon\rb}\rp\\
                          &\le f_{\tau,0}+C_1\varepsilon+\EE\lp\fftaun^2\one_{\lb
                            \|f_0-\ffnH\|_{\infty}> \epsilon\rb}\rp.
      \end{align*}
      For the last term on the last line,
      \begin{align*}
        \EE\lp\fftaun^2\one_{\lb
        \|f_0-\ffnH\|_{\infty}> \epsilon\rb}\rp&\le\EE\lp\|\ffnH\|^2_{\infty}\one_{\lb
                                                 \|f_0-\ffnH\|_{\infty}> \epsilon\rb}\rp\\
                                               &\le |\bH|^{-1}\|K\|^2_{\infty}P\lp\|f_0-\ffnH\|_{\infty}> \epsilon\rp=O(1),
      \end{align*}
      where the last line follows from Theorem~\ref{thm:KDE-sup-rate}.
    \end{longform}
  \end{mylongform}
  And we  have
  \begin{align*}
    &P(\|\ffnH-f_0\|_{\infty}+\|\nabla
      \ffnH-\nabla f_0\|_{\infty}>\eta)\\&\quad\le P(\|\ffnH-f_0\|_{\infty}>\eta/2)+P(\|\nabla
                                           \ffnH-\nabla f_0\|_{\infty}>\eta/2)
                                           =o(n^{-1}).
  \end{align*}
  Then by the Cauchy-Schwarz inequality \eqref{eq:28} is $o\{\tr(\bH)\}$.
\end{proof}

\begin{proof}[Proof of Lemma~\ref{lem:ftaun-var-exp}]
  First, we show
  \begin{equation}
    \label{eq:6}
    \Var\left\{\int_{\mc{L}_{\tau}}\ffnH(\bx)-f_0(\bx)\,d\bx\right\}=O(n^{-1}).
  \end{equation}
  We write the left side of \eqref{eq:6} as
  \begin{equation}
    \label{eq:nvarint}
    \begin{split}
      n^{-1} \Var\left\{\int_{\mc{L}_{\tau}}K_{\bH}(\bx-\bX_i)\,d\bx\right\}
      &= n^{-1} \bb{E} \left\{\int_{\mc{L}_{\tau}}K_{\bH}(\bx-\bX_i)\,d\bx\right\}^2
      \\
      & \qquad - n^{-1} \left[
        \bb{E}\left\{ \int_{\mc{L}_{\tau}}K_{\bH}(\bx-\bX_i)\,d\bx\right\}\right]^2.
    \end{split}
  \end{equation}
  We first consider the first term on the right side of \eqref{eq:nvarint}.
  \begin{mylongform}
    \begin{longform}
      (Which is
      \begin{align*}
        \int_{\mc{L}_{\tau}} K_{\bH}(\bx-\bX_i)\,d\bx
        =\int_{\mc{L}_{\tau}}|\bH|^{-1/2}K(\bH^{-1/2}(\bx-\bX_i))\,d\bx.)
      \end{align*}
    \end{longform}
  \end{mylongform}
  If $\bs{y}$ is an interior point of ${\mc{L}_{\tau}}$, there exists $r>0$ such that
  $B(\bs{y},r)\subset {\mc{L}_{\tau}}$. Then we have
  \begin{align*}
    \int_{\mc{L}_{\tau}}|\bH|^{-1/2}K(\bH^{-1/2}(\bx-\bs{y}))\,d\bx&\ge
                                                                     \int_{B(\bs{y},r)}|\bH|^{-1/2}K(\bH^{-1/2}(\bx-\bs{y}))\,d\bx\\
                                                                   &=\int \one_{\{\|\bH^{1/2}\bs{z}\|<r\}}K(\bs{z})\,d\bs{z},
  \end{align*}
  and $\one_{\{\|\bH^{1/2}\bs{z}\|<r\}}\rightarrow 1$ as $\bH \rightarrow 0$ for every $\bs{z}$; thus by the Dominated Convergence Theorem, $ \int_{\mc{L}_{\tau}}|\bH|^{-1/2}K(\bH^{-1/2}(\bx-\bs{y}))\,d\bx\rightarrow 1$ as $ \bH\rightarrow 0$.  Similarly, if $\bs{y}$ is an exterior point of $\{\bx|f_0(\bx)\ge f_{\tau,0}\}$, that is, there exists $r>0$ such that $B(\bs{y},r)\cap {\mc{L}_{\tau}}=\emptyset$. Then
  \begin{align*}
    \int_{\mc{L}_{\tau}} K_{\bs{H}} (\bx - \bs{y}) \,d\bx
    &\le 1-
      \int_{B(\bs{y},r)}|\bH|^{-1/2}K(\bH^{-1/2}(\bx-\bs{y}))\,d\bx\\
    &=1-\int \one_{\{\|\bH^{1/2}\bs{z}\|<r\}}K(\bs{z})\,
      d\bs{z}\rightarrow 0
  \end{align*}
  as    $ \bH\rightarrow 0$. And by Assumption~\ref{assm:DA-hdr},
  $P(f_0(\bx)=f_{\tau,0})=0$. So we have that almost surely $  \lp \int_{\mc{L}_{\tau}} K_{\bH}(\bx-\bX_i)\,d\bx \rp^p \rightarrow\one_{\mc{L}_{\tau}}$, as $n\rightarrow \infty$, for $p=1,2$.
  \begin{mylongform}
    \begin{longform}
      Now it is also
      obvious that both  random integrals on the right side of (\ref{eq:nvarint})
      are bounded by 1 and
      converge to $\one_{\mc{L}_{\tau}}$ a.s.
    \end{longform}
  \end{mylongform}
  Applying the Dominated Convergence Theorem to the two expectations on the
  right of (\ref{eq:nvarint}) yields
  \begin{align*}
    n \Var\left\{\int_{\mc{L}_{\tau}}\ffnH(\bx)-f_0(\bx)\,d\bx\right\}\rightarrow
    \PP(f_0(\bX_i)\ge f_{\tau,0})(1-     \PP(f_0(\bX_i)\ge f_{\tau,0})),
  \end{align*}
  as $n\rightarrow \infty$,
  which shows $\Var \left\{\int_{\mc{L}_{\tau}}\ffnH(\bx)
    -f_0(\bx)\,d\bx\right\}=O(n^{-1})$.

  Next, we show
  \begin{align*}
    \Var \lb \int_{\beta_\tau}\frac{\ffnH(\bx)-f_0(\bx)}{\|\nabla
    f_0(\bx)\|}\,d\mc{H}(\bx)\rb=o\lp\frac{1}{n|\bH|^{1/2}}\rp.
  \end{align*}
  The left side of the previous display equals
  \begin{align*}
    \MoveEqLeft
    n^{-1} \Var  \int_{\beta_\tau}
    \frac{ K_{\bH}(\bx-\bX_i)}{\|\nabla
    f_0(\bx)\|}\,d\mc{H}(\bx) \\
 & =\frac{1}{n}\bb{E}\left[ \lb \int_{\beta_\tau}\frac{K_{\bH}(\bx-\bX_i)}{
   \|\nabla
   f_0(\bx)\|}\,d\mc{H}(\bx)\rb^2\right]
         -\frac{1}{n}\left[\bb{E} \lb \int_{\beta_\tau}\frac{
         K_{\bH} (\bx-\bX_i)}{
         \|\nabla f_0(\bx)\|}\,d\mc{H}(\bx)\rb\right]^2,
  \end{align*}
  and
  $\lb \int_{\beta_\tau}\frac{|\bH|^{-1/2}K(\bH^{1/2}(\bx-\bX_i))}{\|\nabla
    f_0(\bx)\|}\,d\mc{H}(\bx)\rb^2$ can be written as
  \begin{align*}
    \int_{\beta_\tau}\frac{|\bH|^{-1/2}K(\bH^{1/2}(\bx-\bX_i))}{\|\nabla
    f_0(\bx)\|}\,d\mc{H}(\bx)\int_{\beta_\tau}\frac{|\bH|^{-1/2}K(\bH^{1/2}(\bs{y}-\bX_i))}{\|\nabla
    f_0(\bs{y})\|}\,d\mc{H}(\bs{y}).
  \end{align*}
  By taking the expectation over $\bX_i$ and reordering the integrals by
  Tonelli's theorem, we can then see that
  $n^{-1} \EE
  \lb \int_{\beta_\tau}\frac{ K_{\bH}(\bx-\bX_i)}{\|\nabla
    f_0(\bx)\|}\,d\mc{H}(\bx)\rb^2 $ equals
  \begin{equation}
    \label{eq:hdr-step4-secondmoment}
    \begin{split}
      \inv{n |\bH|}
      \int_{\beta_\tau}\int_{\beta_\tau}
      & \frac{1}{
        \|\nabla
        f_0(\bx)\|}\frac{1}{\|\nabla f_0(\bs{y})\|} \times \\
      & \;
      \int_{\RR^d}
      K(\bH^{-1/2}(\bx-a)) K(\bH^{-1/2}(\bs{y}-a))f_0(\bs{a})
      \,d\bs{a}
      \,
      d\mc{H}(\bs{x})d\mc{H}(\bs{y}).
    \end{split}
  \end{equation}
  And
  \begin{align*}
    \MoveEqLeft \int_{\RR^d} K( \bH^{-1/2} (\bx - \bs a)) K( \bH^{-1/2} (\bs y - \bs a))
    f_0(\bs a) d \bs a \\
 & = \int_{\RR^d} K(\bs z) K( \bs z + \bH^{-1/2} ( \bs y - \bx) )
   f_0( \bx - \bH^{1/2} \bs z) | \bH|^{1/2} d \bs z
  \end{align*}
  by the change of variables $\bs z =  \bH^{-1/2 } ( \bx - \bs a)$.
  And by first-order Taylor expansion, the previous display equals
  \begin{equation*}
    |\bH|^{1/2}
    \int_{\RR^d} K(\bs z) K( \bs z + \bH^{-1/2} ( \bs y - \bx) )
    \lb f_0(\bx) - \bH^{1/2} \bs z \grad f_0(\bx-s\bH^{1/2}\bs{z}) \rb d \bs z
  \end{equation*}
  where $s\in[0,1]$ depends on $\bs{z}$.  Since by
  Assumption~\ref{assm:DA-hdr}, $\nabla f_0(\bx)$ is bounded,
  \begin{mylongform}
    \begin{longform}
      (we
      can get rid of the  term $  |\bH|^{1/2}
      \int_{\RR^d} K(\bs z) K( \bs z + \bH^{-1/2} ( \bs y - \bx) ) \bH^{1/2}
      \bs z \grad f_0(\bx-s\bH^{1/2}\bs{z}) d \bs z$),
    \end{longform}
  \end{mylongform}
  we can express \eqref{eq:hdr-step4-secondmoment} as
  \begin{equation}
    \label{eq:7}
    \begin{split}
      \frac{\fftau}{n | \bH|^{1/2}}
      \int_{\beta_\tau}  \int_{\beta_{\tau}} &
      \inv{ \| \grad f_0(\bx) \| \| \grad f_0(\bs y) \|} \times \\
      & \int_{\RR^d} K(\bs z) K( \bs z + \bH^{-1/2} (\bs y - \bx)) d \bs z
      d\cH (\bs y) d\cH(\bs x)+o\lp n^{-1}|\bH|^{-1/2}\rp.
    \end{split}
  \end{equation}
  Note that if $\bs{x}\ne \bs{y}$, then
  \begin{align*}
    \int_{\RR^d} K(\bs z) K( \bs z + \bH^{-1/2} (\bs y - \bx)) d \bs
    z\rightarrow 0,
  \end{align*}
  as $\bH\rightarrow 0$ by the Dominated Convergence Theorem. For  fixed
  $\bs{x}= \bs{y}$, we have
  $\int_{\RR^d} K(\bs z) K( \bs z + \bH^{-1/2} (\bs y - \bx)) d \bs
  z=R(K),$
  so
  \begin{align*}
    \int_{\RR^d} K(\bs z) K( \bs z + \bH^{-1/2} (\bs y - \bx)) d \bs
    z\rightarrow R(K)\one_{\{\bs{x}=\bs{y}\}},
  \end{align*}
  as $\bH\rightarrow 0$. Then applying the Dominated Convergence Theorem
  shows that  the first summand in \eqref{eq:7} converges to
  \begin{equation*}
    \inv{ n | \bH|^{1/2}}
    \int_{\beta_\tau} f_0(\bx) \int_{\beta_{\tau}}
    \inv{ \| \grad f_0(\bx) \| \| \grad f_0(\bs y) \|}
    R(K)\one_{\{\bs{x}=\bs{y}\}}
    d\cH (\bs y) d\cH(\bs x)=0.
  \end{equation*}
  So we proved
  \begin{equation}
    \label{eq:33}
    \Var \lb \int_{\beta_\tau}
    \frac{\ffnH(\bx)-f_0(\bx)}{ \| \nabla f_0(\bx)\|}
    \, d\mc{H}(\bx)\rb=o\lp\frac{1}{n|\bH|^{1/2}}\rp.
  \end{equation}
  To complete the proof, it remains to show that for any $\eta>0$,
  \begin{align*}
    \bb{E}\lb\fftaun-\bb{E}(\fftaun)\rb^2\one_{\{\|\ffnH-f_0\|_{\infty}+\|\nabla
    \ffnH-\nabla f_0\|_{\infty}>\eta\}}=o\lp\frac{1}{n|\bH|^{1/2}}\rp,
  \end{align*}
  which follows the same steps we used at the end of the proof of Lemma \ref{lem:ftaun-bias-exp}.  The proof is then complete by \eqref{eq:fftaun-ftau}.
\end{proof}

\medskip

\noindent The reader may be surprised by the conclusion of
\eqref{eq:33}, since we have $\Var \ffnH(\bx) = O( n^{-1} | \bH|^{-1/2})$; for intuition, it may help to recall that $\int_{\RR^d} \ffnH(\bx) d \bx = 1$, so has variance $0$.

\medskip

\begin{proof}[Proof of Lemma~\ref{lem:hdr-step5}]
  Note by Theorem~\ref{cor:convergerate}, we have that
  $ \bb{E}\nabla\ffnH(\bx) $ converges to $\nabla \ffnH(\bx)$ uniformly
  in  $\bx\in\RR^d$.  By \citet[Theorem A.5.1]{durrett2010probability},  we
  have $\nabla\bb{E}\ffnH(\bx)=\bb{E}\nabla \ffnH(\bx)$, thus we
  also have $\nabla\bb{E}\ffnH(\bx)$ also
  converges to $\nabla f_0(\bx)$ uniformly in $\bx$.

  Now,   we show $\bb{E}\ffnH(\bx+tu_{\bx})$ is strictly monotone for
  $t\in[-\delta_n,\delta_n]$ when $n$ is sufficiently. From our
  assumption, $\nabla f_0$ is Lipschitz. So when $n$ large enough and
  $\delta_n$ small enough, for each $t\in [-\delta_n,\delta_n]$ there
  exists $\epsilon_t$ such that $\nabla \bb{E}\ffnH(\bx+tu_{\bx})=\nabla
  f_0(\bx)+\epsilon_t$ and $\|\epsilon_t\|<\frac{l}{2}$, where
  $l=\inf_{\bx\in\beta_{\tau}}\|\nabla f_0(\bx)\|$ and we know $l>0$
  from Assumption \ref{assm:DA-hdr}. Then
  \begin{align*}
    \frac{d\bb{E}\ffnH(\bx+tu_{\bx})}{dt}&=\nabla\bb{E}\ffnH(\bx+tu_{\bx})u_{\bx}\\
                                         &=(\nabla
                                           f_0(\bx)+\epsilon_t)\frac{- \nabla f_0(\bx)}{\|\nabla
                                           f_0(\bx)\|}=-\|\nabla f_0(\bx)\|-\frac{ \nabla f_0(\bx)'\epsilon_t}{\|\nabla f_0(\bx)\|}<-\frac{l}{2},
  \end{align*}
  for all $t\in [-\delta_n,\delta_n]$ by the Cauchy-Schwarz inequality. Moreover, from Lemma \ref{lem:ftaun-bias-exp} we
  have
  \begin{align*}
    \bb{E} \fftaun&=f_{\tau,0}+\left\{\int_{\beta_\tau}\inv{\|\nabla
                    f_0\|}\,d\mc{H}\right\}^{-1}\left\{\int_{\beta_\tau}\frac{\mu_2(K)\tr\lp\bs{H}\nabla^2
                    f_0\rp}{2\|\nabla f_0\|}\,d\mc{H}\right.\nonumber\\
                  &\qquad\left.+\int_{\mc{L}_{\tau}}\frac{1}{2}\mu_2(K)\tr\lp\bs{H}\nabla^2
                    f_0\rp\,d\lambda\right\}+o\{\tr(\bs{H})\},
  \end{align*}
  and we also know
  \begin{align*}
    \MoveEqLeft  \bb{E}\ffnH\lp\bx+\frac{t}{\sqrt{n|\bH|^{1/2}}}u_
    {\bx}\rp\\&=f_0\lp\bx+\frac{t}{\sqrt{n|\bH|^{1/2}}}u_
                {\bx}\rp\\
 &\qquad+\inv{2}\int\bs{z}'\bH^{1/2}\nabla^2f_0\lp\bx+\frac{t}{\sqrt{n|\bH|^{1/2}}}u_
   {\bx}-s_{\bs{z}}\bH^{1/2}\bs{z}\rp\bH^{1/2}\bs{z}K(\bs{z})\,d\bs{z}\\
 &=f_0(\bx)+\nabla f_0\lp\bx+\frac{w_{\bx} t}{\sqrt{n|\bH|^{1/2}}}u_
   {\bx}\rp'\frac{t}{\sqrt{n|\bH|^{1/2}}}u_
   {\bx}\\
 &\qquad+\inv{2}\int\bs{z}'\bH^{1/2}\nabla^2 f_0\lp\bx+\frac{t}{\sqrt{n|\bH|^{1/2}}}u_
   {\bx}-s_{\bs{z}}\bH^{1/2}\bs{z}\rp\bH^{1/2}\bs{z}K(\bs{z})\,d\bs{z},
  \end{align*}
  and $\inv{2}\int\bs{z}^T\bH^{1/2}\nabla^2f_0\lp\bx+\frac{t}{\sqrt{n|\bH|^{1/2}}}u_
  {\bx}-s_{\bs{z}}\bH^{1/2}\bs{z}\rp\bH^{1/2}\bs{z}K(\bs{z})\,d\bs{z}$ is $O(\tr(\bH))$
  uniformly in $\bx$. Then
  \begin{align}
    \label{eq:tx-exp}
    \MoveEqLeft\frac{t_{\bx}^{\ast}}{\sqrt{n|\bH|^{1/2}}}\nabla f_0\lp\bx+\frac{w_{\bx}t}{\sqrt{n|\bH|^{1/2}}}u_
    {\bx}\rp'u_
    {\bx}\nonumber\\
 &=\left[w_0\left\{\int_{\beta_\tau}\frac{D_1(\bx,\bH)}{\|\nabla f_0\|}\,d\mc{H}+\int_{\mc{L}_{\tau}}D_1(\bx,\bH)\,d\lambda\right\}\right.\nonumber\\
 &\qquad\left.-\inv{2}\int\bs{z}^T\bH^{1/2}\nabla^2f_0\lp\bx+\frac{t}{\sqrt{n|\bH|^{1/2}}}u_
   {\bx}-s_{\bs{z}}\bH^{1/2}\bs{z}\rp\bH^{1/2}\bs{z}K(\bs{z})\,d\bs{z}\right]\nonumber\\
 &\quad(1+o(1)).
  \end{align}
  Since $\nabla f_0$ is Lipschitz, when $n$ is large enough $\nabla f_0\lp\bx+\frac{t}{\sqrt{n|\bH|^{1/2}}}u_
  {\bx}\rp'u_
  {\bx}<-\frac{l}{2}$ for all $\bx\in \beta_{\tau}$ and all $t\in\left[-\sqrt{n|\bH|^{1/2}}\delta_n,\sqrt{n|\bH|^{1/2}}\delta_n\right]$. To prove the last line of the lemma,
  since
  \begin{align*}
    \frac{d\bb{E}\ffnH(\bx+tu_{\bx})}{dt}<-\frac{l}{2},
  \end{align*}
  for all $t\in[-\delta_n,\delta_n]$ and all $\bx\in\beta_{\tau}$,
  \begin{align*}
    \frac{d\bb{E}\ffnH\lp\bx+\frac{t}{\sqrt{n|\bH|^{1/2}}}u_{\bx}\rp}{dt}
    \le- \frac{l}{2\sqrt{n|\bH|^{1/2}}},
  \end{align*}
  for all $\bx\in\beta_{\tau}$ and all
  $t\in[-\sqrt{n|\bH|^{1/2}}\delta_n,\sqrt{n|\bH|^{1/2}}\delta_n]$. Then
  by first order Taylor expansion, it is easy to get when $t\in
  I_{\bx}^n$,
  \begin{align*}
    \left|\bb{E}\lb\ffnH\lp\bx+\frac{t}{\sqrt{n|\bH|^{1/2}}}u_
    {\bx}\rp-\fftaun\rb\right|\ge
    \frac{l}{2\sqrt{n|\bH|^{1/2}}}\left|t-t_{\bx}^{\ast}\right|,
  \end{align*}
  when $n$ is large enough.

  And then when $t\le 0$,
  \begin{align}
    \label{eq:hdr-step5-probbound}
    \begin{split}
      \MoveEqLeft\left|P\lb\ffnH\lp\bx^t\rp<\fftaun\rb-\one_{\{t>0\}}\right|\\&=  P\lb\ffnH\lp\bx^t\rp<\fftaun\rb\\
      &\le P\lb\fftaun-\ffnH\lp\bx^t\rp   +  \bb{E}\lp\ffnH\lp\bx^t\rp-\fftaun\rp\ge \frac{l}{2\sqrt{n|\bH|^{1/2}}}|t-t_{\bx}^{\ast}|\rb\\
      &\le
      P\lb\left|\ffnH\lp\bx^t\rp-\bb{E}\ffnH\lp\bx^t\rp\right|\ge
      \frac{l}{4\sqrt{n|\bH|^{1/2}}}|t-t_{\bx}^{\ast}| \rb\\
      &\qquad+P\lb\left|\fftaun-\bb{E}\fftaun\right|\ge
      \frac{l}{4\sqrt{n|\bH|^{1/2}}}|t-t_{\bx}^{\ast}|\rb
    \end{split}
  \end{align}
  and we can show the same bound for $t>0$.
  Since
  \begin{align*}
    &\Var\ffnH(\bx)=\\&\quad n^{-1}\left[|\bH|^{-1/2}\int
                        K(\bs{z})f\lp\bx-\bH^{1/2}\bs{z}\rp\,d\bs{z}-\left\{\int K(\bs{z})f\lp\bx-\bH^{1/2}\bs{z}\rp\,d\bs{z}\right\}^2\right],
  \end{align*}
  $\Var\ffnH(\bx)$ is uniformly $O(n^{-1}|\bH|^{-1/2})$.
  And we know
  $\Var\fftaun$ is also $o(n^{-1}|\bH|^{-1/2})$ from Lemma \ref{lem:ftaun-var-exp}. Then
  there exists $C_2>0$ such that \eqref{eq:hdr-step5-probbound} can be
  further bounded as
  \begin{align*}
    \MoveEqLeft\left|P\lb\ffnH\lp\bx^t\rp<\fftaun\rb-\one_{\{t>0\}}\right|\\
 &\le P\lb\left|\frac{\ffnH\lp\bx^t\rp-\bb{E}\ffnH\lp\bx^t\rp}{\Var\ffnH\lp\bx^t\rp}\right|\ge
   C_2|t-t_{\bx}^{\ast}|
   \rb\\&\qquad+P\lb\left|\frac{\fftaun-\bb{E}\fftaun}{\Var\fftaun}\right|\ge
          C_2|t-t_{\bx}^{\ast}|\rb\\
 &\le \frac{2C_2}{(t-t_{\bx}^{\ast})^2}
  \end{align*}
  for all $\bx\in\beta_{\tau}, t\in \bigcup_{\bx\in
    \beta_{\tau}}I_{\bx}^n$ by
  Chebyshev inequality. And note that $(t-t_{\bx}^{\ast})^2\ge
  t_n^2$ for all $t\in \bigcup_{\bx\in
    \beta_{\tau}}I_{\bx}^n$. So
  $\one_{I_{\bx}^n}\cdot|P\{\ffnH(\bx^t)<\fftaun\}-\one_{\{t>0\}}|$ converges to 0 uniformly
  in $t$ and is dominated by $\max\{1/(t-t_{\bx}^{\ast})^2,1\}$ which is
  a integrable function over $\RR$.   Then by Dominate Convergence
  Theorem, we have
  \begin{align*}
    \int_{I_{\bx}^n}\left|P\lb\ffnH\lp\bx^t\rp<\fftaun\rb-\one_{\{t>0\}}\right|\,dt\rightarrow 0,
  \end{align*}
  as $n\to \infty$. Also note that $
  \int_{I_{\bx}^n}|P\{\ffnH(\bx^t)<\fftaun\}-\one_{\{t>0\}}|\,dt\le \int
  \max\{1/t^2,1\}\,dt$ for all $\bx\in
  \beta_{\tau}$. So we have
  \begin{align*}
    \int_{\beta_\tau}\int_{I_{\bx}^n}\left|P\lb\ffnH\lp\bx^t\rp<\fftaun\rb-\one_{\{t>0\}}\right|\,dtd\mc{H}(\bx)\to
    0,
  \end{align*}
  as $n\rightarrow \infty$.
\end{proof}
\begin{mylongform}
  \begin{longform}
    \begin{proof}[Proof of Lemma~\ref{lem:ftaun-cov-exp}]
      First
      note from (\ref{eq:fftaun-ftau}) that
      \begin{align*}
        \Cov(\ffnH(\bx^t),\fftaun)&\approx w_0\cdot\left[\Cov\left\{\ffnH(\bx^t),\int_{\beta_\tau}\frac{\ffnH}{\|\nabla
                                    f_0\|}\,d\mc{H}\right\}\right.\\
                                  &\qquad\left.+\Cov\left\{\ffnH(\bx^t),\inv{f_{\tau,0}}\int_{\mc{L}_{\tau}}\ffnH\,d\lambda\right\}\right].
      \end{align*}
      Since $\Var\ffnH(\bx)=O(n^{-1}|\bH|^{-1/2})$,  and by the proof of
      Lemma \ref{lem:ftaun-var-exp} $\Var\int_{\mc{L}_{\tau}}\ffnH\,d\lambda=O(n^{-1})$, it is easy to see that
      $\Cov\left\{\ffnH(\bx^t),\int_{\mc{L}_{\tau}}\ffnH\,d\lambda\right\}=o(n^{-1}|\bH|^{-1/2})$ by Cauchy-Schwarz
      inequality.
      For $\Cov\left\{\ffnH(\bx^t),\int_{\beta_\tau}\frac{\ffnH}{\|\nabla
          f_0\|}\,d\mc{H}\right\}$, since it is equal to
      \begin{align*}
        \MoveEqLeft\inv{n^2}\Cov\left\{\sum_{i=1}^nK_{\bH}\lp\bx^t-\bX_i\rp,\sum_{i=1}^n\int_{\beta_\tau}\frac{K_{\bH}(\bs{y}-\bX_i)}{\|\nabla
        f_0(\bs{y})\|}\,d\mc{H}(\bs{y})\right\}\\
     &=\inv{n}\Cov\left\{K_{\bH}\lp\bx^t-\bX_i\rp,\int_{\beta_\tau}\frac{K_{\bH}(\bs{y}-\bX_i)}{\|\nabla
       f_0(\bs{y})\|}\,d\mc{H}(\bs{y})\right\},
      \end{align*}
      and
      \begin{align*}
        \MoveEqLeft  \Cov\left\{K_{\bH}\lp\bx^t-\bX_i\rp,\int_{\beta_\tau}\frac{K_{\bH}(\bs{y}-\bX_i)}{\|\nabla
        f_0(\bs{y})\|}\,d\mc{H}(\bs{y})\right\}\\&=\bb{E}\left\{K_{\bH}\lp\bx^t-\bX_i\rp\int_{\beta_\tau}\frac{K_{\bH}(\bs{y}-\bX_i)}{\|\nabla
                                                   f_0(\bs{y})\|}\,d\mc{H}(\bs{y})\right\}
        \\&\qquad-\bb{E}K_{\bH}\lp bx^t-\bX_i\rp\bb{E}\int_{\beta_\tau}\frac{K_{\bH}(\bs{y}-\bX_i)}{\|\nabla
            f_0(\bs{y})\|}\,d\mc{H}(\bs{y}),
      \end{align*}
      it is easy to see that both $\bb{E}K_{\bH}\lp\bx^t-\bX_i\rp$ and
      $\bb{E}\int_{\beta_\tau}\frac{K_{\bH}(\bs{y}-\bX_i)}{\|\nabla f_0(\bs{y})\|}\,d\mc{H}(\bs{y})$
      converge to constants, so
      $\inv{n}\bb{E}K_{\bH}\lp\bx^t-\bX_i\rp\bb{E}\int_{\beta_\tau}\frac{K_{\bH}(\bs{y}-\bX_i)}{\|\nabla
        f_0(\bs{y})\|}\,d\mc{H}(\bs{y})=O(n^{-1})=o(n^{-1}|\bH|^{-1/2})$. For the
      expectation of the product term,
      \begin{align*}
        \MoveEqLeft  \bb{E}\left\{K_{\bH}\lp\bx^t-\bX_i\rp\int_{\beta_\tau}\frac{K_{\bH}(\bs{y}-\bX_i)}{\|\nabla
        f_0(\bs{y})\|}\,d\mc{H}(\bs{y})\right\}\\
     &
       =\int_{\beta_\tau}\frac{\bb{E}K_{\bH}\lp\bx^t-\bX_i\rp K_{\bH}(\bs{y}-\bX_i)}{\|\nabla
       f_0(\bs{y})\|}\,d\mc{H}(\bs{y})\\
     &=\int_{\beta_\tau}\inv{\|\nabla f_0\|}\int
       K_{\bH}\lp\bx^t-\bs{z}\rp K_{\bH}(\bs{y}-\bs{z})f_0(\bs{z})\,d\bs{z}\,d\mc{H}(\bs{y}).
      \end{align*}
      By change of variable,
      \begin{align*}
        \MoveEqLeft  \int
        K_{\bH}\lp\bx^t-\bs{z}\rp K_{\bH}(\bs{y}-\bs{z})f_0(\bs{z})\,d\bs{z}\\
     &=\inv{|\bH|^{1/2}}\int
       K\lp \bH^{-1/2}(\bx^t-\bs{y})+\bs{w}\rp K(\bs{w})f_0\lp\bs{y}-\bH^{1/2}\bs{w}\rp\,d\bs{w}.
      \end{align*}
      We know that $\|\bH^{-1/2}(\bx^t-\bs{y})\|\ge
      \|\bx^t-\bs{y}\|/\sqrt{\lambda_{\max}(\bH)}$.
      By Assumption~\ref{assm:DA}, $K$ is density
      function with compact support, so there exits $r>0$ such that $K(w)=0$
      for all $\|w\|>r$. Then we have
      \begin{align*}
        \MoveEqLeft \int
        K\lp \bH^{-1/2}(\bx^t-\bs{y})+\bs{w}\rp K(\bs{w})f_0\lp\bs{y}-\bH^{1/2}\bs{w}\rp\,d\bs{w}\\
     &=\int\one_{\lb\|\bx^t-\bs{y}\|/\sqrt{\lambda_{\max}(\bH)}\le 2r\rb}
       K\lp \bH^{-1/2}(\bx^t-\bs{y})+\bs{w}\rp K(\bs{w})f_0\lp\bs{y}-\bH^{1/2}\bs{w}\rp\,d\bs{w}\\
     &\le \one_{\lb\|\bx^t-\bs{y}\|/\sqrt{\lambda_{\max}(\bH)}\le
       2r\rb}\|K\|_{\infty}\|f_0\|_{\infty}\\
     &\le \one_{\lb\|\bx^t-\bs{y}\| \le
       2r\sqrt{\lambda_{\max}(\bH)}\rb}\|K\|_{\infty}\|f_0\|_{\infty}.
      \end{align*}
      We know for any
      $t\in\left[-\sqrt{n|\bH|^{1/2}}\delta_n,\sqrt{n|\bH|^{1/2}}\delta_n\right]$,
      $\|\bx^t-\bs{y}\|\ge \|\text{proj}_{\bx^t}(\bs{y})-\bs{y}\|$, where
      $\text{proj}_{\bx^t}(\bs{y})$ denote the projection of $\bs{y}$ on
      the normal vector of $\beta_{\tau}$ at $\bs{x}$. Now we can further have
      \begin{align*}
        \MoveEqLeft  \inv{n}\bb{E}\left\{K_{\bH}(\bx^t-\bX_i)\int_{\beta_\tau}\frac{K_{\bH}(\bs{y}-\bX_i)}{\|\nabla
        f_0(\bs{y})\|}\,d\mc{H}(\bs{y})\right\}\\
       &\le\frac{\|K\|_{\infty}\|f_0\|_{\infty}}{n|\bH|^{1/2}}
         \int_{\beta_{\tau}}\one_{\{\|\text{proj}_{\bx^t}(\bs{y})-\bs{y}\|\le 2r\sqrt{\lambda_{\max}(\bH)}\}}\,d\cH(\bs{y}),
      \end{align*}
      and by Dominate convergence
      $\int_{\beta_{\tau}}\one_{\{\|\text{proj}_{\bx^t}(\bs{y})-\bs{y}\|\le
        2r\sqrt{\lambda_1}\}}\,d\cH(\bs{y})\rightarrow
      0$ \KR{(uniformly in $\bx$?)}, as $\bH\rightarrow 0$. So $ \inv{n}\bb{E}\left\{K_{\bH}(\bx^t-\bX_i)\int_{\beta_\tau}\frac{K_{\bH}(\bs{y}-\bX_i)}{\|\nabla
          f_0(\bs{y})\|}\,d\mc{H}(\bs{y})\right\}=o(n|\bH|^{1/2})$ uniformly
      in $t$ \KR{(and also uniformly in $\bx$?)}.

      \KR{
        we can see from the above integral, if $\bx\ne \bs{y}$, since $K$ and $f_0$ are both bounded,  $\int_{\RR^{d}}
        K(\bH^{-1/2}(\bx^t-\bs{y})-\bs{w})K(\bs{w})f(\bs{y}-\bH^{1/2}\bs{w})\,d\bs{w}\rightarrow
        0$ as $\bH\rightarrow 0$, by Dominated Convergence theorem. So we
        have
        \begin{align*}
          \limsup_{n\rightarrow n} \int_{\RR^{d}}
          K(\bH^{-1/2}(\bx^t-\bs{y})-\bs{w})K(\bs{w})f(\bs{y}-\bH^{1/2}\bs{w})\,d\bs{w}= \|K\|_{\infty}^2\|f_0\|_{\infty}\one_{\{\bx=\bs{y}\}},
        \end{align*}
        and
        \begin{align*}
          \int_{\beta_{\tau}}\inv{\|\nabla f_0\|}\int
          K_{\bH}(\bx^t-\bs{z})K_{\bH}(\bs{y}-\bs{z})f_0(\bs{z})\,d\bs{z}\,d\mc{H}(\bs{y})\rightarrow 0,
        \end{align*}
        as $n\rightarrow \infty$, which indicate $\inv{n}\bb{E}\left\{K_{\bH}(\bx^t-\bX_i)\int_{\beta_\tau}\frac{K_{\bH}(\bs{y}-\bX_i)}{\|\nabla
            f_0(\bs{y})\|}\,d\mc{H}(\bs{y})\right\}=o(n^{-1}|\bH|^{-1/2})$. }

      To
      complete the proof, we just need to show that
      \begin{align*}
        \bb{E}\left[\lb\ffnH(\bx^t)-\bb{E}\ffnH(\bx^t)\rb\lb
        \fftaun-\bb{E}(\fftaun)\rb\one_{\{\|\ffnH-f_0\|_{\infty}+\|\nabla
        \ffnH-\nabla f_0\|>\eta\}}\right]
      \end{align*}
      is $o(n^{-1}|\bH|^{-1/2})$ for any $\eta>0$ and the proof is the same
      as in Lemma \ref{lem:ftaun-bias-exp} and \ref{lem:ftaun-var-exp}.
    \end{proof}
  \end{longform}
\end{mylongform}

\bibliographystyle{abbrvnat}
\bibliography{mybib}

\end{document}